\documentclass[12pt]{article}
\usepackage{amsmath}
\usepackage{graphicx,psfrag,epsf}
\usepackage{enumerate}
\usepackage{natbib}
\usepackage{authblk}

\usepackage{amssymb}
\usepackage[colorlinks,citecolor=blue,urlcolor=blue]{hyperref}
\usepackage{mathrsfs}
\usepackage{amsfonts,anysize,amsmath,indentfirst,geometry,xcolor}
\usepackage{bm}
\usepackage{amsthm}
\usepackage{multirow}
\usepackage{chngpage}
\usepackage{array}
\usepackage{rotating}
\usepackage{setspace}
\usepackage{graphicx}
\usepackage{algorithm}
\usepackage{algorithmic}
\usepackage{epsfig}
\usepackage{amsbsy}
\usepackage{bbm}

\usepackage{booktabs}
\usepackage{hyperref}
\allowdisplaybreaks

\newtheorem{thm}{\bf Theorem}

\newtheorem{lem}{\underline{\bf Lemma}}
\newtheorem{corollary}{\bf Corollary}

\newtheorem{remark}{\underline{\bf Remark}}

\def\boxit#1{\vbox{\hrule\hbox{\vrule\kern6pt\vbox{\kern6pt#1\kern6pt}\kern6pt\vrule}\hrule}}
\def\bqe{\begin{eqnarray}}
	\def\eqe{\end{eqnarray}}

\newcommand{\blind}{0}

\begin{document}
\doublespacing

\def\spacingset#1{\renewcommand{\baselinestretch}%
	{#1}\small\normalsize} \spacingset{1}


\if0\blind
{
	\title{\bf Spatially Varying Coefficient Mallows Model Averaging}
	\author[1]{Yong Zhuang${^*}$}
    \author[1]{Jing Lv\footnote{Zhuang and Lv are cofirst authors and contribute equally.}} 
    \author[1]{Tingting Li\thanks{Corresponding author, email: tinalee@swu.edu.cn.}}
	\affil[1]{\small School of Mathematics and Statistics, Southwest University, Chongqing, China}
	
	\maketitle
} \fi


\begin{abstract}
	Model averaging, as an appealing ensemble technique, strategically integrates all valuable information from candidate models to construct fast and accurate prediction. Despite of having been widely practiced in many fields such as cross-sectional data, censored data and longitudinal data, its application to spatial data characterized by inherent spatial heterogeneity remains surprisingly limited. To mitigate risk of model misspecification and enhance the flexibility of prediction, we propose a combined estimator constructed by computing the weighted average of estimators derived from a set of spatially varying coefficient candidate models. Herein, the model weights are determined via a Mallows-type criterion, which dynamically calibrates the relative importance of individual candidate models in the ensemble. Theoretically, we establish desirable asymptotic properties under
    two practical scenarios. First, in the case where all candidate models are misspecified, the proposed model averaging estimator attains asymptotic optimality in the sense that it minimizes the squared error loss function asymptotically. Second, when the candidate model set encompasses at least one quasi-correct model, the weights assigned by the Mallows-type criterion asymptotically concentrate on the quasi-correct models, and the resulting model averaging estimator converges in probability to the true conditional mean.
	Both simulation studies and a real-world empirical example demonstrate that the proposed method generally outperforms alternative comparative approaches in terms of predictive accuracy and robustness.

\end{abstract}

\\

\spacingset{1.45}
\section{Introduction}\label{sec:intro}

Spatial data are ubiquitous in fields such as environmental, earth, and biological sciences. When building statistical models from observed data, the primary tasks are to explore the relationship between the response and covariates within the region of interest and to make predictions based on this relationship.
According to Tobler's well-known ``First Law of Geography", everything is related to everything else, but near things are more related than distant things\citep{Tobler}. This implies that relationships among variables in geographical data often exhibit stronger similarity at nearby locations, thereby demonstrating apparent spatial dependence. 
There are two primary approaches to modeling spatial dependence. The first category consists of global spatial autocorrelation models, with the spatial autoregressive (SAR) model\citep{Anselin1988} as the typical representative. These models assume that the spatial correlation structure is homogeneous across the entire space, meaning that the spatial interaction effect is consistent regardless of geographical location.
The second category corresponds to local spatial models, among which the geographically weighted regression model\citep{BFC98}—also referred to as the spatial varying-coefficient model (SVCM)—serves as the hallmark example. Distinct from global models, SVCM relaxes the assumption of spatial stationarity by allowing regression coefficients to vary spatially, thus capturing the non-stationary spatial relationships between variables across different locations.
Due to the flexibility and interpretability of the SVCM, it has received much attention in recent years.  For example, \cite{MWW18} studied the estimation and inference in the SVCM for data distributed over complex domains. \cite{KW21} extended the results of \cite{MWW18} to exponential dispersion families of distributions (i.e., binomial, Poisson, and negative binomial). To address the challenge arising from large-scale spatial datasets, \cite{YWW25} proposed a novel distributed heterogeneity learning algorithm, which has a simple, scalable, and communication-efficient implementation, nearly achieving linear speedup. Under the framework of quantile regression, \cite{KWW25} considered the problem of estimation and inference for the SVCM. These existing methodologies significantly enhance the model's ability to explain spatial heterogeneity.

An important objective in spatial data analysis is to leverage the rich information embedded in observed samples so as to optimize the accuracy of model-based predictions. Although the aforementioned literature have made significant contributions to statistical inference for the SVCM, they may exhibit poor prediction performance owing to their reliance on a single model. It is widely recognized that using a single model may lead to severe model misspecification and overlook valuable information embedded in other models, resulting in unsatisfactory prediction performance.
Model averaging, as a well-known ensemble technique, takes potential candidate models
into account and allocates a suitable weight to each for producing a combined estimator.
Therefore, it significantly mitigates the risk of model misspecification and typically yields more stable predictions than a single model.
Over the past few decades, we have witnessed a booming development of model averaging in many fields, see \cite{H2007,WZZ2010,HR2012,LZZZ19,ZWZZ19,ZZLC20,LLWL22,seng2022,LL2024,GZ2024,JLLC2025}. However, all the aforementioned works on model averaging focus on the independent and identically distributed data, single time series data, multi-categorical data or longitudinal data, and hence cannot be directly applied to make prediction for spatial data. So far, recent model averaging approaches to cope with spatial data include \cite{ZY2018,LZG2019,YDT2022,MFZW2025, JLL2025}, but these works are mostly built on SAR-type models, rather than the SVCM.
The SVCM has been developed to address spatial heterogeneity, a crucial consideration across numerous disciplines, including environmental research, epidemiology, agriculture, urban planning, climate science, and spatial econometrics \citep{WFMWQ2014,CWLZH18}.
Consequently, it is theoretically and practically significant to explore a new model averaging scheme by using the SVCM. So far, discussion on model averaging for the SVCM is rather limited.
In this work, we exploit a spatially varying coefficient model averaging procedure to achieve the goal of flexible prediction for the conditional mean of the response given covariates.

Compared with existing literature on model averaging and spatial modeling, this article offers three major contributions:
\begin{itemize} 
	
	\item In contrast to existing single-SVCM-based frameworks \citep{MWW18,KW21,KWW25}, the proposed model averaging method significantly alleviates the risk of model misspecification while achieving more flexible and accurate predictive outcomes. Empirical application for analysis of the soil total nitrogen from the Liudaogou watershed in Loess Plateau, Shaanxi Province, China in Section 5 show clearly that our model averaging proposal has lower prediction risk than competing methods.
    
    \item Compared to the established model averaging approaches for spatial data such as \cite{ZY2018,LZG2019,YDT2022,MFZW2025, JLL2025}, which were typically centered on diverse adjacency matrices and thus were essentially oriented toward addressing model structural uncertainty in global spatial regression frameworks, the proposed spatial varying coefficient model averaging method focuses on local spatial regression paradigms and is dedicated to resolving the uncertainties inherent in spatial heterogeneity analysis. Although optimal model averaging for varying coefficient models has been explored in several studies \citep{ZZLC20, LLWL22, GZ2024}, these works have focused on other types of data (e.g., cross-sectional data and multi-categorical data) rather than spatial data. To the best of our knowledge, there are no literature to investigate the issue of model averaging for the SVCM. Our work serves as a promising estimation framework for spatial data and fills this gap timely.
	
	\item We provide systematic theoretical guarantees for our proposal.
	When all candidate models are misspecified, we prove that the proposed model averaging estimator possesses asymptotic optimality, in the sense that its squared error loss asymptotically coincides with that of the theoretically best but infeasible model averaging estimator. The theoretical property guarantees that the proposed procedure exhibits strong predictive performance. When there is at least one quasi-correct model in the candidate model
	set, the sum of the weights determined by the Mallows-type criterion on the quasi-correct models converges to one in probability, which is named as the over-consistency property of model weights. Despite this property has been investigated by existing findings such as \citep{ZY2018,ZZLC20}, we further extend
	this property to spatial data.
	Further, the consistency of the proposed model averaging estimator is also established, that is, it converges to the true conditional mean with probability tending to one as the sample size increases to infinity.
	
	\end{itemize}

The remainder of this paper is organized as follows. In Section 2, we build a series of 
sub-SVCMs and introduce the Mallows-type criterion to determine the optimal model weights.
In Section 3, we give the asymptotic results. In Section 4, we carry out simulation studies to evaluate the practical performance of the proposed method and compare its performance with that of existing alternative methods. In Section 5,
a soil data set is analyzed to further illustrate the utility of our proposal. Discussions and concluding remarks are given in Section 6.
The poofs of the theorems and corollaries are provided in the Appendix.

\section{Methodology}
\subsection{Model Framework}

Consider spatial data observed at locations \( \mathbf{s}_1, \ldots, \mathbf{s}_n \in D \), denoted as \( \left\{ \left( \mathbf{X}(\mathbf{s}_i), y(\mathbf{s}_i) \right)\right.,\\ \left.i \in \{1, \ldots, n\} \right\} \), where \( D \subset \mathbb{R}^2 \) is a fixed region of interest, the covariate vector \( \mathbf{X}(\mathbf{s}_i) = \left( x_1(\mathbf{s}_i), \ldots, x_p(\mathbf{s}_i) \right)^{\mathrm{T}} \), and \( x_1(\mathbf{s}_i) = 1 \). The underlying spatial process is assumed to follow the model:
\begin{equation}\label{eq1}
	y(\mathbf{s}_i) = \sum_{k=1}^{p} x_k(\mathbf{s}_i) \beta_k(\mathbf{s}_i) + \epsilon(\mathbf{s}_i), \quad i=1,\ldots,n,
\end{equation}
where the regression coefficient \( \beta_k(\mathbf{s}_i) \) is the value of a smooth function \( \beta_k(\mathbf{s}) \) at location \( \mathbf{s}_i \), \( \bm{\beta}(\mathbf{s}) = (\beta_1(\mathbf{s}), \ldots, \beta_p(\mathbf{s}))^{\mathrm{T}} \) is the vector of unknown coefficient functions, and \( \{\epsilon(\mathbf{s}_i)\}_{i=1}^{n} \) are independent random errors satisfying \( E(\epsilon(\mathbf{s}_i) \mid \mathbf{X}(\mathbf{s}_i), \mathbf{s}_i) = 0 \) and \( E(\epsilon^2(\mathbf{s}_i) \mid \mathbf{X}(\mathbf{s}_i), \mathbf{s}_i) = \sigma^2 \). The spatial dependence in model \eqref{eq1} is typically captured by the spatially varying coefficients. The dimension \( p \) of the covariates can be divergent. For simplicity, the dependence on \( \mathbf{s} \) is omitted in the subsequent notation, letting \( y(\mathbf{s}_i) = y_i \), \( \mathbf{X}(\mathbf{s}_i) = \mathbf{X}_i = (x_{i1}, \ldots, x_{ip})^{\mathrm{T}} \), and \( \epsilon(\mathbf{s}_i) = \epsilon_i \).

Model \eqref{eq1} can be equivalently expressed in matrix form as
\begin{equation}\label{eq2}
	\mathbf{Y} = \bm{\mu} + \bm{\epsilon},
\end{equation}
where \(\mathbf{Y} = (y_1, \ldots, y_n)^{\mathrm{T}}\) is the response vector,  
\(\mathbf{X} = (\mathbf{X}_1, \ldots, \mathbf{X}_n)^{\mathrm{T}}\) is the covariate matrix,  
and \(\bm{\mu} = (\mu_1,\ldots,\mu_n)^{\mathrm{T}} = \bigl( \mathbf{X}_1^{\mathrm{T}} \bm{\beta}(\mathbf{s}_1), \ldots, 
\mathbf{X}_n^{\mathrm{T}} \bm{\beta}(\mathbf{s}_n) \bigr)^{\mathrm{T}}\) depends on 
the covariates \(\mathbf{X}\) and the spatial locations \(\mathbf{s}_i\).  
Denote \(\bm{\Psi} = (\mathbf{s}_1, \ldots, \mathbf{s}_n)^{\mathrm{T}}\).  
The error vector \(\bm{\epsilon} = (\epsilon_1, \ldots, \epsilon_n)^{\mathrm{T}}\) satisfies
\(E(\bm{\epsilon}\mid\mathbf{X}, \bm{\Psi}) = \mathbf{0}\) and 
\(E(\bm{\epsilon} \bm{\epsilon}^{\mathrm{T}}\mid\mathbf{X}, \bm{\Psi}) = \bm{\Omega} = \sigma^2 \mathbf{I}_n\), 
with \(\mathbf{I}_n\) being the \(n\)-dimensional identity matrix.

Due to the uncertainty in covariates, suppose there are $M$ candidate models, where $M$ is allowed to diverge to infinity as $n$ tends to infinity. Each model includes a distinct set of covariates, with the dimension of covariates denoted by $p_m$. The $m$th candidate model is given by
\begin{equation}\label{eq3}
	\mathbf{Y} = \bm{\mu}_{(m)} + \bm{\epsilon}_{(m)}, \quad m = 1,\dots,M,
\end{equation}
where $\bm{\mu}_{(m)} = \left(\mathbf{X}_{(m)1}^{\mathrm{T}}\bm{\beta}_{(m)}(\mathbf{s}_{1}), \ldots, \mathbf{X}_{(m)n}^{\mathrm{T}}\bm{\beta}_{(m)}(\mathbf{s}_{n})\right)^{\mathrm{T}}$, $\mathbf{X}_{(m)} = (\mathbf{X}_{(m)1}, \dots, \mathbf{X}_{(m)n})^{\mathrm{T}}$ is an $n \times p_m$ matrix whose columns are among those of $\mathbf{X}$, $\bm{\beta}_{(m)}(\mathbf{s}) = (\beta_{(m)1}(\mathbf{s}), \ldots, \beta_{(m)p_m}(\mathbf{s}))^{\mathrm{T}}$ is a $p_m \times 1$ vector of unknown coefficient functions, and $\bm{\epsilon}_{(m)} = (\epsilon_{(m)1}, \ldots, \epsilon_{(m)n})^{\mathrm{T}}$ is the random error vector for the  $m$th candidate model.

To estimate $\bm{\beta}_{(m)}(\mathbf{s})$, where $\mathbf{s} = (s_1, s_2)^{\mathrm{T}} \in D$, we follow the idea of \cite{BFC98} that samples with closer spatial locations have more similar regression coefficients. Using the local constant estimation method, $\bm{\beta}_{(m)}(\mathbf{s})$ is the solution to the following weighted least squares criterion:
\begin{equation}\label{eq4}
	\min_{\bm{\beta}_{(m)}} \sum_{i=1}^{n} \left(y_i - \mathbf{X}_{(m)i}^{\mathrm{T}}\bm{\beta}_{(m)}\right)^2 K_{h_m}(\|\mathbf{s}_i - \mathbf{s}\|_q),
\end{equation}
where $\bm{\beta}_{(m)} = (\beta_{(m)1}, \ldots, \beta_{(m)p_m})^{\mathrm{T}}$ is a constant vector, $K_{h_m}(\cdot) = K(\cdot/h_m)/h_m^2$, $K(\cdot)$ is a kernel function, $h_m$ is the corresponding bandwidth, and $\|\cdot\|_q$ with $q \geq 1$ denotes the $L_q$ norm used for measuring spatial distance. Denote $\|\cdot\|_2$ as $\|\cdot\|$. The estimator of $\bm{\beta}_{(m)}(\mathbf{s})$ is
\begin{equation}\label{eq5}
	\widehat{\bm{\beta}}_{(m)}(\mathbf{s}) = \left(\mathbf{X}_{(m)}^{\mathrm{T}} \mathbf{W}_{(m)\mathbf{s}} \mathbf{X}_{(m)}\right)^{-1} \mathbf{X}_{(m)}^{\mathrm{T}} \mathbf{W}_{(m)\mathbf{s}} \mathbf{Y},
\end{equation}
where $\mathbf{W}_{(m)\mathbf{s}} = \mathrm{diag}\{K_{h_m}(\|\mathbf{s}_1 - \mathbf{s}\|_q), \ldots, K_{h_m}(\|\mathbf{s}_n-\mathbf{s}\|_q)\}$, and $\mathbf{X}_{(m)} = \mathbf{X} \bm{\Pi}_{m}^{\mathrm{T}}$. Here, $\bm{\Pi}_{m}$ is a $p_m \times p$ selection matrix that maps the explanatory vector $\mathbf{X}_i = (x_{i1}, x_{i2},\ldots, x_{ip} )^{\mathrm{T}}$ to the subvector $\mathbf{X}_{(m)i}=\bm{\Pi}_{m}\mathbf{X}_i$. For nested models, $\bm{\Pi}_{m}$ is defined as $(\mathbf{I}_{p_m}, \bm{0}_{p_m \times (p-p_m)})$. 

Based on the above estimation method, we obtain the estimator of $\bm{\beta}_{(m)}(\mathbf{s}_i)$ under the $m$th model. The estimator of the conditional mean for the $i$th sample is
$$
\widehat{\mu}_{(m)i} = \mathbf{X}_{(m)i}^{\mathrm{T}} \left(\mathbf{X}_{(m)}^{\mathrm{T}} \mathbf{W}_{(m)\mathbf{s}_i} \mathbf{X}_{(m)}\right)^{-1} \mathbf{X}_{(m)}^{\mathrm{T}} \mathbf{W}_{(m)\mathbf{s}_i} \mathbf{Y}.
$$
Let $\widehat{\bm{\mu}}_{(m)} = (\widehat{\mu}_{(m)1}, \ldots, \widehat{\mu}_{(m)n})^{\mathrm{T}}$, then we have $\widehat{\bm{\mu}}_{(m)} = \mathbf{P}_{(m)} \mathbf{Y}$, where
\begin{equation}\label{eq6}
	\mathbf{P}_{(m)} =
	\left[
	\begin{array}{c}
		\mathbf{X}_{(m)1}^{\mathrm{T}} \left(\mathbf{X}_{(m)}^{\mathrm{T}} \mathbf{W}_{(m)\mathbf{s}_1} \mathbf{X}_{(m)}\right)^{-1} \mathbf{X}_{(m)}^{\mathrm{T}} \mathbf{W}_{(m)\mathbf{s}_1} \\
		\mathbf{X}_{(m)2}^{\mathrm{T}} \left(\mathbf{X}_{(m)}^{\mathrm{T}} \mathbf{W}_{(m)\mathbf{s}_2} \mathbf{X}_{(m)}\right)^{-1} \mathbf{X}_{(m)}^{\mathrm{T}} \mathbf{W}_{(m)\mathbf{s}_2} \\
		\cdots \\
		\mathbf{X}_{(m)n}^{\mathrm{T}} \left(\mathbf{X}_{(m)}^{\mathrm{T}} \mathbf{W}_{(m)\mathbf{s}_n} \mathbf{X}_{(m)}\right)^{-1} \mathbf{X}_{(m)}^{\mathrm{T}} \mathbf{W}_{(m)\mathbf{s}_n}
	\end{array}
	\right].
\end{equation}

How to select the optimal bandwidth $h_m$ for each candidate model is a critical issue that requires careful consideration. When strong theoretical prior information regarding the value of $h_m$ is unavailable, some form of automated data-driven selection of $h_m$ may be more appropriate. In both numerical simulations and empirical analysis, we select $h_m$ by minimizing a cross-validation (CV) criterion. The CV criterion used for selecting $h_m$ is defined as
$$\mathrm{CV}(h_m)=\sum_{i=1}^{n}(y_i-\widetilde{\mu}_i(h_m))^2,$$
where $\widetilde{\mu}_i(h_m)$ denotes the predicted value at $\mathbf{s}_i$ obtained from the remaining $n-1$ samples after excluding the $i$th sample. The optimal bandwidth selected for the $m$th candidate model is $\widehat{h}_m=\mathop{\arg\min}_{h_m}\mathrm{CV}(h_m)$.

\subsection{Model Averaging and Weight Choice Criterion}

Let \( \mathbf{w} = (w_1, w_2, \dots, w_{M})^{\mathrm{T}} \) be a weight vector satisfying \( \mathbf{w} \in \mathcal{H}_{n} = \{\mathbf{w} \in [0,1]^{M} : \sum_{m=1}^{M} w_m = 1\} \). The model averaging estimator of \( \bm{\mu} \) is defined as
\begin{equation}\label{eq7}
	\widehat{\bm{\mu}}(\mathbf{w}) = \sum_{m=1}^{M} w_m \widehat{\bm{\mu}}_{(m)} = \sum_{m=1}^{M} w_m \mathbf{P}_{(m)} \mathbf{Y} = \mathbf{P}(\mathbf{w}) \mathbf{Y},
\end{equation}
where \( \mathbf{P}(\mathbf{w}) = \sum_{m=1}^{M} w_m \mathbf{P}_{(m)} \). Define the squared error loss function as \( L_n(\mathbf{w}) = \|\widehat{\bm{\mu}}(\mathbf{w}) - \bm{\mu}\|^2 \), and the corresponding risk function as
\begin{align}\label{eq8}
	R_n(\mathbf{w}) &= E(L_n(\mathbf{w}) \mid \mathbf{X}, \bm{\Psi}) \notag \\
	&= E(\|(\mathbf{P}(\mathbf{w}) - \mathbf{I}_n)\bm{\mu} + \mathbf{P}(\mathbf{w})\bm{\epsilon}\|^2 \mid \mathbf{X}, \bm{\Psi}) \notag \\
	&= \|(\mathbf{P}(\mathbf{w}) - \mathbf{I}_n)\bm{\mu}\|^2 + \operatorname{tr}(\mathbf{P}^{\mathrm{T}}(\mathbf{w}) \mathbf{P}(\mathbf{w}) \bm{\Omega}).
\end{align}
Following \cite{H2007}, we use the following Mallows-type model averaging criterion to select \( \mathbf{w} \):
\begin{equation}\label{eq9}
	C_n(\mathbf{w}) = \|\mathbf{Y} - \widehat{\bm{\mu}}(\mathbf{w})\|^2 + 2\operatorname{tr}(\mathbf{P}(\mathbf{w})\bm{\Omega}).
\end{equation}
It can be shown that
\begin{align}\label{eq10}
	R_n(\mathbf{w}) &= E \left\{ \|\widehat{\bm{\mu}}(\mathbf{w}) - \bm{\mu}\|^2 \mid \mathbf{X}, \bm{\Psi} \right\} \notag \\
	&= E \left\{ \|\mathbf{Y} - \widehat{\bm{\mu}}(\mathbf{w})\|^2 + 2\bm{\epsilon}^{\mathrm{T}}(\widehat{\bm{\mu}}(\mathbf{w}) - \bm{\mu}) - \bm{\epsilon}^{\mathrm{T}}\bm{\epsilon} \mid \mathbf{X}, \bm{\Psi} \right\} \notag \\
	&= E \left\{ \|\mathbf{Y} - \widehat{\bm{\mu}}(\mathbf{w})\|^2 \mid \mathbf{X}, \bm{\Psi} \right\} + 2E \left\{ \operatorname{tr}(\mathbf{P}(\mathbf{w})\bm{\Omega}) \mid \mathbf{X}, \bm{\Psi} \right\} - n\sigma^2 \notag \\
	&= E(C_n(\mathbf{w}) \mid \mathbf{X}, \bm{\Psi}) - n\sigma^2,
\end{align}
which implies that \( C_n(\mathbf{w}) \) is an unbiased estimator of the squared error loss risk function \( R_n(\mathbf{w}) \) plus a constant.

The optimal weight vector is obtained by minimizing \( C_n(\mathbf{w}) \) over the weight set \( \mathcal{H}_{n} \):
\begin{equation}\label{eq11}
	\widehat{\mathbf{w}} = \mathop{\arg\min}_{\mathbf{w} \in \mathcal{H}_n} C_n(\mathbf{w}).
\end{equation}
\eqref{eq11} represents a quadratic programming problem, which can be solved using the "Rsolnp" package in R. The resulting estimator \( \widehat{\bm{\mu}}(\widehat{\mathbf{w}}) \) is referred to as the spatially varying coefficient Mallows model averaging (SVMMA) estimator.

Since the unknown parameter \( \sigma^2 \) remains in \( C_n(\mathbf{w}) \), following \cite{H2007}, we substitute the estimator from the largest candidate model \( M^* = \mathop{\arg\min}_{1 \leq m \leq M} p_m \) (i.e., the candidate model contains the largest number of covariates) \( \widehat{\bm{\Omega}} = \widehat{\sigma}^2_{M^*} \mathbf{I}_n \) into \eqref{eq9}, where \( \widehat{\sigma}^2_{M^*} = (n-\operatorname{tr}(\mathbf{P}_{(M^*)}))^{-1} (\mathbf{Y} - \mathbf{P}_{(M^*)} \mathbf{Y})^{\mathrm{T}} (\mathbf{Y} - \mathbf{P}_{(M^*)} \mathbf{Y}) \), and $n-\operatorname{tr}(\mathbf{P}_{(M^*)})$ denotes the adjusted degrees of freedom. Then, \eqref{eq9} can be transformed into
the feasible criterion:
\begin{equation}\label{eq12}
	\widehat{C}_n(\mathbf{w}) = \|\mathbf{Y} - \widehat{\bm{\mu}}(\mathbf{w})\|^2 + 2\operatorname{tr}(\mathbf{P}(\mathbf{w})\widehat{\bm{\Omega}}).
\end{equation}
Minimizing \eqref{eq12} gives:
\begin{equation}\label{eq13}
	\widetilde{\mathbf{w}} = \mathop{\arg\min}_{\mathbf{w} \in \mathcal{H}_n} \widehat{C}_n(\mathbf{w}).
\end{equation}
When \( \bm{\Omega} \) is unknown, substituting \( \widetilde{\mathbf{w}} \) into \( \widehat{\bm{\mu}}(\mathbf{w}) \) yields the feasible SVMMA estimator \( \widehat{\bm{\mu}}(\widetilde{\mathbf{w}}) \).

\section{Asymptotic Properties}

\subsection{Asymptotic Optimality of the SVMMA Estimator}

In this section, we establish the asymptotic optimality of the proposed SVMMA estimator when all candidate models are misspecified, which implies that its squared loss is asymptotically equivalent to that of the infeasible optimal averaging estimator. We first give some notations.
Let $\xi_n = \inf_{\mathbf{w} \in \mathcal{H}_n} R_n(\mathbf{w})$, $\underline{h} = \min_{1 \leq m \leq M} h_{m}$, $\bar{h} = \max_{1 \leq m \leq M} h_m$, $\widetilde{p} = \max_{1 \leq m \leq M} p_m$. Let $\bar{\lambda}(\cdot)$ and $\underline{\lambda}(\cdot)$ denote the largest and smallest singular values of a given matrix, respectively. Let $\mathbf{w}_{m}^{0}$ be an $M \times 1$ weight vector where the $m$th element is 1 and all other elements are 0. Unless stated otherwise, all limits are taken as $n \rightarrow \infty$. We impose the following technical conditions:

\noindent \textbf{Condition 1.} For some fixed integer $G$ (with $1 \leq G < \infty$) and a positive constant $c$, for any $i=1,2,\ldots,n$, $E(\epsilon^{4G}_i \mid \mathbf{X}_i, \mathbf{s}_{i}) < c$, a.s..

\noindent \textbf{Condition 2.} $M \xi_{n}^{-2G} \widetilde{p}^{G} \sum_{m=1}^{M} \{R_{n}(\mathbf{w}_{m}^{0})\}^{G} = o_p(1)$.

\noindent \textbf{Condition 3.} The spatial locations $\mathbf{s}$ have a bounded support set $D$. The density function $f(\cdot)$ of $\mathbf{s}_i$ is bounded away from zero and infinity on its support, and is at least twice continuously differentiable. For all $k=1,2,\dots,p$, the function $\beta_k(\cdot)$ is twice differentiable on the support.

\noindent \textbf{Condition 4.} $\max_{1 \leq m \leq M} \max_{1 \leq i \leq n} \mathbf{X}_{(m)i}^{\mathrm{T}} \mathbf{X}_{(m)i} = O_p(\widetilde{p})$. For any $1 \leq i \leq n$, $\mathbf{C}_\mathbf{X} = E(\mathbf{X}_i \mathbf{X}_i^{\mathrm{T}})$ is non-singular. For any $i=1,2,\dots,n$ and $l,k=1,2,\dots,p$, there exist positive constants $\kappa$, $\underline{c}_\mathbf{X}$, and $\bar{c}_\mathbf{X}$ such that $E\{(x_{il}x_{ik})^2\} \leq \kappa$ and $0 < \underline{c}_\mathbf{X} \leq \min_{1 \leq m \leq M} \underline{\lambda}(\bm{\Pi}_{m} \mathbf{C}_\mathbf{X} \bm{\Pi}_{m}^{\mathrm{T}}) \leq \max_{1 \leq m \leq M} \bar{\lambda}(\bm{\Pi}_{m} \mathbf{C}_\mathbf{X} \bm{\Pi}_{m}^{\mathrm{T}}) \leq \bar{c}_\mathbf{X} < \infty$.

\noindent \textbf{Condition 5.} The kernel function $K(\cdot)$ is a bounded symmetric function with a symmetric compact support set $\operatorname{supp}(K)$. For a fixed $q$, define $K^{\prime}_q(s_1,s_2) = K(\|\mathbf{s}\|_q)$, which satisfies $\iint_{(v_1,v_2) \in \operatorname{supp}(K^{\prime}_q)} K^{\prime}_q\left( v_{1},v_{2} \right) dv_1 dv_2 = 1$.

\noindent \textbf{Condition 6.} $\bar{h} \rightarrow 0$ and $n^{-1} \bar{h}^{-2} \rightarrow 0$. The bandwidth $\underline{h}$ satisfies the same conditions as $\bar{h}$.

\noindent \textbf{Condition 7.} $\bm{\mu}^{\mathrm{T}} \bm{\mu} / n = O(1)$,a.s..

\noindent \textbf{Condition 8.} $\xi_{n}^{-1}\underline{h}^{-2} \widetilde{p} = o_p(1)$ and $\bar{h}^4\widetilde{p}=o(1)$.

Condition 1 imposes restrictions on the conditional moments of the error terms, which is similar to Condition (C.1) in \cite{ZWZZ19}. Condition 2 is analogous to Condition (8) in \cite{WZZ2010} and Condition (C.2) in \cite{ZWZZ19}. This condition, common in model averaging research, requires that all candidate models be misspecified to hold. Condition 3 consists of common assumptions regarding the spatial locations and the regression coefficient functions, which are similar to Conditions (C.3) in \cite{ZWZZ19}. Condition 4 relates to the moment conditions of the covariates $\mathbf{X}$, paralleling Conditions (C.4) in the same reference. Condition 5 represents standard assumptions on the kernel function, which is common in the literature; see, for example, \cite{BFC98} and \cite{ZWZZ19}. The compact support condition is primarily imposed for proof simplicity and could potentially be relaxed with more complex arguments. Specifically, the Gaussian kernel is permissible. Condition 6 is a common assumption of bandwidth. Condition 7 concerns the sum of the elements of $\bm{\mu}$, an assumption analogous to those in \cite{WZZ2010} and \cite{ZWZZ19}. Condition 8 requires $\xi_n$ to grow faster than $\underline{h}^{-2} \widetilde{p}$ and allows $\widetilde{p}$ to diverge as $n \to \infty$, subject to a constraint on its divergence rate.

\begin{thm}\label{th.1} 
Suppose that Conditions 1-8 hold. Then as $n \rightarrow \infty$, we have
\begin{equation}
	\frac{L_n(\widehat{\mathbf{w}})}{\inf_{\mathbf{w} \in \mathcal{H}_n} L_n(\mathbf{w})} \overset{p}{\rightarrow} 1.
\end{equation}
\end{thm}

\begin{thm}\label{th.2} 
Suppose that Conditions 1-8 hold. Then as $n \rightarrow \infty$, we have
\begin{equation}
	\frac{L_n(\widetilde{\mathbf{w}})}{\inf_{\mathbf{w} \in \mathcal{H}_n} L_n(\mathbf{w})} \overset{p}{\rightarrow} 1.
\end{equation}
\end{thm}

\begin{remark}\label{rk2}
It's worth noting that prediction is mainly focused when all candidate models are misspecified.
Theorem 1 shows that when all candidate models are misspecified, the SVMMA estimator is asymptotically optimal
in the sense of making the squared error loss as small as possible among all feasible weight vectors $\mathbf{w}$ lying in the set $\mathcal{H}_{n}$.
This property offers a theoretical advantage of our method over single model-based methods (e.g., a single optimal model chosen by AIC, AICc or BIC), as well as other model weight selection methods (e.g., the model weights are selected via the smoothed AIC or BIC). Simulation results in Section 4 have verified this property. Theorem 2 shows that Theorem 1 still holds when $\bm{\Omega}$ is replaced by $\widehat{\bm{\Omega}}$. 
\end{remark}

\subsection{Consistency of the Model Weights and SVMMA Estimator}

In this section, we consider the scenario where the candidate set includes at least one quasi-correct model. It is well known that not all covariates in model \eqref{eq1} have contributions in estimating the conditional mean. Let $\mathcal{A}_{true}\triangleq \{k:\beta_k(\mathbf{s}_i)\neq 0,1\leq k \leq p\}$ be the true active index set in model \eqref{eq1}. Suppose that $\mathcal{A}_m $ is the index set of covariates to be included in the $m$th candidate model for $m=1,\cdots,M$.
The $m$th candidate model is called as the \emph{quasi-correct} model if $\mathcal{A}_{true} \subseteq\mathcal{A}_m$. Particularly, when $\mathcal{A}_{true}=\mathcal{A}_m$, the $m$th candidate model is called as the \emph{correct} model. Thus, it is not hard to understand that the set of quasi-correct models consists of the correct and over-fitted models, also implying that the quasi-correct model may not be unique. 

To ensure the consistency of the averaging weights, we assume that the $j$th $(j = 1, 2, \ldots , p_m)$ component of $\widehat{\bm{\beta}}_{(m)}(\mathbf{s})(m = 1, 2, \ldots , M)$ has a limiting value $\beta_{(m)j}^{*}(\mathbf{s})$ for any $\mathbf{s} = (s_1, s_2)^{\mathrm{T}} \in D$. Moreover, $\bm{\beta}^{*}_{(m)}(\mathbf{s})=(\beta_{(m)1}^{*}(\mathbf{s}),\beta_{(m)2}^{*}(\mathbf{s}),\ldots,\beta_{(m)p_m}^{*}(\mathbf{s}))^{\mathrm{T}}$ satisfies
$$E\{K_{h_m}(\|\mathbf{u}-\mathbf{s}\|_q)\mathbf{X}_{(m)}(\mathbf{s})(y(\mathbf{s})-\mathbf{X}^{\mathrm{T}}_{(m)}(\mathbf{s})\bm{\beta}^{*}_{(m)}(\mathbf{s}))\}=\bm{0},$$
where $\mathbf{u} = (u_1, u_2)^{\mathrm{T}}$ is a random vector with probability density function $f(\mathbf{u})$. If the $m$th model is quasi-correct model, then for $j\in \mathcal{A}_m\cap\mathcal{A}_{true}$, we have $\beta^{*}_{(m)j}(\mathbf{s})=\beta_{j}(\mathbf{s})$, while for $j\in \mathcal{A}_m\setminus\mathcal{A}_{true}$ we have $\beta^{*}_{(m)j}(\mathbf{s})=0$. For any $m = 1, \ldots, M$, define $\bm{\mu}^{*}_{(m)}=(\mu^{*}_{(m)1},\mu^{*}_{(m)2},\ldots,\mu^{*}_{(m)n})^{\mathrm{T}}$, where $\mu^{*}_{(m)i}=\mathbf{X}_{(m)i}^{\mathrm{T}} \bm{\beta}^{*}_{(m)}(\mathbf{s}_i)$. Without loss of generality,
we assume that the first $M_0$ candidate models are quasi-correct. Let ${\tau}_{M_0}=\sum_{m=1}^{M_0}{w}_m$ be the sum of weights assigned to quasi-correct models. Recall that $\widehat{\mathbf{w}}=(\hat w_1,\ldots,\hat w_M)^{\mathrm{T}}$ and $\widetilde{\mathbf{w}}=(\tilde w_1,\ldots,\tilde w_M)^{\mathrm{T}}$ are defined in \eqref{eq11} and \eqref{eq13}, respectively. Denote $\widehat{\tau}_{M_0}=\sum_{m=1}^{M_0}\widehat{w}_m$ and $\widetilde{\tau}_{M_0}=\sum_{m=1}^{M_0}\widetilde{w}_m$.
Define $\mathcal{H}_F = \{\mathbf{w} \in \mathcal{H}_n: w_m = 0, m = 1, \dots, M_0\}$ as the set of weight vectors that assign zero weight to all quasi-correct models. Let $\xi_F=\inf_{\mathbf{w}\in\mathcal{H}_F}L_n^{*}(\mathbf{w})$, where 
$L_n^{*}(\mathbf{w})=\|\bm{\mu}^{*}(\mathbf{w})-\bm{\mu}\|^2$, $\bm{\mu}^{*}(\mathbf{w})=\sum_{m=1}^{M}w_m\bm{\mu}^{*}_{(m)}$.
To establish the limiting properties of the weights, the following condition is required.

\noindent \textbf{Condition 9.} For every $k = 1, 2, \dots, p_m$ and $m = 1, 2, \dots, M$, the function  $\beta_{(m)k}^{*}(\cdot)$ is twice differentiable on its support. Moreover, the first- and second-order partial derivatives of the function family $\{\beta_{(m)k}^{*}(\cdot) : k = 1, \dots, p_m; m = 1, \dots, M\}$ are uniformly bounded, and this uniform bound is independent of $p_m$ and $M$.

\noindent \textbf{Condition 10.} For any $m = 1, \ldots, M$, there exists a positive constant $\delta_1$ such that $E(\mu_{(m)1}^{*4}|\mathbf{s}_1=s)\leq\delta_1$.

\noindent \textbf{Condition 11.} $\xi_F^{-1}n^{1/2} = o_p(1)$, $\xi_F^{-1}M\widetilde{p}^4= o_p(1)$, $\xi_F^{-1}\underline{h}^{-2} M\widetilde{p}^2 = o_p(1)$, and $\xi_F^{-1}nM\widetilde{p}^2 \bar{h}^4 = o_p(1)$.

Condition 9 stipulates that every coefficient function $\beta_{(m)k}^{*}(\cdot)$ is twice differentiable and that its first and second partial derivatives are uniformly bounded across all $k$ and $m$, with bounds independent of $p_m$ and $M$. This smoothness and derivative boundedness ensure stable asymptotic behavior as the model dimensions grow. 
Condition 10 requires that the conditional fourth moment of $\mu_{(m)1}^{*}$ be uniformly bounded. This avoids the moment explosion problem caused by spatial heterogeneity or model misspecification when the number of models $M$ diverges.
When $\widetilde{p}$ is finite, Condition 10 clearly holds under Conditions 4 and 9. Condition 11 implies that the order of $\xi_F$ is higher than $n^{1/2}$, $ M\widetilde{p}^4$, $\underline{h}^{-2} M\widetilde{p}^2$, and $nM\widetilde{p}^2 \bar{h}^4$. This is similar to Condition 3.8 in \cite{SZZ2024}, but it should be noted that the latter assumes finite-dimensional covariates and is based on a linear regression framework. 

\begin{thm}\label{th.3} 
Suppose that Conditions 3-6 and 9-11 hold. Then as $n \rightarrow \infty$, we have
\begin{equation}
	\widehat{\tau}_{M_0}\overset{p}{\rightarrow} 1.
\end{equation}
\end{thm}

\begin{corollary}\label{corollary1} 
Under Conditions 3-6 and 9-11, we have
\begin{equation}
	(\widehat{\mu}_i(\widehat{\mathbf{w}})-\mu_i)^2=o_p\left(1\right),
\end{equation}
where $\widehat{\mu}_i(\widehat{\mathbf{w}})=\sum_{m=1}^{M}\widehat{w}_m\widehat{\mu}_{(m)i}$.
\end{corollary}

\begin{remark}\label{rk3}
Theorem 3 shows that if at least one quasi-correct model exists in the candidate set, SVMMA tends to assign all weights to the quasi-correct models, which also demonstrates the effectiveness of the proposed method. 
Corollary 1 demonstrates that the proposed model averaging estimator converges to the true conditional mean. This result provides theoretical consistency for the SVMMA estimator when the candidate set contains at least one quasi-correct model.
\end{remark}

\begin{thm}\label{th.4} 
Suppose that Conditions 3-6 and 9-11 hold. Then as $n \rightarrow \infty$, we have
\begin{equation}
	\widetilde{\tau}_{M_0}\overset{p}{\rightarrow} 1.
\end{equation}
\end{thm}

\begin{corollary}\label{corollary2} 
Under Conditions 3-6 and 9-11, we have
\begin{equation}
	(\widehat{\mu}_i(\widetilde{\mathbf{w}})-\mu_i)^2=o_p\left(1\right),
\end{equation}
where $\widehat{\mu}_i(\widetilde{\mathbf{w}})=\sum_{m=1}^{M}\widetilde{w}_m\widehat{\mu}_{(m)i}$.
\end{corollary}

Theorem 4 and Corollary 2 show that Theorem 3 and Corollary 1 still hold when $\bm{\Omega}$ is replaced by $\widehat{\bm{\Omega}}$.

\section{Simulation Study}

In this section, we consider three simulation designs to validate the theories presented in Section 3. Designs 1 and 2 examine scenarios where no quasi-correct model is included in the candidate set, aiming to compare the finite-sample prediction performance of the SVMMA estimator with those of existing approaches.
Design 3 investigates situations where at least one quasi-correct model is present in the candidate set, in order to verify whether the sum of weights on the quasi-correct models gradually approaches to one as the sample size increases.

\subsection{Data Generation Process}

\textbf{Design 1.} The samples $\{y_i, \mathbf{X}_i, \mathbf{s}_i\}$, $i=1,2,\ldots,n$, are generated from the model
$$y_i = \mu_i+c\varepsilon_i=\sum_{j=1}^{M+200}\theta_jx_{ij}+c\varepsilon_i,i = 1,2,\ldots,n,$$
where the covariate $x_{i1} = 1$, and $x_{ij}$ for $j=2,\ldots,M+200$ are independently and identically distributed as $N(0,1)$. The coefficients are set as $\theta_1 = 10^{-\alpha-0.5}$, $\theta_j = j^{-\alpha-0.5}$ for $j \geq 2$, with $\alpha$ being the decay parameter controlling the rate of coefficient decay, taken as 0.5 or 1. The constant $c$ is used to control the proportion of covariate information in the total variance, and is chosen such that the coefficient of determination $R^2 = \mathrm{var}(\mu_i)/\mathrm{var}(y_i)$ varies between 0.1 and 0.9. The random error term $\varepsilon_i$ is considered under both homoscedastic and heteroscedastic settings: (Case i) $\varepsilon_i \overset{\text{i.i.d}}{\sim} N(0,1)$, (Case ii) $\varepsilon_i \overset{\text{i.i.d}}{\sim} t(5)$, (Case iii) $\varepsilon_i = \sqrt{a_i} u_i$, where $a_i = 0.2 + 0.5x_{i2}^2$, $u_i \overset{\text{i.i.d}}{\sim} N(0,1)$. In Design 1, the samples are from the traditional linear regression model.

\textbf{Design 2.} The sample data $\{y_i, \mathbf{X}_i, \mathbf{s}_i\}$, $i=1,2,\ldots,n$, are generated from the model
$$y_i = \mu_i+c\varepsilon_i=\sum_{j=1}^{M+200}\theta_jF(\mathbf{s}_i)x_{ij}+c\varepsilon_i,i = 1,2,\ldots,n,$$
where $\mathbf{s}_i = (s_{i1}, s_{i2})^{\mathrm{T}}$ are taken from grid points $(\frac{j}{\sqrt{n}}, \frac{k}{\sqrt{n}})^{\mathrm{T}}$, $j,k = 1,2,\ldots,\sqrt{n}$, over the square region $[0,1] \times [0,1]$. We set $F(\mathbf{s}_i) = 1 - (1-2s_{i1})^2 - (1-2s_{i2})^2$. The configurations of the remaining variables and parameters are consistent with Design 1. Clearly, the samples are from the SVCM in this experiment.

We set the sample sizes to $n = 100$, $225$, and $400$, with the number of simulation replications N = 200. The number of candidate models M is determined following the criterion in \cite{H2007}, taken as the largest integer not exceeding $3n^{\frac{1}{3}}$. For simplicity, we consider only nested candidate models, where the $m$th candidate model includes the first 1 to $m$ covariates from $x_{i1}, x_{i2}, \ldots, x_{i,M+200}$.

\textbf{Design 3.} The sample data $\{y_i, \mathbf{X}_i, \mathbf{s}_i\}$, $i=1,2,\ldots,n$, are generated from the model
$$y_i = \mu_i+c\varepsilon_i=\sum_{j=1}^{6}\theta_jF(\mathbf{s}_i)x_{ij}+c\varepsilon_i,i = 1,2,\ldots,n,$$
were $\mathbf{X}_{i} = (x_{i1}, \ldots, x_{i6})^{\mathrm{T}}$ follows a multivariate normal distribution with mean zero and covariance matrix $\bm{\Sigma}_x$. The diagonal elements of $\bm{\Sigma}_x$ are 1 and all off-diagonal elements are 0.5. We set $\bm{\theta}=(1,1.2,-1,0.9,0,0)^{\mathrm{T}}$, while keeping the settings of other variables and parameters consistent with Design 2.

In this design, all explanatory variables are used to construct candidate models, with each candidate model containing at least one explanatory variable. This results in $2^6 - 1 = 63$ candidate models in total. Among them, there are $2^2=4$ quasi-correct models, each consisting of all four explanatory variables with non-zero coefficient functions, together with any subset of the explanatory variables whose coefficient functions are zero. We set the sample sizes to $n = 100$, $169$, $225$, $324$, and $400$, the number of simulation replications to $N = 200$, and vary the $R^2$ between 0.5 and 0.9.

\subsection{Evaluation and Comparison}

In the simulation study, we employ the Gaussian kernel function to estimate each candidate model. The spatial distance metric is chosen with  q = 2 , and the bandwidth in the candidate models is selected via the CV criterion, see Section 2.1. 

We compare the finite sample performance of the SVMMA estimator with the AIC- and
BIC-based model selection methods and existing model averaging estimators. The AIC, BIC, and AICc \citep{HDM06} criteria for the $m$th candidate model are defined respectively as $\mathrm{AIC}_m = \log(\widehat{\sigma}_m^2) + 2n^{-1} \operatorname{tr}(\mathbf{P}_{(m)})$,
$\mathrm{BIC}_m = \log(\widehat{\sigma}_m^2) + n^{-1} \operatorname{tr}(\mathbf{P}_{(m)}) \log(n)$
and $\mathrm{AICc}_m = \log(\widehat{\sigma}_m^2) + \frac{n + \operatorname{tr}(\mathbf{P}_{(m)})}{n - \operatorname{tr}(\mathbf{P}_{(m)}) - 2}$,
where $\widehat{\sigma}_m^2 = n^{-1} \|\mathbf{Y} - \widehat{\bm{\mu}}_{(m)}\|^2$ and $\mathbf{P}_{(m)}$ is given in \eqref{eq6} of Section 2.1. Each of the three criteria selects the optimal SVCM among all candidate models that minimizes the corresponding criterion score. Then, one can use the optimal SVCM (chosen by AIC or BIC or AICc) to estimate the conditional mean.  We refer them as AIC, BIC and AICc. Meanwhile, the performance of the proposed SVMMA estimator heavily depends on the weight selection criterion, and hence we consider the popular smoothed AIC (SAIC) and smoothed BIC (SBIC) \citep{BBA97,GZ2024} to determine the model weights. Specifically, the SAIC (SBIC) assigns the model weight to the $m$th candidate model by $\widehat w_m^\mathrm{SAIC} = \frac{\exp(-\mathrm{AIC}_m / 2)}{\sum_{l=1}^{M} \exp(-\mathrm{AIC}_l / 2)}$ ($\widehat w_m^\mathrm{SBIC} = \frac{\exp(-\mathrm{BIC}_m / 2)}{\sum_{l=1}^{M} \exp(-\mathrm{BIC}_l / 2)}$) for $m=1,\cdots,M$. We refer them as SAIC and SBIC. 
 Since the data-generating process in Design 1 adheres to a linear specification, we further benchmark our SVMMA estimator against two linear framework counterparts, namely the MMA \citep{H2007} and JMA \citep{HR2012} estimators to assess the prediction robustness under candidate model misspecifications. For the linear model scenario, we adopt the identical AIC and BIC configurations following \cite{H2007}. 
For the MMA and JMA, the candidate models consist of linear models, while for all other methods, the candidate models are the SVCMs.

Following \cite{H2007}, the performance of the estimators is evaluated by the relative risk (Risk),
$$\mathrm{Risk} = \frac{ \frac{1}{N} \sum_{j=1}^{N} \| \widehat{\bm{\mu}}^{(j)} - \bm{\mu}^{(j)} \|^2 }{ \frac{1}{N} \sum_{j=1}^{N} \| \widehat{\bm{\mu}}^{(j)}(\bar{\mathbf{w}}^{(j)}) - \bm{\mu}^{(j)} \|^2 },$$
where $\widehat{\bm{\mu}}^{(j)}$ and $\bm{\mu}^{(j)}$ denote the respective generic estimator and true conditional mean at the $j$th repetition, $\bar{\mathbf{w}}^{(j)} = \arg\min_{\mathbf{w} \in \mathcal{H}_n} \| \widehat{\bm{\mu}}^{(j)}(\mathbf{w}) - \bm{\mu}^{(j)} \|^2$ represents the infeasible optimal weight vector and $\widehat{\bm{\mu}}^{(j)}(\bar{\mathbf{w}}^{(j)}) = \sum_{m=1}^{M} \bar w_m^{(j)} \widehat{\bm{\mu}}^{(j)}_{(m)}$ are the infeasible optimal model averaging estimator with $\widehat{\bm{\mu}}^{(j)}_{(m)}$ being the estimator of $\bm{\mu}^{(j)}$ under the $m$th candidate model at the $j$th repetition. In simulation studies, we consider two type of infeasible optimal model averaging estimators, including the infeasible optimal linear model averaging (OLMA) and infeasible optimal spatially varying coefficient model averaging (OSVCMA). Specifically, when the candidate model is a linear model (or a SVCM), we refer to $\widehat{\bm{\mu}}^{(j)}(\bar{\mathbf{w}}^{(j)})$ as the infeasible OLMA (or OSVCMA). Thus, the relative risk is measured against either the infeasible OLMA risk or infeasible OSVCMA risk. 

\subsection{Results}

Figures~\ref{fig:1}--\ref{fig:3} display the simulation results for Design 1. 
As the true data generation process is from the linear model, the relative risk in Design 1 is measured against the infeasible OLMA risk. Given that the true data generation process is derived from a linear model, it comes as no surprise that the MMA and JMA methods exhibit the lowest relative risk across the set of comparative approaches, but SVMMA achieve competitive performance when compared with the MMA and JMA. This indicates that SVMMA has exhibits broader adaptability and greater flexibility. In addition, it is easy to see that SVMMA obviously performs better than the single model-based methods (e.g., a single optimal model chosen by AIC or BIC), showing the benefits of model averaging. Meanwhile, SVMMA outperforms SAIC and SBIC under all situations, which effectively verifies the asymptotic optimality of SVMMA given in Theorem \ref{th.2}.

\begin{figure}[H]
	\centering
	\includegraphics[height=11cm]{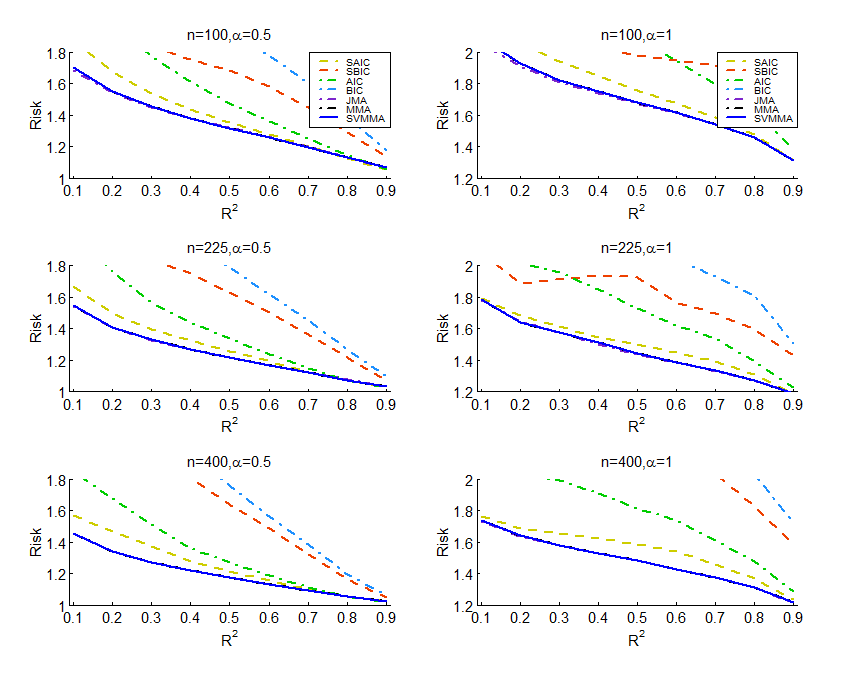}
	\vspace{-5mm}
	\caption{Risk Comparisons under Design 1 for Case i. The relative risk is measured against the infeasible OLMA risk.}
	\label{fig:1}
\end{figure}
\vspace{-5mm}

\begin{figure}[H]
	\centering
	\includegraphics[height=11cm]{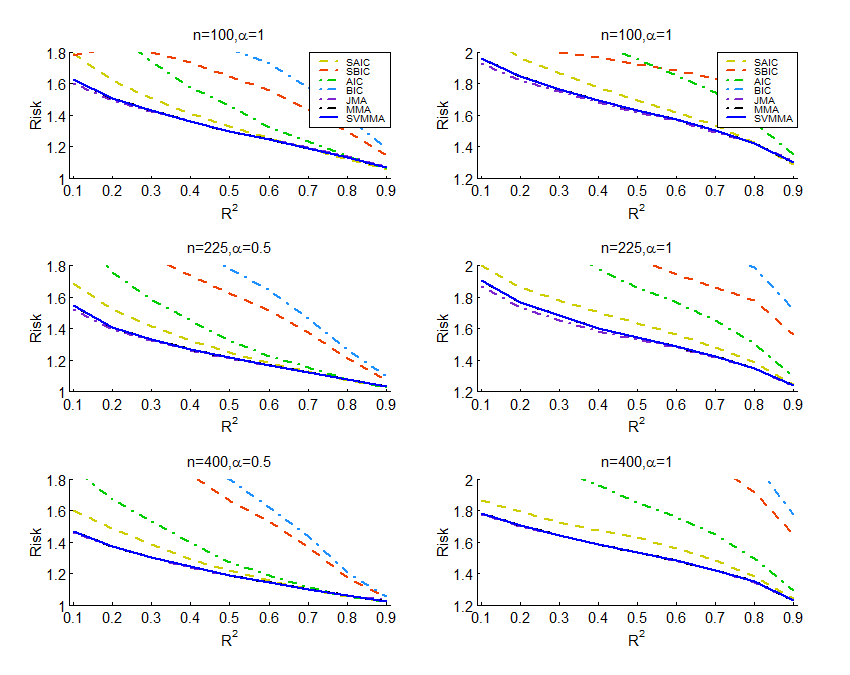}
	\vspace{-5mm}
	\caption{Risk Comparisons under Design 1 for Case ii. The relative risk is measured against the infeasible OLMA risk.}
	\label{fig:2}
\end{figure}

\begin{figure}[H]
	\centering
	\includegraphics[height=11cm]{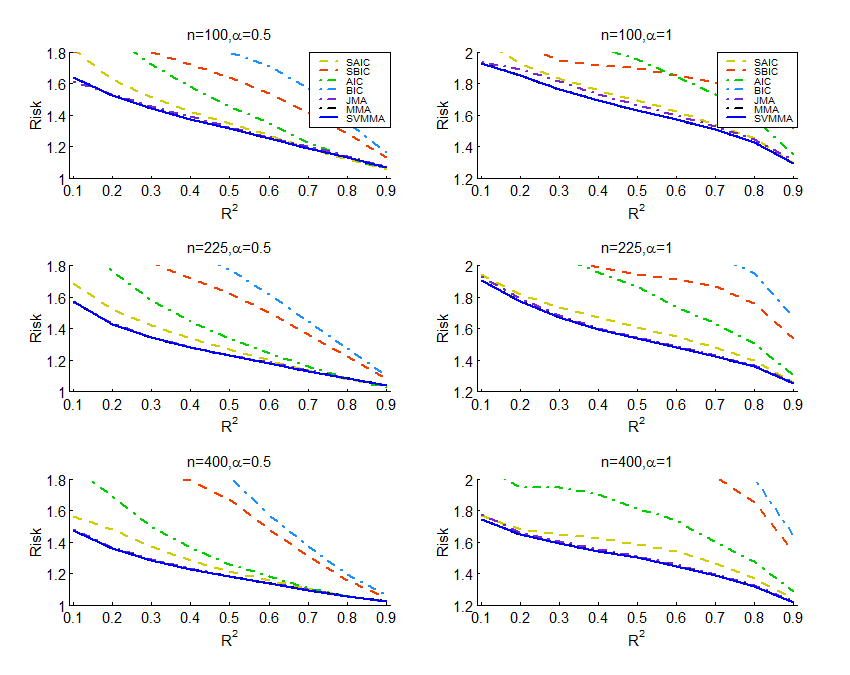}
	\vspace{-5mm}
	\caption{Risk Comparisons under Design 1 for Case iii. The relative risk is measured against the infeasible OLMA risk.}
	\label{fig:3}
\end{figure}

\begin{figure}[H]
	\centering
	\includegraphics[height=11cm]{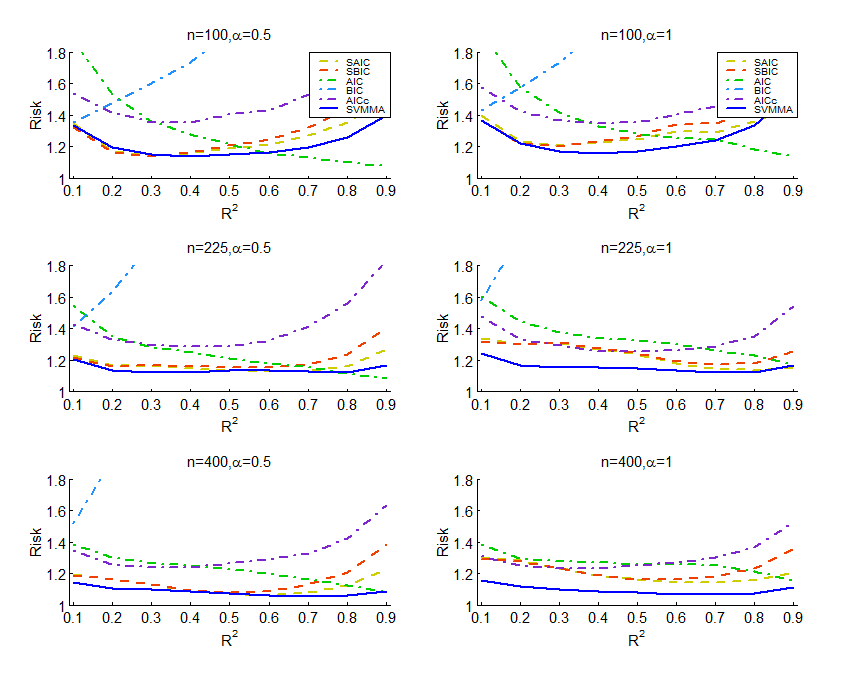}
	\vspace{-5mm}
	\caption{Risk Comparisons under Design 2 for Case i. The relative risk is measured against the infeasible OSVCMA risk.}
	\label{fig:4}
\end{figure}

\begin{figure}[H]
	\centering
	\includegraphics[height=11cm]{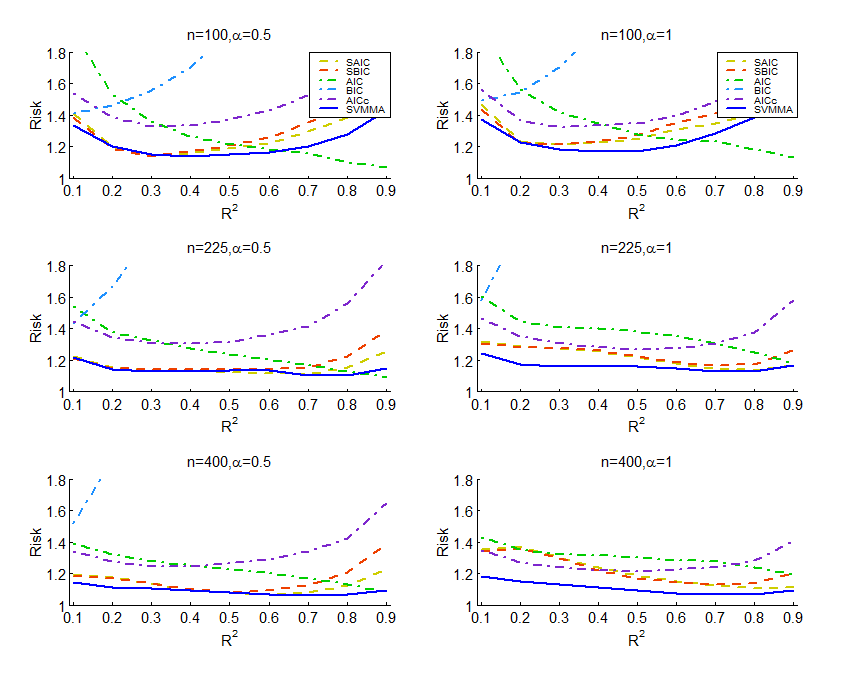}
	\vspace{-5mm}
	\caption{Risk Comparisons under Design 2 for Case ii. The relative risk is measured against the infeasible OSVCMA risk.}
	\label{fig:5}
\end{figure}

\begin{figure}[H]
	\centering
	\includegraphics[height=11cm]{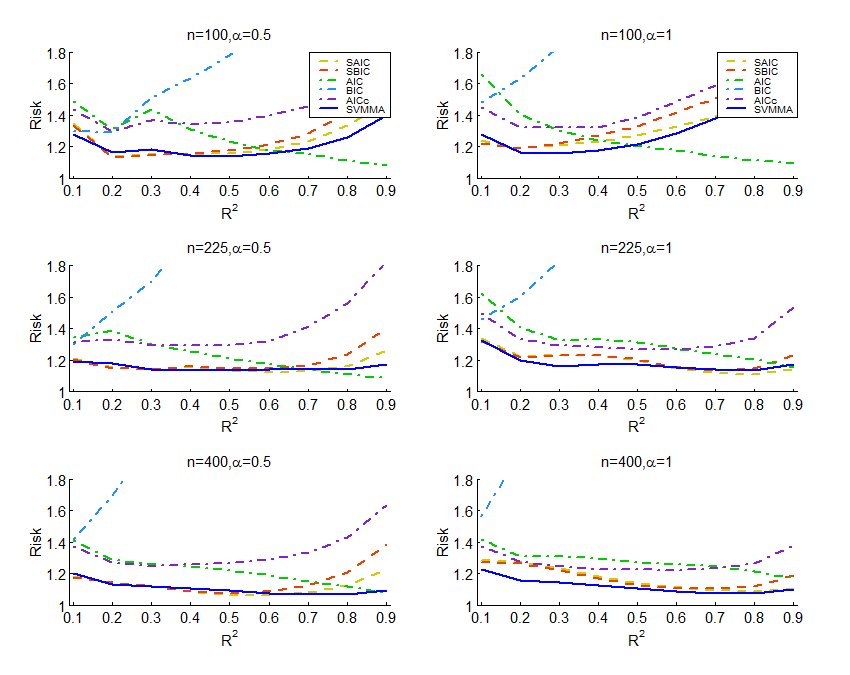}
	\vspace{-5mm}
	\caption{Risk Comparisons under Design 2 for Case iii. The relative risk is measured against the infeasible OSVCMA risk.}
	\label{fig:6}
\end{figure}

\begin{figure}[H]
	\centering
	\includegraphics[height=11cm]{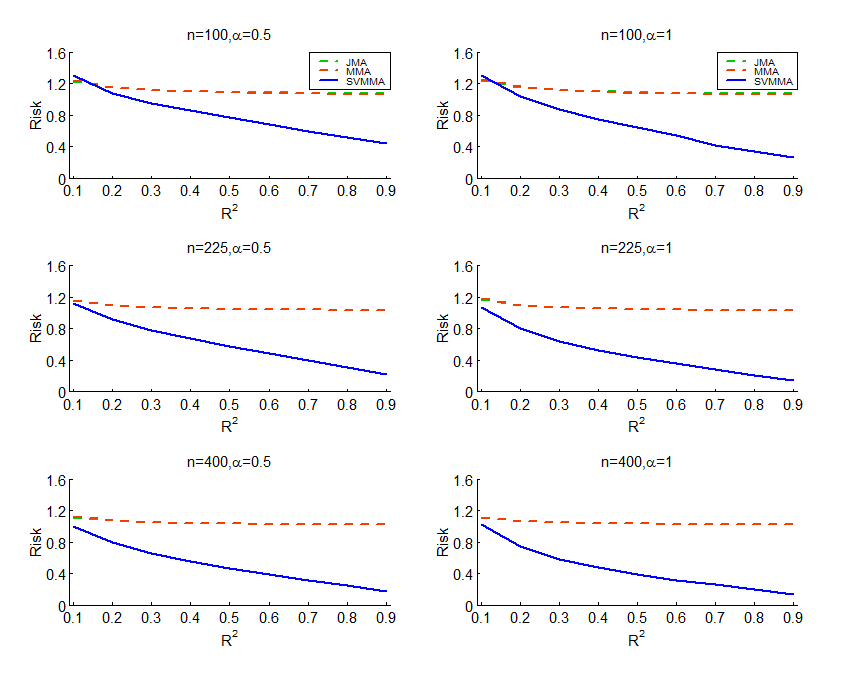}
	\vspace{-5mm}
	\caption{Risk Comparisons under Design 2 for Case i.The relative risk is measured against the infeasible OLMA risk.}
	\label{fig:7}
\end{figure}

\begin{figure}[H]
	\centering
	\includegraphics[height=11cm]{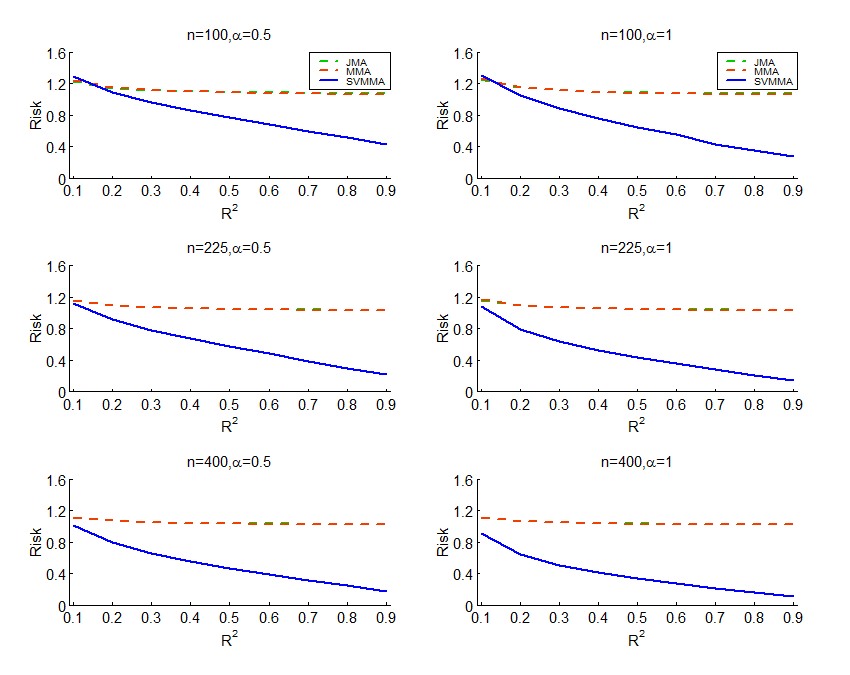}
	\vspace{-5mm}
	\caption{Risk Comparisons under Design 2 for Case ii. The relative risk is measured against the infeasible OLMA risk.}
	\label{fig:8}
\end{figure}

\begin{figure}[H]
	\centering
	\includegraphics[height=11cm]{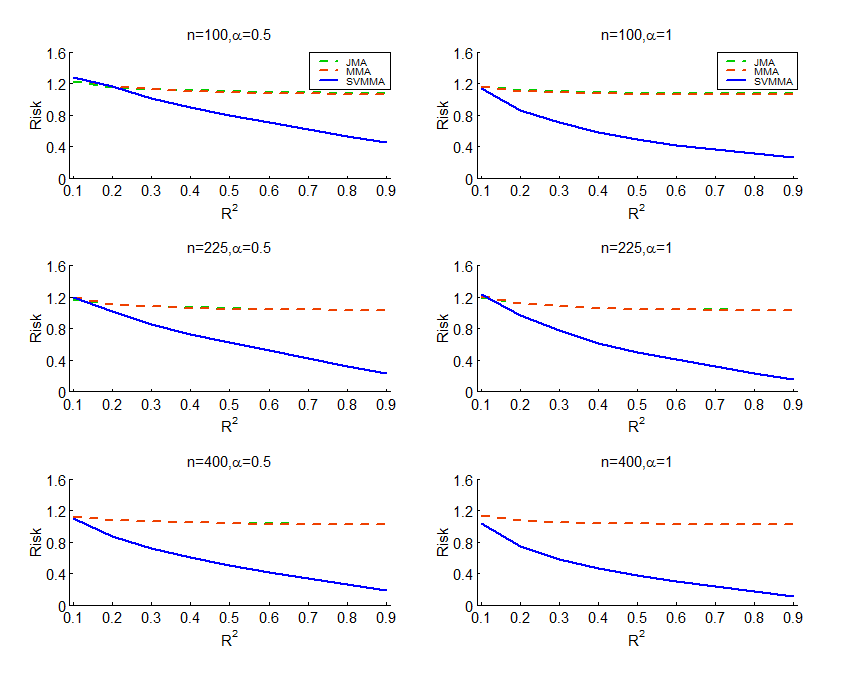}
	\vspace{-5mm}
	\caption{Risk Comparisons under Design 2 for Case iii. The relative risk is measured against the infeasible OLMA risk.}
	\label{fig:9}
\end{figure}

To conduct a comprehensive analysis, the relative risk in Design 2 is measured against the infeasible OSVCMA risk or OLMA risk. The corresponding results are listed in Figures~\ref{fig:4}--\ref{fig:9}.
It is clear from Figures~\ref{fig:4}--\ref{fig:6} that SVMMA generally leads to the lowest risk, especially for $\alpha=1$, highlighting the importance of accounting for spatial heterogeneity. Further, Figures~\ref{fig:7}--\ref{fig:9} illustrate that SVMMA has distinct advantages over MMA and JMA, with its superiority becoming increasingly evident as $R^2$ increases.
Moreover, the relative risk of SVMMA is less than one in most areas of $R^2$, indicating that SVMMA is superior to the infeasible OLMA approach. Actually, such a phenomenon is not surprising because OLMA does not take spatial heterogeneity into account.

Figure~\ref{fig:10} depicts the trends of $\widehat{\tau}_{4}$ and $\widetilde{\tau}_{4}$ against sample size under Design 3. Note that  $\widehat{\tau}_{4}$ and $\widetilde{\tau}_{4}$ are the sum of weights assigned to the quasi-correct models using \eqref{eq9} and \eqref{eq12} respectively.
The results show that the weights assigned to the quasi-correct models gradually approach to one when the sample size increases, which verifies the weight consistency established in Theorems~\ref{th.3}--\ref{th.4}. Figure~\ref{fig:11} illustrates the trend of the sample mean squared error (MSE) of the response variable with increasing sample size, where $\mathrm{MSE} = \frac{1}{nN}\sum_{j=1}^{N}\|\widehat{\bm{\mu}}^{(j)} - \bm{\mu}^{(j)}\|^2$, with $\widehat{\bm{\mu}}^{(j)}$ and $\bm{\mu}^{(j)}$ representing the estimated conditional mean obtained via \eqref{eq9} or \eqref{eq12} and the true conditional mean, respectively, at the $j$th repetition. The results indicate that as the sample size increases, our proposed conditional mean estimator converges to the true conditional mean, thereby validating the consistency of the estimator established in Corollaries~\ref{corollary1}--\ref{corollary2}.

\begin{figure}
	\centering
	\includegraphics[height=11cm]{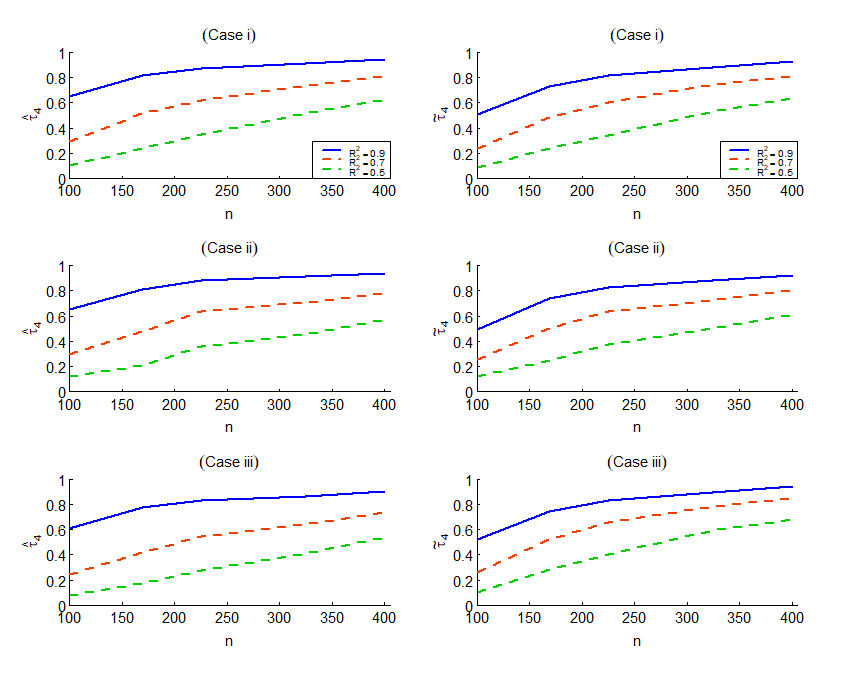}
	\vspace{-5mm}
	\caption{Under Design 3, the sum of quasi-correct model weights is obtained by two methods: \eqref{eq9} for the left column, and \eqref{eq12} for the right column.}
	\label{fig:10}
\end{figure}

\begin{figure}
	\centering
	\includegraphics[height=11cm]{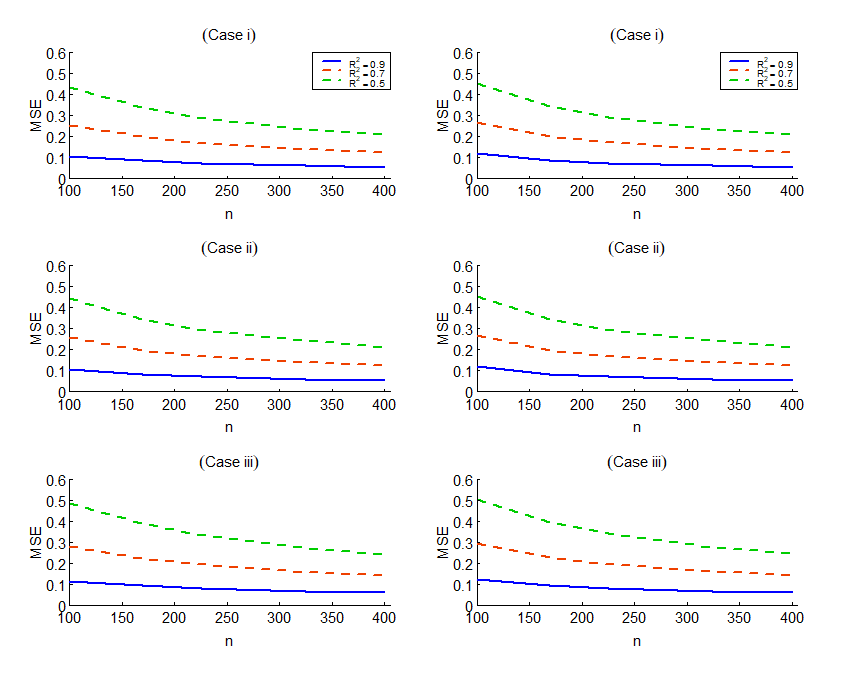}
	\vspace{-5mm}
	\caption{Under Design 3, the MSE results in the left column were obtained via \eqref{eq9}, while those in the right column were obtained via \eqref{eq12}.}
	\label{fig:11}
\end{figure}

\section{Empirical Application}

To validate the proposed SVMMA methodology, we employ a real data set from the Liudaogou watershed (Loess Plateau, Shaanxi Province, China), comprising measurements at 689 locations. The data set is publicly available from \cite{CBCD23} at the associated repository: \url{https://github.com/lexcomber/GWRroutemap}, with the authors explicitly stating that it is provided for reproducibility studies. This data set was initially collected and used in the research by \cite{WZH2009} and has since been adopted in subsequent studies including those by \cite{CWLZH18} and \cite{CBCD23}. The data were collected at approximate $100\mathrm{m} \times 100\mathrm{m}$ grid locations, as shown in Figure~\ref{fig:12}. The dependent variable of interest is soil total nitrogen (STN). The covariates include soil organic carbon (SOCgkg), nitrate-nitrogen (NO3Ngkg), ammonium (NH4Ngkg), and percentage clay (ClayPC), silt (SiltPC), sand (SandPC) content. As in the studies of \cite{WZH2009}, \cite{CWLZH18}, and \cite{CBCD23}, we transform a subset of variables: STN, SOCgkg, NO3Ngkg, and NH4Ngkg using natural logarithms, while applying a square root transformation to ClayPC. Further, we standardize SiltPC and SandPC.

\begin{figure}
	\centering
	\includegraphics[height=10cm]{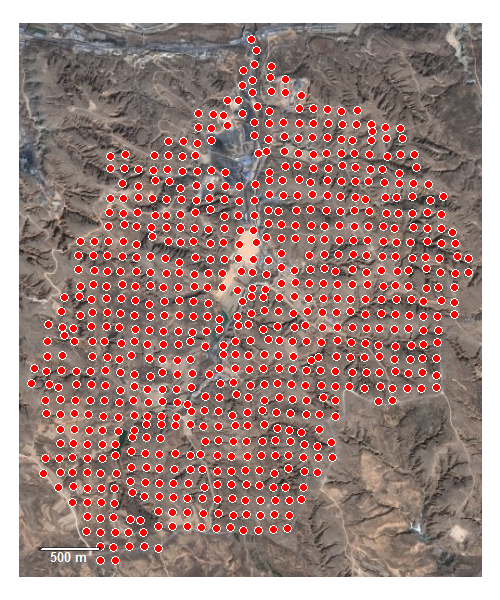}
	\vspace{-5mm}
	\caption{Location map of the 689 sampling sites in the Liudaogou watershed, west of Shenmu City, Shaanxi Province, China.}
	\label{fig:12}
\end{figure}

Since it is typically uncertain which covariates should be included in the model, we consider a set of candidate models with different covariate specifications to improve prediction accuracy. In this study, we require all candidate models to include at least one covariate, resulting in a total of $2^6 - 1 = 63$ candidate models. In each candidate model, we set the spatial distance metric to  q = 2  and adopt the Bisquare kernel function  $K(v) = (1-v^2)^2 I_{\{|v| \leq 1\}}$ (the integral normalization constant required in condition 5 is omitted, which does not affect the estimation results), while bandwidth selection is performed based on the cross-validation criterion. We randomly split the data into a training set and a test set. Let $n_0$ denote the number of observations in the training set, which is set to 400, 500, and 600 in our experiments. Based on the model estimated from the training set, we predict the soil total nitrogen (STN) for the remaining $n_1 = n - n_0$ observations in the corresponding test set. The prediction performance is measured by the mean squared prediction error (MSPE). In addition to the methods considered in the simulation study, we also include the SPMMA method proposed by \cite{LZG2019}. This method captures spatial dependence by incorporating a spatial autocorrelation structure into the error term, but it assumes that the effects of covariates are spatially stationary and thus cannot account for spatial heterogeneity. This procedure is repeated $N = 200$ times, and the mean and median MSPE of each method over all repetitions are reported. Specifically, we define
$$\mathrm{MSPE}_{\mathrm{mean}} = \frac{1}{N} \sum_{j=1}^{N} \mathrm{MSPE}^{(j)}, $$
and
$$\mathrm{MSPE}_{\mathrm{median}} = \mathop{\mathrm{median}}\limits_{j=1,2,\ldots,N} \bigl(\mathrm{MSPE}^{(j)}\bigr),$$
where
$\mathrm{MSPE}^{(j)} = \frac{1}{n_1} \sum_{i=1}^{n_1} \bigl(Y_i^{(j)} - \widehat{\mu}_i^{(j)}\bigr)^2,$ and $\widehat{\mu}_i^{(j)}$ denotes the predicted value of $Y_i^{(j)}$ obtained by a given method in the $j$th trial.

Tables~\ref{table:1} and~\ref{table:2} present the MSPE results of different methods under the candidate models of linear model and spatially varying coefficient model, respectively. The results indicate that the effect of covariates on STN exhibits spatial heterogeneity, and the SVMMA method consistently demonstrates the best predictive performance across different sample size settings. The relatively poorer predictive performance of SAIC and SBIC compared to model selection methods may be attributed to the fact that candidate models with lower predictive ability are still assigned certain weights, thereby compromising overall predictive performance.

\begin{table}
	\centering
	\caption{The MSPE results when the candidate model is the linear model.}\label{table:1}
	\vspace{2mm}
	\begin{tabular}{llccccccc}
		\toprule
		$n_0$ & Method & MMA & JMA & SAIC & SBIC & AIC & BIC & SPMMA \\
		\midrule
		400	& Mean   & 0.3101 & 0.3083 & 0.3091 & 0.3098 & 0.3095 & 0.3104 & 0.3096 \\
		& Median & 0.2765 & 0.2804 & 0.2834 & 0.2818 & 0.2845 & 0.2825 & 0.2789 \\
		500	& Mean   & 0.2972 & 0.2960 & 0.2966 & 0.2972 & 0.2972 & 0.2975 & 0.2966 \\
		& Median & 0.2450 & 0.2460 & 0.2444 & 0.2453 & 0.2465 & 0.2453 & 0.2452 \\
		600	& Mean   & 0.2895 & 0.2885 & 0.2887 & 0.2892 & 0.2896 & 0.2890 & 0.2895 \\
		& Median & 0.2216 & 0.2198 & 0.2189 & 0.2200 & 0.2193 & 0.2202 & 0.2219 \\
		\bottomrule
	\end{tabular}
\end{table}

\begin{table}
	\centering
	\caption{The MSPE results when the candidate model is the SVCM.}\label{table:2}
	\vspace{2mm}
	\begin{tabular}{llcccccc}
		\toprule
		$n_0$ & Method & SVMMA & SAIC & SBIC & AIC & BIC & AICc \\
		\midrule
		400 & Mean   & \textbf{0.3001} & 0.3110 & 0.3103 & 0.3117 & 0.3073 & 0.3082 \\
		& Median & \textbf{0.2763} & 0.2822 & 0.2817 & 0.2959 & 0.2816 & 0.2896 \\
		500 & Mean   & \textbf{0.2853} & 0.2957 & 0.2952 & 0.2916 & 0.2946 & 0.2900 \\
		& Median & \textbf{0.2411} & 0.2509 & 0.2497 & 0.2474 & 0.2445 & 0.2452 \\
		600 & Mean   & \textbf{0.2739} & 0.2852 & 0.2847 & 0.2776 & 0.2820 & 0.2763 \\
		& Median & \textbf{0.2065} & 0.2238 & 0.2234 & 0.2091 & 0.2134 & 0.2091 \\
		\bottomrule
	\end{tabular}
\end{table}

Moreover, we also calculated the mean squared errors for different methods under the full sample, and the corresponding results are: 0.2269 (SVMMA), 0.2467 (SAIC), 0.2464 (SBIC), 0.2325 (AIC), 0.2479 (BIC), and 0.2325 (AICc), indicating that the SVMMA estimator has satisfactory in-sample
performance. Meanwhile, we present the weight allocation results for candidate models obtained using the SVMMA method. For conciseness, only candidate models with weights greater than $10^{-6}$ are listed in Table~\ref{table:3}. The results show that the covariates ClayPC and SandPC are not included in any candidate models presented in Table~\ref{table:3}. Among the selected covariates, the weights are allocated in descending order as follows: SOCgkg, NO3Ngkg, SiltPC, and NH4Ngkg. Compared to single-model selection methods, the SVMMA method integrates information from different models, allowing for more robust handling of model uncertainty. To further elucidate the relationship between each covariate and STN, we present the spatially varying coefficient patterns. These patterns are derived from the estimated coefficient function vector $\widehat{\bm{\beta}}(\mathbf{s}) = \sum_{m=1}^{63} \widetilde{w}_m \widehat{\bm{\beta}}_{(m)}(\mathbf{s})$, obtained using the SVMMA method. Considering that the absolute values of the coefficient functions for covariates ClayPC and SandPC in $\widehat{\bm{\beta}}(\mathbf{s})$ are all less than $10^{-5}$, their effects on the response variable STN are negligible. Therefore, Figure~\ref{fig:13} only presents the spatially varying coefficient distributions of the intercept term and the remaining covariates, intuitively revealing the spatial differentiation patterns of the local regression coefficients for each variable across geographical space.  As shown in the figure, the estimated coefficients of the covariates exhibit distinct spatial heterogeneity in their relationship with STN. Specifically, SOCgkg and SiltPC generally exert a stable positive influence across the study area. In contrast, the effects of NO3Ngkg and NH4Ngkg are spatially variable, shifting between positive and negative values in different subregions. Meanwhile, the intercept term shows a consistently negative influence throughout the entire region.

\begin{table}
	\centering
	\caption{Candidate models with weights exceeding $10^{-6}$ obtained via the SVMMA method}\label{table:3}
	\vspace{2mm}
	\begin{tabular}{cccccccc}
		\toprule
		       & SOCgkg & ClayPC & SiltPC & SandPC & NO3Ngkg & NH4Ngkg & Weight \\
		\midrule
		 Model1& Yes & -- & Yes & -- & Yes & --  & 0.3490 \\
		 Model2& Yes & -- & --  & -- & Yes & --  & 0.2802 \\
		 Model3& Yes & -- & Yes & -- & Yes & Yes & 0.2059 \\
		 Model4& --  & -- & Yes & -- & --  & --  & 0.0835 \\
		 Model5& Yes & -- & --  & -- & --  & --  & 0.0465 \\
		 Model6& --  & -- & Yes & -- & Yes & Yes & 0.0207 \\
		 Model7& --  & -- & --  & -- & Yes & --  & 0.0141 \\
		\bottomrule
	\end{tabular}
\end{table}

\begin{figure}[H]
	\centering
	\includegraphics[height=11cm]{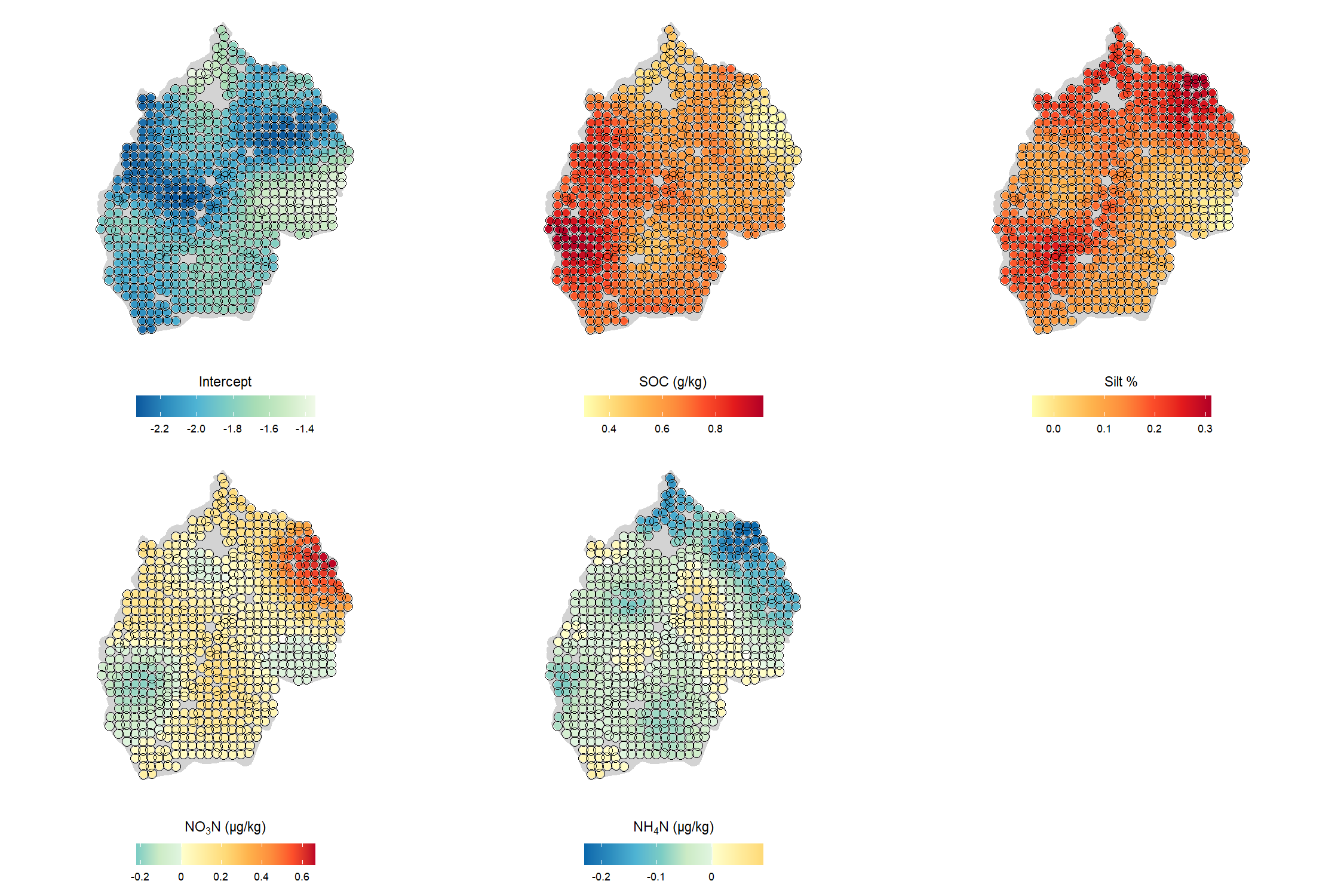}
	\vspace{-5mm}
	\caption{Spatial plots of the coefficient function estimates obtained using the SVMMA method.}
	\label{fig:13}
\end{figure}

\section{Concluding Remarks}

Over the past few decades, model averaging has gained extensive research attention as an effective prediction tool across numerous fields. However, most existing methods have primarily focused on frameworks such as independent and identically distributed samples, single time series, and cross-sectional spatial autoregressive models. To address this gap, this paper introduces a spatially varying coefficient Mallows model averaging method tailored for datasets exhibiting spatial heterogeneity, thereby extending the applicability of model averaging to a broader class of spatial processes. 

Several methodological extensions can be further developed based on our work. First, the proposed framework can be extended by replacing the varying-coefficient model with semi-parametric alternatives, such as partially linear models\citep{BFC1999,MXW2016}, to enhance its flexibility and applicability. Second, the current approach assumes a common spatial smoothing scale (i.e., a single bandwidth parameter) for all covariates. In practice, however, different drivers may operate over distinct spatial scales. Imposing a uniform bandwidth could obscure genuine spatial heterogeneity and introduce estimation bias. Allowing covariate-specific bandwidths\citep{FYK2017} would thus improve estimation accuracy and interpretability. Finally, extending the methodology from a purely spatial to a spatiotemporal setting represents a natural and important frontier. Future work could involve constructing candidate models that incorporate both spatial variation and temporal dynamics\citep{FCY2015,YWWG2022}, thereby enabling the application of the proposed averaging approach to more challenging spatiotemporal datasets.

\section*{Acknowledgments}
This paper is supported by National Natural Science Foundation of China (Grant No. 12571290), Fundamental Research Funds for the Central Universities (Grants No. SWU-KT25018 and No. SWU-KU24002).

\appendix
\section*{Appendix: Proofs of Theorems and Corollaries}
\addcontentsline{toc}{section}{Appendix: Proofs of Theorems 1-3}

\setcounter{equation}{0}
\renewcommand{\theequation}{A.\arabic{equation}}

The proofs of Theorems and Corollaries require the following lemmas.

\begin{lem}\label{lem1}
Suppose that Conditions 3-6 hold. Define, for each $m=1,2,\dots,M$ and $\mathbf{s} = (s_1, s_2)^{\mathrm{T}} \in D$, $\mathbf{J}_n^m=\frac{1}{n}\mathbf{X}_{(m)}^{\mathrm{T}}\mathbf{W}_{(m)\mathbf{s}}\mathbf{X}_{(m)}$. Then the following holds:
\begin{equation}
	\underline{\lambda}(\mathbf{J}_n^m)\geq\underline{c}_\mathbf{X}f(\mathbf{s})-O_{up}\left(\sqrt{\frac{p_m}{nh_m^2}}+h^2_m\sqrt{p_m}\right),
\end{equation}
where the notation $g(\mathbf{s}) = O_{up}(h_m^2)$ means that $g(\mathbf{s})/h_m^2$ is bounded in probability uniformly for any $\mathbf{s}$ within the interior of $D$.
\end{lem}

\noindent\textit{Proof.} Let $\mathbf{u} = (u_1, u_2)^{\mathrm{T}}$ be a random vector with probability density function $f(\mathbf{u})$. Under Conditions 3-6, we have
\begin{align}
	&E\left\{K_{h_m}(\|\mathbf{u}-\mathbf{s}\|_q)\right\}\notag\\
	=&\int_{(u_{1},u_{1}) \in D} \frac{1}{h_{m}^2} K^{\prime}_q \left( \frac{u_{1} - s_{1}}{h_{m}},\frac{u_{2} - s_{2}}{h_{m}} \right) f(u_{1},u_{2}) \, du_{1}du_{2}\notag\\
	=&\int_{(v_1,v_2) \in \operatorname{supp}(K^{\prime}_q)}\frac{1}{h_{m}^2} K^{\prime}_q \left( v_{1},v_{2} \right)f(s_1+h_{m}v_1,s_2+h_{m}v_2) \, h_{m}^2dv_1dv_2\notag\\
	=&\int_{(v_1,v_2) \in \operatorname{supp}(K^{\prime}_q)}K^{\prime}_q \left( v_{1},v_{2} \right)\left\{f(s_1,s_2)+f_x(s_1,s_2)h_{m}v_1+f_y(s_1,s_2)h_{m}v_2\right.\notag\\
	\,&\left.+f_{xx}(s_1^*,s_2^*)h_{m}^2v_1^2+2f_{xy}(s_1^*,s_2^*)h_{m}^2v_1v_2+f_{yy}(s_1^*,s_2^*)h_{m}^2v_2^2\right\}dv_1dv_2\notag\\
	=&f(s_1,s_2)+h_{m}^2\int_{(v_1,v_2) \in \operatorname{supp}(K^{\prime}_q)}K^{\prime}_q \left( v_{1},v_{2} \right)\left(f_{xx}(s_1^*,s_2^*)v_1^2\right.\notag\\
	\,&\left.+2f_{xy}(s_1^*,s_2^*)v_1v_2+f_{yy}(s_1^*,s_2^*)v_2^2\right)dv_1dv_2\vphantom{\int}\notag\\
	=&f(\mathbf{s})+O_{up}(h_{m}^2),
\end{align}
where $(s_1^*, s_2^*)$ lies between $(s_1, s_2)$ and $(s_1 + h_{m} v_1, s_2 + h_{m} v_2)$. This notation persists in what follows, although its precise bounds may differ in different contexts. Similarly,
\begin{align}
	&E\left\{K_{h_m}^2(\|\mathbf{u}-\mathbf{s}\|_q)\right\}\notag\\
	=&\int_{(u_{1},u_{2}) \in D} \frac{1}{h_{m}^4} K^{\prime2}_q \left( \frac{u_{1} - s_{1}}{h_{m}},\frac{u_{2} - s_{2}}{h_{m}} \right) f(u_{1},u_{2}) \, du_{1}du_{2}\notag\\
	=&\int_{(v_1,v_2) \in \operatorname{supp}(K^{\prime}_q)}\frac{1}{h_{m}^2} K^{\prime2}_q \left( v_{1},v_{2} \right)f(s_1+h_{m}v_1,s_2+h_{m}v_2) \, dv_1dv_2\notag\\
	=&\frac{1}{h_{m}^2} \int_{(v_1,v_2) \in \operatorname{supp}(K^{\prime}_q)}K^{\prime2}_q \left( v_{1},v_{2} \right)\left\{f(s_1,s_2)+f_x(s_1,s_2)h_{m}v_1+f_y(s_1,s_2)h_{m}v_2\right.\notag\\
	\,&\left.+f_{xx}(s_1^*,s_2^*)h_{m}^2v_1^2+2f_{xy}(s_1^*,s_2^*)h_{m}^2v_1v_2+f_{yy}(s_1^*,s_2^*)h_{m}^2v_2^2\right\}dv_1dv_2\notag\\
	=&\frac{1}{h_{m}^2} \left\{f(s_1,s_2)R(K^{\prime}_q)+h_{m}^2\int_{(v_1,v_2) \in \operatorname{supp}(K^{\prime}_q)}K^{\prime2}_q\left( v_{1},v_{2} \right)\left(f_{xx}(s_1^*,s_2^*)v_1^2\right.\right.\notag\\
	\,&\left.\left.+2f_{xy}(s_1^*,s_2^*)v_1v_2+f_{yy}(s_1^*,s_2^*)v_2^2\right)dv_1dv_2\vphantom{\int}\right\}\notag\\
	=&\frac{1}{h_{m}^2}f(\mathbf{s})R(K^{\prime}_q)+O_{up}(1),
\end{align}
where $R(K^{\prime}_q) = \int_{(v_1,v_2) \in \operatorname{supp}(K^{\prime}_q)} K^{\prime2}_q(v_1, v_2) \, dv_1 dv_2$. Now, consider $x_{il}^m x_{ik}^m$ and $J_{n,lk}^m$, which are the $(l,k)$th elements of $\mathbf{X}_{(m)i} \mathbf{X}_{(m)i}^{\mathrm{T}}$ and $\mathbf{J}_n^m$, respectively. From Conditions 3-6, (A.3), and Chebyshev's inequality, it follows that
\begin{align}
	&\left|J_{n,lk}^m-E(x_{1l}^mx_{1k}^m)\left\{ \frac{1}{n}\sum_{i=1}^{n} K_{h_m}(\|\mathbf{s}_i-\mathbf{s}\|_q)\right\}\right|\notag\\
	= &\left|\frac{1}{n}\sum_{i=1}^{n} \{x_{il}^mx_{ik}^m-E(x_{1l}^mx_{1k}^m)\} K_{h_m}(\|\mathbf{s}_i-\mathbf{s}\|_q)\right|\notag\\
	=&O_{up}\left(\frac{1}{\sqrt{nh_m^2}}\right).
\end{align}
Combining (A.2), (A.3), (A.4), and Condition 4, we have
\begin{align}
	&|J_{n,lk}^m-E(x_{1l}^mx_{1k}^m)f(\mathbf{s})|\notag\\
	\leq&\left|J_{n,lk}^m-E(x_{1l}^mx_{1k}^m)\left\{ \frac{1}{n}\sum_{i=1}^{n} K_{h_m}(\|\mathbf{s}_i-\mathbf{s}\|_q)\right\}\right|\notag\\
	&+\left|E(x_{1l}^mx_{1k}^m)\left\{ \frac{1}{n}\sum_{i=1}^{n} K_{h_m}(\|\mathbf{s}_i-\mathbf{s}\|_q)\right\}-E(x_{1l}^mx_{1k}^m)E\left\{K_{h_m}(\|\mathbf{u}-\mathbf{s}\|_q)\right\}\right|\notag\\
	&+\left|E(x_{1l}^mx_{1k}^m)E\left\{K_{h_m}(\|\mathbf{u}-\mathbf{s}\|_q)\right\}-E(x_{1l}^mx_{1k}^m)f(\mathbf{s})\right|\notag\\
	=&O_{up}\left(\frac{1}{\sqrt{nh_m^2}}+h^2_m\right).
\end{align}

Let $\|A\|_{F}$ denote the $F$-norm of matrix $A$. According to (A.5), we have
\begin{align}
	&\|\mathbf{J}_n^m-f(\mathbf{s})\bm{\Pi}_{m} \mathbf{C}_\mathbf{X} \bm{\Pi}_{m}^{\mathrm{T}}\|_{F}\notag\\
	=&\left(\sum_{l=1}^{p_m}\sum_{k=1}^{p_m}|J_{n,lk}^m-E(x_{1l}^mx_{1k}^m)f(\mathbf{s})|^2\right)^{1/2}\notag\\
	\leq&O_{up}\left(\frac{p_m}{\sqrt{nh_m^2}}+h^2_mp_m\right).
\end{align}
Then, based on the relationship between the spectral norm and the $F$-norm, we obtain
\begin{align}
	&\underline{\lambda}(\mathbf{J}_n^m)\notag\\
	\geq&\underline{\lambda}(f(\mathbf{s})\bm{\Pi}_{m} \mathbf{C}_\mathbf{X} \bm{\Pi}_{m}^{\mathrm{T}})-\frac{1}{\sqrt{p_m}}\|\mathbf{J}_n^m-f(\mathbf{s})\bm{\Pi}_{m} \mathbf{C}_\mathbf{X} \bm{\Pi}_{m}^{\mathrm{T}}\|_{F}\notag\\
	=&\underline{c}_\mathbf{X}f(\mathbf{s})-O_{up}\left(\sqrt{\frac{p_m}{nh_m^2}}+h^2_m\sqrt{p_m}\right),
\end{align}
which implies that (A.1) holds.

\begin{lem}\label{lem2}
	Suppose that Conditions 3-8 hold. Then the following hold:
	\begin{equation}
		\max_{1 \le m \le M} \left| \operatorname{tr}\left( \mathbf{P}_{(m)} \right) \right| = O_p\left( \underline{h}^{-2} \widetilde{p} \right),
	\end{equation}
	and
	\begin{equation}
		\max_{1 \le m \le M} \bar{\lambda}\left( \mathbf{P}_{(m)} \right) = O_p\left( \widetilde{p}^{1/2} \right).
	\end{equation}
\end{lem}

\noindent\textit{Proof.} Define $\mathbf{e}_i = (0, \dots, 1, \dots, 0)^{\mathrm{T}}$ as an $n \times 1$ vector with 1 at the $i$th position and 0 elsewhere; and $\mathbf{{E}}_i$ as an $n \times n$ matrix with 1 at the $(i,i)$th entry and 0 elsewhere. By Lemma 1, \eqref{eq6}, and Conditions 3-8, we have
\begin{align}
	&\mathop{\max }\limits_{{1 \leq  m \leq  M}}\left| {\operatorname{tr}\left( {\mathbf{P}}_{\left( m\right) }\right) }\right|\notag\\
	= &\mathop{\max }\limits_{{1 \leq  m \leq  M}}\left| {\operatorname{tr}\left( \begin{matrix}  {\mathbf{X}}_{\left( m\right) 1}^{\mathrm{T}} {\left\{  {\mathbf{X}}_{\left( m\right)}^{\mathrm{T}}\mathbf{{W}}_{\left( m\right) \mathbf{s}_{1}}{\mathbf{X}}_{\left( m\right)}\right\}  }^{-1}{\mathbf{X}}_{\left( m\right)}^{\mathrm{T}}\mathbf{{W}}_{\left( m\right) \mathbf{s}_{1}} \\  \vdots \\   {\mathbf{X}}_{\left( m\right) n}^{\mathrm{T}} {\left\{  {\mathbf{X}}_{\left( m\right)}^{\mathrm{T}}\mathbf{{W}}_{\left( m\right) \mathbf{s}_{n}}{\mathbf{X}}_{\left( m\right)}\right\}  }^{-1}{\mathbf{X}}_{\left( m\right) }^{\mathrm{T}}\mathbf{{W}}_{\left( m\right) \mathbf{s}_{n}} \end{matrix}\right) }\right|\notag\\
	= &\mathop{\max }\limits_{{1 \leq  m \leq  M}}\left| {\mathop{\sum }\limits_{{i = 1}}^{n} {\mathbf{X}}_{\left( m\right) i}^{\mathrm{T}} {\left\{  {\mathbf{X}}_{\left( m\right)}^{\mathrm{T}}\mathbf{{W}}_{\left( m\right) \mathbf{s}_{i}}{\mathbf{X}}_{\left( m\right)}\right\}  }^{-1}{\mathbf{X}}_{\left( m\right)}^{\mathrm{T}}\mathbf{{W}}_{\left( m\right) \mathbf{s}_{i}}\mathbf{e}_{i}}\right|\notag\\
	= &\mathop{\max }\limits_{{1 \leq  m \leq  M}}\left| {\mathop{\sum }\limits_{{i = 1}}^{n}\operatorname{tr}\left\{  {{\left( \frac{1}{n}{\mathbf{X}}_{\left( m\right)}^{\mathrm{T}}\mathbf{{W}}_{\left( m\right) \mathbf{s}_{i}}{\mathbf{X}}_{\left( m\right)}\right) }^{-1}\frac{1}{n}{\mathbf{X}}_{\left( m\right)}^{\mathrm{T}}\mathbf{{W}}_{\left( m\right) \mathbf{s}_{i}}\mathbf{{E}}_{i}{\mathbf{X}}_{\left( m\right)}}\right\}  }\right|\notag\\
	=&\mathop{\max }\limits_{{1 \leq  m \leq  M}}\left| {\mathop{\sum }\limits_{{i = 1}}^{n}\operatorname{tr}\left\{ {\left( \frac{1}{n}{\mathbf{X}}_{\left( m\right)}^{\mathrm{T}}\mathbf{{W}}_{\left( m\right) \mathbf{s}_{i}}{\mathbf{X}}_{\left( m\right)}\right) }^{-1} 
		\frac{1}{n}{K}_{{h}_{m}}\left( 0\right) {\mathbf{X}}_{\left( m\right) i}{\mathbf{X}}_{\left( m\right) i}^{\mathrm{T}}\right\} }\right|\notag\\
	\leq & \mathop{\max }\limits_{{1 \leq  m \leq  M}}\frac{1}{n}{K}_{{h}_{m}}\left( 0\right) \mathop{\sum }\limits_{{i = 1}}^{n}
	\left| {{\mathbf{X}}_{\left( m\right) i}^{\mathrm{T}}{\left( \frac{1}{n}{\mathbf{X}}_{\left( m\right)}^{\mathrm{T}}\mathbf{{W}}_{\left( m\right) \mathbf{s}_{i}}{\mathbf{X}}_{\left( m\right)}\right) }^{-1} {\mathbf{X}}_{\left( m\right) i}}\right|\notag\\
	\leq & \mathop{\max }\limits_{{1 \leq  m \leq  M}}\frac{1}{n}{K}_{{h}_{m}}\left( 0\right) \mathop{\sum }\limits_{{i = 1}}^{n}\bar{\lambda }\left\{\left( \frac{1}{n}{\mathbf{X}}_{\left( m\right)}^{\mathrm{T}}\mathbf{{W}}_{\left( m\right) \mathbf{s}_{i}}{\mathbf{X}}_{\left( m\right)}\right)^{-1} \right\} {\mathbf{X}}_{\left( m\right) i}^{\mathrm{T}}{\mathbf{X}}_{\left( m\right) i}\notag\\
	\leq & \mathop{\max }\limits_{{1 \leq  m \leq  M}}\frac{1}{n{h}_{m}^2}K\left( 0\right) \mathop{\sum }\limits_{{i = 1}}^{n} \left(\underline{c}_\mathbf{X}f(\mathbf{s}_i)-O_{up}\left(\sqrt{\frac{p_m}{nh_m^2}}+h^2_m\sqrt{p_m}\right)\right)^{-1} {\mathbf{X}}_{\left( m\right) i}^{\mathrm{T}}{\mathbf{X}}_{\left( m\right) i}\notag\\
	\leq & \frac{1}{n\underline{h}^2}K\left( 0\right) n{O}_{P}(1) {O}_{p}\left( \widetilde{p}\right)\notag\\
	= &{O}_{P}\left( {{\underline{h}}^{-2}\widetilde{p}}\right) ,
\end{align}
which implies that (A.8) holds.

Denote
\begin{equation}
	\mathbf{B} = \operatorname{\mathrm{diag}}{\left( \mathbf{{e}}_{1}^{\mathrm{T}},\mathbf{{e}}_{2}^{\mathrm{T}},\ldots ,\mathbf{{e}}_{n}^{\mathrm{T}}\right) }_{n \times  {n}^{2}},
\end{equation}
\begin{equation}
	\mathbf{{C}}_{\left( m\right) } = \operatorname{\mathrm{diag}}{\left( {\mathbf{X}}_{\left( m\right) },{\mathbf{X}}_{\left( m\right) },\ldots ,{\mathbf{X}}_{\left( m\right)}\right) }_{{n}^{2} \times {p}_{m}n},
\end{equation}
\begin{equation}
	\mathbf{{E}}_{\left( m\right) } = \operatorname{\mathrm{diag}}{\left\{  {\left( \frac{1}{n}{\mathbf{X}}_{\left( m\right) }^{\mathrm{T}}\mathbf{{W}}_{\left( m\right) \mathbf{s}_{1}}{\mathbf{X}}_{\left( m\right)}\right) }^{-1},\ldots ,{\left( \frac{1}{n}{\mathbf{X}}_{\left( m\right) }^{\mathrm{T}}\mathbf{{W}}_{\left( m\right) \mathbf{s}_{n}}{\mathbf{X}}_{\left( m\right) }\right) }^{-1}\right\}  }_{p_mn\times p_mn}
\end{equation}
\begin{equation}
	\mathbf{{F}}_{\left( m\right) } = \operatorname{\mathrm{diag}}{\left( \frac{1}{\sqrt{n}}{\mathbf{X}}_{\left( m\right) }^{\mathrm{T}}\mathbf{{W}}_{\left( m\right) \mathbf{s}_{1}}^{1/2},\frac{1}{\sqrt{n}}{\mathbf{X}}_{\left( m\right) }^{\mathrm{T}}\mathbf{{W}}_{\left( m\right) \mathbf{s}_{2}}^{1/2},\ldots ,\frac{1}{\sqrt{n}}{\mathbf{X}}_{\left( m\right) }^{\mathrm{T}}\mathbf{{W}}_{\left( m\right) \mathbf{s}_{n}}^{1/2}\right) }_{{p}_{m}n \times  {n}^{2}}, 
\end{equation}
and
\begin{equation}
	\mathbf{{H}}_{\left( m\right) } = {\left( \begin{matrix} \frac{1}{\sqrt{n}}\mathbf{{W}}_{\left( m\right) \mathbf{s}_{1}}^{1/2} \\  \vdots \\  \frac{1}{\sqrt{n}}\mathbf{{W}}_{\left( m\right) \mathbf{s}_{n}}^{1/2} \end{matrix}\right) }_{{n}^{2} \times  n}. 
\end{equation}
It is easily seen that
\begin{align}  
	& {\mathbf{P}}_{\left( m\right) } \notag\\   
	= & \left( \begin{matrix}  {\mathbf{X}}_{\left( m\right) 1}^{\mathrm{T}} {\left\{  {\mathbf{X}}_{\left( m\right)}^{\mathrm{T}}\mathbf{{W}}_{\left( m\right) \mathbf{s}_{1}}{\mathbf{X}}_{\left( m\right)}\right\}  }^{-1}{\mathbf{X}}_{\left( m\right)}^{\mathrm{T}}\mathbf{{W}}_{\left( m\right) \mathbf{s}_{1}} \\  \vdots \\   {\mathbf{X}}_{\left( m\right) n}^{\mathrm{T}} {\left\{  {\mathbf{X}}_{\left( m\right)}^{\mathrm{T}}\mathbf{{W}}_{\left( m\right) \mathbf{s}_{n}}{\mathbf{X}}_{\left( m\right)}\right\}  }^{-1}{\mathbf{X}}_{\left( m\right) }^{\mathrm{T}}\mathbf{{W}}_{\left( m\right) \mathbf{s}_{n}} \end{matrix}\right)_{n \times  n}\notag\\
	= &\mathbf{B}\mathbf{{C}}_{\left( m\right) }\mathbf{{E}}_{\left( m\right) }\mathbf{{F}}_{\left( m\right) }\mathbf{{H}}_{\left( m\right) },
\end{align}
and
\begin{align}
	&\mathop{\max }\limits_{{1 \leq  m \leq  M}}\bar{\lambda }\left( {{\mathbf{P}}_{\left( m\right) }{\mathbf{P}}_{\left( m\right) }^{\mathrm{T}}}\right)\notag\\
	= &\mathop{\max }\limits_{{1 \leq  m \leq  M}}\bar{\lambda }\left( {\mathbf{B}\mathbf{{C}}_{\left( m\right) }\mathbf{{E}}_{\left( m\right) }\mathbf{{F}}_{\left( m\right) }\mathbf{{H}}_{\left( m\right) }\mathbf{{H}}_{\left( m\right) }^{\mathrm{T}}\mathbf{{F}}_{\left( m\right) }^{\mathrm{T}}\mathbf{{E}}_{\left( m\right) }^{\mathrm{T}}\mathbf{{C}}_{\left( m\right) }^{\mathrm{T}}{\mathbf{B}}^{\mathrm{T}}}\right)\notag\\
	\leq & \mathop{\max }\limits_{{1 \leq  m \leq  M}}\bar{\lambda }\left( {\mathbf{{H}}_{\left( m\right) }\mathbf{{H}}_{\left( m\right) }^{\mathrm{T}}}\right) \bar{\lambda }\left( {\mathbf{{E}}_{\left( m\right) }\mathbf{{F}}_{\left( m\right) }\mathbf{{F}}_{\left( m\right) }^{\mathrm{T}}\mathbf{{E}}_{\left( m\right) }^{\mathrm{T}}}\right) \bar{\lambda }\left( {\mathbf{B}\mathbf{{C}}_{\left( m\right) }\mathbf{{C}}_{\left( m\right) }^{\mathrm{T}}{\mathbf{B}}^{\mathrm{T}}}\right) .
\end{align}

To prove (A.9), we only need to verify that
\begin{equation}
	\mathop{\max }\limits_{{1 \leq  m \leq  M}}\bar{\lambda }\left( {\mathbf{{H}}_{\left( m\right) }\mathbf{{H}}_{\left( m\right) }^{\mathrm{T}}}\right)  = {O}_{p}\left( 1\right) ,
\end{equation}
\begin{equation}
	\mathop{\max }\limits_{{1 \leq  m \leq  M}}\bar{\lambda }\left( {\mathbf{{E}}_{\left( m\right) }\mathbf{{F}}_{\left( m\right) }\mathbf{{F}}_{\left( m\right) }^{\mathrm{T}}\mathbf{{E}}_{\left( m\right) }^{\mathrm{T}}}\right)  = {O}_{p}\left( 1\right) ,
\end{equation}
and
\begin{equation}
	\mathop{\max }\limits_{{1 \leq  m \leq  M}}\bar{\lambda }\left( {\mathbf{B}\mathbf{{C}}_{\left( m\right) }\mathbf{{C}}_{\left( m\right) }^{\mathrm{T}}{\mathbf{B}}^{\mathrm{T}}}\right)  = {O}_{p}\left( \widetilde{p}\right) .
\end{equation}

Let us first consider (A.18). By (A.15) and Conditions 3, 5-6, we have
\begin{align}
	&\mathop{\max }\limits_{{1 \leq  m \leq  M}}\bar{\lambda }\left( {\mathbf{{H}}_{\left( m\right) }\mathbf{{H}}_{\left( m\right) }^{\mathrm{T}}}\right)\notag\\
	=& \mathop{\max }\limits_{{1 \leq  m \leq  M}}\bar{\lambda }\left( {\mathbf{{H}}_{\left( m\right) }^{\mathrm{T}}\mathbf{{H}}_{\left( m\right) }}\right)\notag\\
	=& \mathop{\max }\limits_{{1 \leq  m \leq  M}}\bar{\lambda }\left( {\frac{1}{n}\mathop{\sum }\limits_{{i = 1}}^{n}\mathbf{{W}}_{\left( m\right) \mathbf{s}_{i}}}\right)\notag\\
	=& \mathop{\max }\limits_{{1 \leq  m \leq  M}}\bar{\lambda }\left( {\frac{1}{n}\mathop{\sum }\limits_{{i = 1}}^{n}\operatorname{\mathrm{diag}}\left\{  {K_{h_m}(\|\mathbf{s}_1-\mathbf{s}_i\|) ,K_{h_m}(\|\mathbf{s}_2-\mathbf{s}_i\|) ,\ldots ,K_{h_m}(\|\mathbf{s}_n-\mathbf{s}_i\|) }\right\}  }\right)\notag\\
	=& \mathop{\max }\limits_{{1 \leq  m \leq  M}}\mathop{\max }\limits_{{1 \leq  j \leq  n}}\frac{1}{n}\mathop{\sum }\limits_{{i = 1}}^{n}K_{h_m}(\|\mathbf{s}_j-\mathbf{s}_i\|)\notag\\
	=&\mathop{\max }\limits_{{1 \leq  m \leq  M}}\mathop{\max }\limits_{1 \leq  j \leq  n}\left(f(\mathbf{s}_j)+O_{up}\left(\frac{1}{\sqrt{nh_m^2}}+ h_{m}^2\right)\right)\notag\\
	=& O_{p}\left( 1\right).
\end{align}
Next, consider (A.18). By (A.13), (A.14), Lemma 1 and Conditions 7-8, we can show that
\begin{align}
	&\mathop{\max }\limits_{{1 \leq  m \leq  M}}\bar{\lambda }\left( {\mathbf{{E}}_{\left( m\right) }\mathbf{{F}}_{\left( m\right) }\mathbf{{F}}_{\left( m\right) }^{\mathrm{T}}\mathbf{{E}}_{\left( m\right) }^{\mathrm{T}}}\right)\notag\\
	= &\mathop{\max }\limits_{{1 \leq  m \leq  M}}\bar{\lambda }\left( \mathbf{{E}}_{\left( m\right) }\right)\notag\\
	= &\mathop{\max }\limits_{{1 \leq  m \leq  M}}\bar{\lambda }\left(\operatorname{\mathrm{diag}}{\left\{  {\left( \frac{1}{n}{\mathbf{X}}_{\left( m\right) }^{\mathrm{T}}\mathbf{{W}}_{\left( m\right) \mathbf{s}_{1}}{\mathbf{X}}_{\left( m\right)}\right) }^{-1},\ldots ,{\left( \frac{1}{n}{\mathbf{X}}_{\left( m\right) }^{\mathrm{T}}\mathbf{{W}}_{\left( m\right) \mathbf{s}_{n}}{\mathbf{X}}_{\left( m\right) }\right) }^{-1}\right\}  } \right)\notag\\
	= &\mathop{\max }\limits_{{1 \leq  m \leq  M}}\mathop{\max }\limits_{{1 \leq  i \leq  n}}\bar{\lambda }\left\{ {\left( \frac{1}{n}{\mathbf{X}}_{\left( m\right) }^{\mathrm{T}}\mathbf{{W}}_{\left( m\right) \mathbf{s}_{i}}{\mathbf{X}}_{\left( m\right)}\right) }^{-1} \right\}\notag\\
	= & \mathop{\max }\limits_{{1 \leq  m \leq  M}}\mathop{\max }\limits_{{1 \leq  i \leq  n}} \left(\underline{c}_\mathbf{X}f(\mathbf{s}_i)-O_{up}\left(\sqrt{\frac{p_m}{nh_m^2}}+h^2_m\sqrt{p_m}\right)\right)^{-1} \notag\\
	= &{O}_{p}\left( 1\right) .
\end{align}
Hence (A.19) is correct. We now consider (A.20). By Condition 4, (A.11) and (A.12), we have
\begin{align}
	&\mathop{\max }\limits_{{1 \leq  m \leq  M}}\bar{\lambda }\left( {\mathbf{B}\mathbf{{C}}_{\left( m\right) }\mathbf{{C}}_{\left( m\right) }^{\mathrm{T}}{\mathbf{B}}^{\mathrm{T}}}\right)\notag\\
	= &\mathop{\max }\limits_{{1 \leq  m \leq  M}}\bar{\lambda }\left\{  {\operatorname{\mathrm{diag}}\left( {\mathbf{{e}}_{1}^{\mathrm{T}},\mathbf{{e}}_{2}^{\mathrm{T}},\ldots ,\mathbf{{e}}_{n}^{\mathrm{T}}}\right) \operatorname{\mathrm{diag}}\left( {{\mathbf{X}}_{\left( m\right) },{\mathbf{X}}_{\left( m\right) },\ldots ,{\mathbf{X}}_{\left( m\right) }}\right) }\right.\notag\\
	&\left. {\operatorname{\mathrm{diag}}\left( {{\mathbf{X}}_{\left( m\right) }^{\mathrm{T}},{\mathbf{X}}_{\left( m\right) }^{\mathrm{T}},\ldots ,{\mathbf{X}}_{\left( m\right) }^{\mathrm{T}}}\right) \operatorname{\mathrm{diag}}\left( {\mathbf{{e}}_{1},\mathbf{{e}}_{2},\ldots ,\mathbf{{e}}_{n}}\right) }\right\}\notag\\
	= &\mathop{\max }\limits_{{1 \leq  m \leq  M}}\bar{\lambda }\left\{  {\operatorname{\mathrm{diag}}\left( {\mathbf{{e}}_{1}^{\mathrm{T}}{\mathbf{X}}_{\left( m\right) }{\mathbf{X}}_{\left( m\right) }^{\mathrm{T}}\mathbf{{e}}_{1},\mathbf{{e}}_{2}^{\mathrm{T}}{\mathbf{X}}_{\left( m\right) }{\mathbf{X}}_{\left( m\right) 2}^{\mathrm{T}}\mathbf{{e}}_{2},\ldots ,\mathbf{{e}}_{n}^{\mathrm{T}}{\mathbf{X}}_{\left( m\right) }{\mathbf{X}}_{\left( m\right) }^{\mathrm{T}}\mathbf{{e}}_{n}}\right) }\right\}\notag\\
	= &\mathop{\max }\limits_{{1 \leq  m \leq  M}}\bar{\lambda }\left\{  {\operatorname{\mathrm{diag}}\left( {{\mathbf{X}}_{\left( m\right) 1}^{\mathrm{T}}{\mathbf{X}}_{\left( m\right) 1},{\mathbf{X}}_{\left( m\right) 2}^{\mathrm{T}}{\mathbf{X}}_{\left( m\right) 2},\ldots ,{\mathbf{X}}_{\left( m\right) n}^{\mathrm{T}}{\mathbf{X}}_{\left( m\right) n}}\right) }\right\}\notag\\
	= &\mathop{\max }\limits_{{1 \leq  m \leq  M}}\mathop{\max }\limits_{{1 \leq  i \leq  n}}{\mathbf{X}}_{\left( m\right) i}^{\mathrm{T}}{\mathbf{X}}_{\left( m\right) i}\notag\\
	= &{O}_{p}\left( \widetilde{p}\right).
\end{align}

By (A.17)-(A.20), we obtain
\begin{align}
	&\mathop{\max }\limits_{{1 \leq  m \leq  M}}\bar{\lambda }\left( {\mathbf{P}}_{\left( m\right) }\right)\notag\\
	= &\mathop{\max }\limits_{{1 \leq  m \leq  M}}{\bar{\lambda }}^{1/2}\left( {{\mathbf{P}}_{\left( m\right) }{\mathbf{P}}_{\left( m\right) }^{\mathrm{T}}}\right)\notag\\
	= &\mathop{\max }\limits_{{1 \leq  m \leq  M}}{\bar{\lambda }}^{1/2}\left( {\mathbf{B}\mathbf{{C}}_{\left( m\right) }\mathbf{{E}}_{\left( m\right) }\mathbf{{F}}_{\left( m\right) }\mathbf{{H}}_{\left( m\right) }\mathbf{{H}}_{\left( m\right) }^{\mathrm{T}}\mathbf{{F}}_{\left( m\right) }^{\mathrm{T}}\mathbf{{E}}_{\left( m\right) }^{\mathrm{T}}\mathbf{{C}}_{\left( m\right) }^{\mathrm{T}}{\mathbf{B}}^{\mathrm{T}}}\right)\notag\\
	\leq & \mathop{\max }\limits_{{1 \leq  m \leq  M}}{\bar{\lambda }}^{1/2}\left( {\mathbf{{H}}_{\left( m\right) }\mathbf{{H}}_{\left( m\right) }^{\mathrm{T}}}\right) {\bar{\lambda }}^{1/2}\left( {\mathbf{{E}}_{\left( m\right) }\mathbf{{F}}_{\left( m\right) }\mathbf{{F}}_{\left( m\right) }^{\mathrm{T}}\mathbf{{E}}_{\left( m\right) }^{\mathrm{T}}}\right) {\bar{\lambda }}^{1/2}\left( {\mathbf{B}\mathbf{{C}}_{\left( m\right) }\mathbf{{C}}_{\left( m\right) }^{\mathrm{T}}{\mathbf{B}}^{\mathrm{T}}}\right)\notag\\
	\leq & \mathop{\max }\limits_{{1 \leq  m \leq  M}}{\bar{\lambda }}^{1/2}\left( {\mathbf{{H}}_{\left( m\right) }\mathbf{{H}}_{\left( m\right) }^{\mathrm{T}}}\right) \mathop{\max }\limits_{{1 \leq  m \leq  M}}{\bar{\lambda }}^{1/2}\left( {\mathbf{{E}}_{\left( m\right) }\mathbf{{F}}_{\left( m\right) }\mathbf{{F}}_{\left( m\right) }^{\mathrm{T}}\mathbf{{E}}_{\left( m\right) }^{\mathrm{T}}}\right) \mathop{\max }\limits_{{1 \leq  m \leq  M}}{\bar{\lambda }}^{1/2}\left( {\mathbf{B}\mathbf{{C}}_{\left( m\right) }\mathbf{{C}}_{\left( m\right) }^{\mathrm{T}}{\mathbf{B}}^{\mathrm{T}}}\right)\notag\\
	= &{O}_{p}\left( {\widetilde{p}}^{1/2}\right) ,
\end{align}
which implies that (A.9) is true.

\begin{lem}\label{lem3}
	Under Conditions 3-6 and 9-11, for the $m$th $(1 \leq m \leq M)$ candidate model, we have
	\begin{equation}
		\|\widehat{\bm{\beta}}_{(m)}(\mathbf{s}) - \bm{\beta}^{*}_{(m)}(\mathbf{s})\|^2 = O_{up}\left( \frac{p_m^3}{n}  + \frac{p_m}{n h_m^2}+ p_m h^4_m \right),
	\end{equation}
	where the notation $g(\mathbf{s}) = O_{up}(h_m^2)$ means that $g(\mathbf{s})/h_m^2$ is bounded in probability uniformly for any $\mathbf{s}$ within the interior of $D$.
\end{lem}

\noindent\textit{Proof.} Let $\mathbf{u} = (u_1, u_2)^{\mathrm{T}}$ be a random vector with probability density function $f(\mathbf{u})$. Then, under Conditions 3-6 and 9, it follows that
\begin{align}
	&E\left\{K_{h_m}(\|\mathbf{u}-\mathbf{s}\|_q)(\beta^{*}_{(m)k}(\mathbf{u})-\beta^{*}_{(m)k}(\mathbf{s}))\right\}\notag\\
	=&\int_{(v_1,v_2) \in \operatorname{supp}(K^{\prime}_q)}K^{\prime}_q \left( v_{1},v_{2} \right)\left\{f(s_1,s_2)+f_x(s_1,s_2)h_{m}v_1+f_y(s_1,s_2)h_{m}v_2\right.\notag\\
	\,&\left.+f_{xx}(s_1^*,s_2^*)h_{m}^2v_1^2+2f_{xy}(s_1^*,s_2^*)h_{m}^2v_1v_2+f_{yy}(s_1^*,s_2^*)h_{m}^2v_2^2\right\}\notag\\
	\,&\left\{\beta^{*}_{(m)k,_x}(s_1,s_2)h_{m}v_1+\beta^{*}_{(m)k,_y}(s_1,s_2)h_{m}v_2+\beta^{*}_{(m)k,_{xx}}(s_3^*,s_4^*)h_{m}^2v_1^2\right.\notag\\
	\,&\left.+2\beta^{*}_{(m)k,_{xy}}(s_3^*,s_4^*)h_{m}^2v_1v_2+\beta^{*}_{(m)k,_{yy}}(s_3^*,s_4^*)h_{m}^2v_2^2\right\}dv_1dv_2\notag\\
	=&O_{up}(h_{m}^2),
\end{align}
where $(s_3^*, s_4^*)$ is an intermediate point between $(s_1, s_2)$ and $(s_1 + h_m v_1, s_2 + h_m v_2)$. The notation $(s_3^*, s_4^*)$ is retained in the subsequent formulae, even though its specific value may not be uniform throughout. Similarly,
\begin{align}
	&E\left\{K_{h_m}^2(\|\mathbf{u}-\mathbf{s}\|_q)(\beta^{*}_{(m)k}(\mathbf{u})-\beta^{*}_{(m)k}(\mathbf{s}))^2\right\}\notag\\
	=&\frac{1}{h_m^2}\int_{(v_1,v_2) \in \operatorname{supp}(K^{\prime}_q)}K^{\prime2}_q \left( v_{1},v_{2} \right)\left\{f(s_1,s_2)+f_x(s_1,s_2)h_{m}v_1+f_y(s_1,s_2)h_{m}v_2\right.\notag\\
	&\left.+f_{xx}(s_1^*,s_2^*)h_{m}^2v_1^2+2f_{xy}(s_1^*,s_2^*)h_{m}^2v_1v_2+f_{yy}(s_1^*,s_2^*)h_{m}^2v_2^2\right\}\notag\\
	&\left\{\beta^{*}_{(m)k,_{x}}(s_1,s_2)h_{m}v_1+\beta^{*}_{(m)k,_{y}}(s_1,s_2)h_{m}v_2+\beta^{*}_{(m)k,_{xx}}(s_3^*,s_4^*)h_{m}^2v_1^2\right.\notag\\
	\,&\left.+2\beta^{*}_{(m)k,_{xy}}(s_3^*,s_4^*)h_{m}^2v_1v_2+\beta^{*}_{(m)k,_{yy}}(s_3^*,s_4^*)h_{m}^2v_2^2\right\}^2dv_1dv_2\notag\\
	=&O_{up}(1).
\end{align}

Let $\mathbf{b}^{*}_{(m)}=(b^{*}_{(m)1},b^{*}_{(m)2},\ldots,b^{*}_{(m)n})^{\mathrm{T}}=\bm{\mu}-\bm{\mu}^{*}_{(m)}$. Note that
\begin{align}
	\widehat{\bm{\beta}}_{(m)}(\mathbf{s})&=(\mathbf{X}_{(m)}^{\mathrm{T}}\mathbf{W}_{(m)\mathbf{s}}\mathbf{X}_{(m)})^{-1}\mathbf{X}_{(m)}^{\mathrm{T}}\mathbf{W}_{(m)\mathbf{s}}\mathbf{Y}\notag\\
	&=\bm{\beta}^{*}_{(m)}(\mathbf{s})+(\mathbf{X}_{(m)}^{\mathrm{T}}\mathbf{W}_{(m)\mathbf{s}}\mathbf{X}_{(m)})^{-1}\mathbf{X}_{(m)}^{\mathrm{T}}\mathbf{W}_{(m)\mathbf{s}}(\bm{\mu}^{*}_{(m)}+\mathbf{b}^{*}_{(m)}-\mathbf{X}_{(m)}\bm{\beta}^{*}_{(m)}(\mathbf{s}))\notag\\
	&\quad+(\mathbf{X}_{(m)}^{\mathrm{T}}\mathbf{W}_{(m)\mathbf{s}}\mathbf{X}_{(m)})^{-1}\mathbf{X}_{(m)}^{\mathrm{T}}\mathbf{W}_{(m)\mathbf{s}}\bm{\epsilon}\notag\\
	&=\bm{\beta}^{*}_{(m)}(\mathbf{s})+\mathbf{B}_{m1}(\mathbf{s})+\mathbf{B}_{m2}(\mathbf{s})+\mathbf{B}_{m3}(\mathbf{s}),
\end{align}
where
\begin{equation}
	\mathbf{B}_{m1}(\mathbf{s})=(\mathbf{X}_{(m)}^{\mathrm{T}}\mathbf{W}_{(m)\mathbf{s}}\mathbf{X}_{(m)})^{-1}\mathbf{X}_{(m)}^{\mathrm{T}}\mathbf{W}_{(m)\mathbf{s}}(\bm{\mu}^{*}_{(m)}-\mathbf{X}_{(m)}\bm{\beta}^{*}_{(m)}(\mathbf{s})),
\end{equation}
\begin{equation}
	\mathbf{B}_{m2}(\mathbf{s})=(\mathbf{X}_{(m)}^{\mathrm{T}}\mathbf{W}_{(m)\mathbf{s}}\mathbf{X}_{(m)})^{-1}\mathbf{X}_{(m)}^{\mathrm{T}}\mathbf{W}_{(m)\mathbf{s}}\mathbf{b}^{*}_{(m)},
\end{equation}
and
\begin{equation}
	\mathbf{B}_{m3}(\mathbf{s})=(\mathbf{X}_{(m)}^{\mathrm{T}}\mathbf{W}_{(m)\mathbf{s}}\mathbf{X}_{(m)})^{-1}\mathbf{X}_{(m)}^{\mathrm{T}}\mathbf{W}_{(m)\mathbf{s}}\bm{\epsilon}.
\end{equation}

For $\mathbf{B}_{m1}(\mathbf{s})$, note that
\begin{align}
	&\frac{1}{n}\mathbf{X}_{(m)}^{\mathrm{T}}\mathbf{W}_{(m)\mathbf{s}}(\bm{\mu}^{*}_{(m)}-\mathbf{X}_{(m)}\bm{\beta}^{*}_{(m)}(\mathbf{s}))\notag\\
	=&\frac{1}{n}\sum_{i=1}^{n}K_{h_m}(\|\mathbf{s}_i-\mathbf{s}\|_q){\mathbf{X}}_{\left( m\right) i}{\mathbf{X}}_{\left( m\right) i}^{\mathrm{T}}(\bm{\beta}^{*}_{(m)}(\mathbf{s}_i)-\bm{\beta}^{*}_{(m)}(\mathbf{s}))\notag\\
	=&\frac{1}{n}\sum_{i=1}^{n}\left\{{\mathbf{X}}_{\left( m\right) i}{\mathbf{X}}_{\left( m\right) i}^{\mathrm{T}}-E({\mathbf{X}}_{\left( m\right) 1}{\mathbf{X}}_{\left( m\right) 1}^{\mathrm{T}})\right\}\mathbf{Z}_{(m)i}\notag\\
	&+\frac{1}{n}\sum_{i=1}^{n}E({\mathbf{X}}_{\left( m\right) 1}{\mathbf{X}}_{\left( m\right) 1}^{\mathrm{T}})\{\mathbf{Z}_{(m)i}-E(\mathbf{Z}_{(m)1})\}+E({\mathbf{X}}_{\left( m\right) 1}{\mathbf{X}}_{\left( m\right) 1}^{\mathrm{T}})E(\mathbf{Z}_{(m)1})\notag\\
	=&\mathbf{B}_{m4}(\mathbf{s})+\mathbf{B}_{m5}(\mathbf{s})+\mathbf{B}_{m6}(\mathbf{s}).
\end{align}
where 
\begin{equation}
	\mathbf{Z}_{(m)i}=K_{h_m}(\|\mathbf{s}_i-\mathbf{s}\|_q)(\bm{\beta}^{*}_{(m)}(\mathbf{s}_i)-\bm{\beta}^{*}_{(m)}(\mathbf{s})),
\end{equation} 
\begin{equation}
	\mathbf{B}_{m4}(\mathbf{s})=\frac{1}{n}\sum_{i=1}^{n}\left\{{\mathbf{X}}_{\left( m\right) i}{\mathbf{X}}_{\left( m\right) i}^{\mathrm{T}}-E({\mathbf{X}}_{\left( m\right) 1}{\mathbf{X}}_{\left( m\right) 1}^{\mathrm{T}})\right\}\mathbf{Z}_{(m)i},
\end{equation}
\begin{equation}
	\mathbf{B}_{m5}(\mathbf{s})=\frac{1}{n}\sum_{i=1}^{n}E({\mathbf{X}}_{\left( m\right) 1}{\mathbf{X}}_{\left( m\right) 1}^{\mathrm{T}})\{\mathbf{Z}_{(m)i}-E(\mathbf{Z}_{(m)1})\},
\end{equation}
and
\begin{equation}
	\mathbf{B}_{m6}(\mathbf{s})=E({\mathbf{X}}_{\left( m\right) 1}{\mathbf{X}}_{\left( m\right) 1}^{\mathrm{T}})E(\mathbf{Z}_{(m)1}).
\end{equation}
Let $\|A\|_{F}$ denote the $F$-norm of matrix $A$. For $\mathbf{B}_{m4}(\mathbf{s})$, according to Condition 4, we have
\begin{align}
	E\|{\mathbf{X}}_{\left( m\right) 1}{\mathbf{X}}_{\left( m\right) 1}^{\mathrm{T}}-E({\mathbf{X}}_{\left( m\right) 1}{\mathbf{X}}_{\left( m\right) 1}^{\mathrm{T}})\|_F^2
	=&E\left(\sum_{l=1}^{p_m}\sum_{k=1}^{p_m}(x^m_{1l}x^m_{1k}-E(x^m_{1l}x^m_{1k}))^2\right)\notag\\
	=&O(p^2_m).
\end{align}
Furthermore, we can obtain
\begin{align}
	&E\|\mathbf{B}_{m4}(\mathbf{s})\|^2\notag\\
	=&\frac{1}{n}E\{\mathbf{Z}_{(m)1}^{\mathrm{T}}\left\{{\mathbf{X}}_{\left( m\right) 1}{\mathbf{X}}_{\left( m\right) 1}^{\mathrm{T}}-E({\mathbf{X}}_{\left( m\right) 1}{\mathbf{X}}_{\left( m\right) 1}^{\mathrm{T}})\right\}^{\mathrm{T}}\left\{{\mathbf{X}}_{\left( m\right) 1}{\mathbf{X}}_{\left( m\right) 1}^{\mathrm{T}}-E({\mathbf{X}}_{\left( m\right) 1}{\mathbf{X}}_{\left( m\right) 1}^{\mathrm{T}})\right\}\mathbf{Z}_{(m)1}\}\notag\\
	\leq&\frac{1}{n}E\|{\mathbf{X}}_{\left( m\right) 1}{\mathbf{X}}_{\left( m\right) 1}^{\mathrm{T}}-E({\mathbf{X}}_{\left( m\right) 1}{\mathbf{X}}_{\left( m\right) 1}^{\mathrm{T}})\|_F^2E\{\mathbf{Z}_{(m)1}^{\mathrm{T}}\mathbf{Z}_{(m)1}\}\notag\\
	=&O\left(\frac{p_m^3}{n}\right).
\end{align}
For $\mathbf{B}_{m5}(\mathbf{s})$, combining Condition 4 and (A.27), we have
\begin{align}
	&E\|\mathbf{B}_{m5}(\mathbf{s})\|^2\notag\\
	=&\frac{1}{n}E\{(\mathbf{Z}_{(m)1}-E(\mathbf{Z}_{(m)1}))^{\mathrm{T}}E({\mathbf{X}}_{\left( m\right) 1}{\mathbf{X}}_{\left( m\right) 1}^{\mathrm{T}})E({\mathbf{X}}_{\left( m\right) 1}{\mathbf{X}}_{\left( m\right) 1}^{\mathrm{T}})(\mathbf{Z}_{(m)1}-E(\mathbf{Z}_{(m)1}))\}\notag\\
	\leq&\frac{\bar{c}_\mathbf{X}^2}{n}E\{(\mathbf{Z}_{(m)1}-E(\mathbf{Z}_{(m)1}))^{\mathrm{T}}(\mathbf{Z}_{(m)1}-E(\mathbf{Z}_{(m)1}))\}\notag\\
	=&O\left(\frac{p_m}{n}\right).
\end{align}
Meanwhile, for $\mathbf{B}_{m6}(\mathbf{s})$, according to Condition 4 and (A.26), we obtain
\begin{equation}
	\|\mathbf{B}_{m6}(\mathbf{s})\|^2\leq\bar{c}_\mathbf{X}^2\|E(\mathbf{Z}_{(m)1})\|^2=O(p_mh_m^4).
\end{equation}
Combining Lemma 1 with (A.38)-(A.40), we have
\begin{equation}
	\|\mathbf{B}_{m1}(\mathbf{s})\|^2 = O_{up}\left(\frac{p_m^3}{n}+p_mh_m^4\right).
\end{equation}

For $\mathbf{B}_{m2}(\mathbf{s})$, we note that $\mathbf{B}_{m2}(\mathbf{s}) = \bm{0}$ when $1 \leq m \leq M_0$, whereas for $M_0+1 \leq m \leq M$,
\begin{align}
	&\frac{1}{n}\mathbf{X}_{(m)}^{\mathrm{T}}\mathbf{W}_{(m)\mathbf{s}}\mathbf{b}^{*}_{(m)}\notag\\
	=&\frac{1}{n}\sum_{i=1}^{n}K_{h_m}(\|\mathbf{s}_i-\mathbf{s}\|_q)\mathbf{X}_{(m)i}{b}^{*}_{(m)i}.
\end{align}
Taking the $j$th row element in (A.42) and applying the definition of $\mathbf{b}^{*}_{(m)}$, we have
\begin{equation}
	E\left\{\frac{1}{n}\sum_{i=1}^{n}K_{h_m}(\|\mathbf{s}_i-\mathbf{s}\|_q)x_{ij}^m{b}^{*}_{(m)i}\right\}=0,
\end{equation}
Meanwhile, from (A.3), Condition 4, Condition 10, and the Cauchy–Schwarz inequality, it follows that
\begin{align}
	&\mathrm{var}\left\{\frac{1}{n}\sum_{i=1}^{n}K_{h_m}(\|\mathbf{s}_i-\mathbf{s}\|_q)x_{ij}^m{b}^{*}_{(m)i}\right\} \notag\\
	=&\frac{1}{n}E\left\{K_{h_m}^2(\|\mathbf{s}_1-\mathbf{s}\|_q)x_{1j}^{m2}{b}^{*2}_{(m)1}\right\} \notag\\
	\leq&\frac{1}{n}[E\{K_{h_m}^2(\|\mathbf{s}_1-\mathbf{s}\|_q)x_{1j}^{m4}\}]^{1/2}[E\{K_{h_m}^2(\|\mathbf{s}_1-\mathbf{s}\|_q)E({b}^{*4}_{(m)1}|\mathbf{s}_1)\}]^{1/2}\notag\\
	=& O_{up}\left(\frac{1}{nh^2_m}\right).
\end{align}
Based on (A.44) and Lemma 1, we have $\left\|\mathbf{B}_{m2}(\mathbf{s})\right\|^2 = O_{up}\left(\frac{p_m}{nh^2_m}\right)$.

For $\mathbf{B}_{m3}(\mathbf{s})$, note that
\begin{align}
	&\frac{1}{n}\mathbf{X}_{(m)}^{\mathrm{T}}\mathbf{W}_{(m)\mathbf{s}}\bm{\epsilon}\notag\\
	=&\frac{1}{n}\sum_{i=1}^{n}K_{h_m}(\|\mathbf{s}_i-\mathbf{s}\|_q)\mathbf{X}_{(m)i}\epsilon_{i}.
\end{align}
Taking the $j$th row element in (A.45), we have
\begin{equation}
	E\left\{\frac{1}{n}\sum_{i=1}^{n}K_{h_m}(\|\mathbf{s}_i-\mathbf{s}\|_q)x_{ij}^m\epsilon_{i}\right\}=0,
\end{equation}
Meanwhile, from (A.3) and Condition 4, it follows that
\begin{align}
	&\mathrm{var}\left\{\frac{1}{n}\sum_{i=1}^{n}K_{h_m}(\|\mathbf{s}_i-\mathbf{s}\|_q)x_{ij}^m\epsilon_{i}\right\}\notag\\
	=&E\left[\mathrm{var}\left\{\frac{1}{n}\sum_{i=1}^{n}K_{h_m}(\|\mathbf{s}_i-\mathbf{s}\|_q)x_{ij}^m\epsilon_{i}\Big|\mathbf{X},\bm{\Psi}\right\}\right]\notag\\
	&+\mathrm{var}\left[E\left\{\frac{1}{n}\sum_{i=1}^{n}K_{h_m}(\|\mathbf{s}_i-\mathbf{s}\|_q)x_{ij}^m\epsilon_{i}\Big|\mathbf{X},\bm{\Psi} \right\}\right]\notag\\
	=&\frac{\sigma^2}{n}E\left\{\frac{1}{n}\sum_{i=1}^{n}K_{h_m}^2(\|\mathbf{s}_i-\mathbf{s}\|_q)x^{m2}_{ij}\right\}\notag\\
	=&\frac{\sigma^2}{n}E\left\{x^{m2}_{1j}\right\}E\left\{K_{h_m}^2(\|\mathbf{u}-\mathbf{s}\|_q)\right\}\notag\\
	=&O_{up}(\frac{1}{nh_{m}^2}).
\end{align}
From Condition 4, (A.46), (A.47), and Chebyshev's inequality, it follows that
\begin{equation}
	\left|\frac{1}{n}\sum_{i=1}^{n}K_{h_m}(\|\mathbf{s}_i-\mathbf{s}\|_q)x_{ij}^m\epsilon_{i}\right|=O_{up}(\frac{1}{\sqrt{nh^2_m}}),
\end{equation}
Based on (A.48) and Lemma 1, we have $\left\|\mathbf{B}_{m3}(\mathbf{s})\right\|^2 = O_{up}\left( \frac{p_m}{n h_{m}^2} \right)$. By the triangle inequality, for the $m$th $(1 \leq m \leq M)$ candidate model, we have
\begin{equation}
	\|\widehat{\bm{\beta}}_{(m)}(\mathbf{s}) - \bm{\beta}^{*}_{(m)}(\mathbf{s})\|^2 = O_{up}\left( \frac{p_m^3}{n} + \frac{p_m}{n h_m^2} + p_m h^4_m \right).
\end{equation}
Lemma 3 has been proved.

\noindent\textbf{Proof of Theorem 1}

Note that the \({C}_{n}\left( \mathbf{w}\right)\) is
\begin{align}
	{C}_{n}\left( \mathbf{w}\right)  &= \parallel \mathbf{Y} - \widehat{\bm{\mu} }\left( \mathbf{w}\right) {\parallel }^{2} + 2\operatorname{tr}\left( {\mathbf{P}\left( \mathbf{w}\right) \bm{\Omega} }\right)\notag\\
	&= {L}_{n}\left( \mathbf{w}\right)  + \parallel \bm{\epsilon} {\parallel }^{2} - 2{\bm{\epsilon} }^{\mathrm{T}}\left( {\mathbf{P}\left( \mathbf{w}\right)  - \mathbf{{I}}_{n}}\right) \bm{\mu}  - 2\left\{  {{\bm{\epsilon} }^{\mathrm{T}}\mathbf{P}\left( \mathbf{w}\right) \bm{\epsilon}  - \operatorname{tr}\left( {\mathbf{P}\left( \mathbf{w}\right) \bm{\Omega} }\right) }\right\}  , 
\end{align}
where \(\parallel \bm{\epsilon} {\parallel }^{2}\)
is independent of \(\mathbf{w}\). Hence to prove Theorem 1, we need only to verify that
\begin{equation}
	\mathop{\sup }\limits_{{\mathbf{w} \in  {\mathcal{H}}_{n}}}\left| \frac{{\bm{\epsilon} }^{\mathrm{T}}\left( {\mathbf{P}\left( \mathbf{w}\right)  - \mathbf{{I}}_{n}}\right) \bm{\mu} }{{R}_{n}\left( \mathbf{w}\right) }\right| \overset{p}{ \rightarrow  }0, 
\end{equation}
\begin{equation}
	\mathop{\sup }\limits_{{\mathbf{w} \in  {\mathcal{H}}_{n}}}\left| \frac{{\bm{\epsilon} }^{\mathrm{T}}\mathbf{P}\left( \mathbf{w}\right) \bm{\epsilon}  - \operatorname{tr}\left( {\mathbf{P}\left( \mathbf{w}\right) \bm{\Omega} }\right) }{{R}_{n}\left( \mathbf{w}\right) }\right| \overset{p}{ \rightarrow  }0,
\end{equation}
and
\begin{equation}
	\mathop{\sup }\limits_{{\mathbf{w} \in  {\mathcal{H}}_{n}}}\left| {\frac{{L}_{n}\left( \mathbf{w}\right) }{{R}_{n}\left( \mathbf{w}\right) } - 1}\right| \overset{p}{ \rightarrow  }0.
\end{equation}

We observe that for any \(\delta  > 0\),
\begin{align}
	&P\left( {\mathop{\sup }\limits_{{\mathbf{w} \in  {\mathcal{H}}_{n}}}\left| \frac{{\bm{\epsilon} }^{\mathrm{T}}\left( {\mathbf{P}\left( \mathbf{w}\right)  - \mathbf{{I}}_{n}}\right) \bm{\mu} }{{R}_{n}\left( \mathbf{w}\right) }\right|  > \delta |\mathbf{X},\bm{\Psi} }\right)\notag\\
	\leq & P\left( {\mathop{\sup }\limits_{{\mathbf{w} \in  {\mathcal{H}}_{n}}}\mathop{\sum }\limits_{{m = 1}}^{M}{w}_{m}\left| {{\bm{\epsilon} }^{\mathrm{T}}\left( {{\mathbf{P}}_{\left( m\right) } - \mathbf{{I}}_{n}}\right) \bm{\mu} }\right|  > \delta {\xi }_{n} |\mathbf{X},\bm{\Psi} }\right)\notag\\
	=& P\left( {\mathop{\max }\limits_{{1 \leq  m \leq  M}}\left| {{\bm{\epsilon} }^{\mathrm{T}}\left( {{\mathbf{P}}_{\left( m\right) } - \mathbf{{I}}_{n}}\right) \bm{\mu} }\right|  > \delta {\xi }_{n} |\mathbf{X},\bm{\Psi} }\right)\notag\\
	=& P\left( {\left\{  {\left| {{\bm{\epsilon} }^{\mathrm{T}}\left( {\mathbf{P}\left( {\mathbf{w}}_{1}^{0}\right)  - \mathbf{{I}}_{n}}\right) \bm{\mu} }\right|  > \delta {\xi }_{n}}\right\}   \cup  \left\{  {\left| {{\bm{\epsilon} }^{\mathrm{T}}\left( {\mathbf{P}\left( {\mathbf{w}}_{2}^{0}\right)  - \mathbf{{I}}_{n}}\right) \bm{\mu} }\right|  > \delta {\xi }_{n}}\right\}  }\right.\notag\\
	&\cup  \ldots  \cup  \left\{  {\left| {{\bm{\epsilon} }^{\mathrm{T}}\left( {\mathbf{P}\left( {\mathbf{w}}_{M}^{0}\right)  - \mathbf{{I}}_{n}}\right) \bm{\mu} }\right|  > \delta {\xi }_{n}}\right\} |\mathbf{X},\bm{\Psi} )\notag\\
	\leq &  \mathop{\sum }\limits_{{m = 1}}^{M}P\left( {\left\{  {\left| {{\bm{\epsilon} }^{\mathrm{T}}\left( {\mathbf{P}\left( {\mathbf{w}}_{m}^{0}\right)  - \mathbf{{I}}_{n}}\right) \bm{\mu} }\right|  > \delta {\xi }_{n}}\right\} |\mathbf{X},\bm{\Psi}}\right)\notag\\
	\leq & \mathop{\sum }\limits_{{m = 1}}^{M}\frac{E\left( {{\left| {\bm{\epsilon} }^{\mathrm{T}}\left( \mathbf{P}\left( {\mathbf{w}}_{m}^{0}\right)  - \mathbf{{I}}_{n}\right) \bm{\mu} \right| }^{2G} |\mathbf{X},\bm{\Psi} }\right) }{{\delta }^{2G}{\xi }_{n}^{2G}}\notag\\
	\leq & \mathop{\sum }\limits_{{m = 1}}^{M}\frac{{c}_{1}{\begin{Vmatrix}\left( \mathbf{P}\left( {\mathbf{w}}_{m}^{0}\right)  - \mathbf{{I}}_{n}\right) \bm{\mu} \end{Vmatrix}}^{2G}}{{\delta }^{2G}{\xi }_{n}^{2G}}\notag\\
	\leq & \mathop{\sum }\limits_{{m = 1}}^{M}\frac{{c}_{1}{\left\{  {R}_{n}\left( {\mathbf{w}}_{m}^{0}\right) \right\}  }^{G}}{{\delta }^{2G}{\xi }_{n}^{2G}}\notag\\
	= &\frac{{c}_{1}}{{\delta }^{2G}} \cdot  \frac{\mathop{\sum }\limits_{{m = 1}}^{M}{\left\{  {R}_{n}\left( {\mathbf{w}}_{m}^{0}\right) \right\}  }^{G}}{{\xi }_{n}^{2G}}\notag\\
	= &{o}_{p}\left( 1\right) ,
\end{align}
where  \({c}_{1}\) is a positive constant. The second inequality follows from the Bonferroni inequality, the third is obtained by an extension of Chebyshev's inequality, the fourth is guaranteed by Theorem 2 of \cite{Whittle} and Condition 1, the fifth follows from \eqref{eq8}, and the final equality is derived from Condition 2.

Define \({A}_{n1}\left( {\mathbf{X},\bm{\Psi}}\right) = P\left( {\mathop{\sup }\limits_{{\mathbf{w} \in  {\mathcal{H}}_{n}}}\left| \frac{{\bm{\epsilon} }^{\mathrm{T}}\left( {P\left( \mathbf{w}\right)  - \mathbf{{I}}_{n}}\right) \bm{\mu} }{{R}_{n}\left( \mathbf{w}\right) }\right|  > \delta |\mathbf{X},\bm{\Psi} }\right)\). Let \(F\left( {\mathbf{X},\bm{\Psi}}\right)\) be the joint cumulative distribution function of $\mathbf{X}$ and $\bm{\Psi}$, and let \({c}_{2},\cdots ,{c}_{6}\) be a sequence of positive constants. Note that \({A}_{n1}\left( {\mathbf{X},\bm{\Psi}}\right)\) is a conditional probability density given $\mathbf{X}$ and $\bm{\Psi}$, hence \(\left| {{A}_{n1}\left( {\mathbf{X},\bm{\Psi}}\right) }\right|  \leq  1\). From (A.54), for any $\varsigma$ > 0, we obtain \(P\left( {\left| {{A}_{n1}\left( {\mathbf{X},\bm{\Psi}}\right) }\right|  \geq  \varsigma }\right)  \rightarrow  0\). Then we have 
\begin{align}
	&P\left( {\mathop{\sup }\limits_{{\mathbf{w} \in  {\mathcal{H}}_{n}}}\left| \frac{{\bm{\epsilon} }^{\mathrm{T}}\left( {\mathbf{P}\left( \mathbf{w}\right)  - \mathbf{{I}}_{n}}\right) \bm{\mu} }{{R}_{n}\left( \mathbf{w}\right) }\right|  > \delta }\right)\notag\\
	=& E\left\{  {P\left( {\mathop{\sup }\limits_{{\mathbf{w} \in  {\mathcal{H}}_{n}}}\left| \frac{{\bm{\epsilon} }^{\mathrm{T}}\left( {\mathbf{P}\left( \mathbf{w}\right)  - \mathbf{{I}}_{n}}\right) \bm{\mu} }{{R}_{n}\left( \mathbf{w}\right) }\right|  > \delta |\mathbf{X},\bm{\Psi}}\right) }\right\}\notag\\
	= &E\left\{  {{A}_{n1}\left( {\mathbf{X},\bm{\Psi}}\right) }\right\}\notag\\
	= &{\int }_{\left| {{A}_{n1}\left( {\mathbf{X},\bm{\Psi}}\right) }\right|  \geq  \varsigma }{A}_{n1}\left( {\mathbf{X},\bm{\Psi}}\right) {dF}\left( {\mathbf{X},\bm{\Psi}}\right)  + {\int }_{\left| {{A}_{n1}\left( {\mathbf{X},\bm{\Psi}}\right) }\right|  < \varsigma }{A}_{n1}\left( {\mathbf{X},\bm{\Psi}}\right) {dF}\left( {\mathbf{X},\bm{\Psi}}\right)\notag\\
	\leq & P\left( {\left| {{A}_{n1}\left( {\mathbf{X},\bm{\Psi}}\right) }\right|  \geq  \varsigma }\right)  + \varsigma  \rightarrow  0, 
\end{align}
which implies that (A.51) is true.

Similar to the proof of (A.51), we have
\begin{align}
	&P\left( {\mathop{\sup }\limits_{{\mathbf{w} \in  {\mathcal{H}}_{n}}}\left| \frac{{\bm{\epsilon} }^{\mathrm{T}}\mathbf{P}\left( \mathbf{w}\right) \bm{\epsilon}  - \operatorname{tr}\left( {\mathbf{P}\left( \mathbf{w}\right) \bm{\Omega} }\right) }{{R}_{n}\left( \mathbf{w}\right) }\right|  > \delta  |\mathbf{X},\bm{\Psi}}\right)\notag\\
	\leq & P\left( {\mathop{\sup }\limits_{{\mathbf{w} \in  {\mathcal{H}}_{n}}}\mathop{\sum }\limits_{{m = 1}}^{M}{w}_{m}\left| {{\bm{\epsilon} }^{\mathrm{T}}{\mathbf{P}}_{\left( m\right) }\bm{\epsilon}  - \operatorname{tr}\left( {{\mathbf{P}}_{\left( m\right) }\bm{\Omega} }\right) }\right|  > \delta {\xi }_{n} |\mathbf{X},\bm{\Psi}}\right)\notag\\
	=& P\left( {\mathop{\max }\limits_{{1 \leq  m \leq  M}}\left| {{\bm{\epsilon} }^{\mathrm{T}}{\mathbf{P}}_{\left( m\right) }\bm{\epsilon}  - \operatorname{tr}\left( {{\mathbf{P}}_{\left( m\right) }\bm{\Omega} }\right) }\right|  > \delta {\xi }_{n} |\mathbf{X},\bm{\Psi}}\right)\notag\\
	\leq & \mathop{\sum }\limits_{{m = 1}}^{M}P\left( {\left| {{\bm{\epsilon} }^{\mathrm{T}}\mathbf{P}\left( {\mathbf{w}}_{m}^{0}\right) \bm{\epsilon}  - \operatorname{tr}\left( {\mathbf{P}\left( {\mathbf{w}}_{m}^{0}\right) \bm{\Omega} }\right) }\right|  > \delta {\xi }_{n} |\mathbf{X},\bm{\Psi}}\right)\notag\\
	\leq & \mathop{\sum }\limits_{{m = 1}}^{M}\frac{E\left( {{\left| {\bm{\epsilon} }^{\mathrm{T}}\mathbf{P}\left( {\mathbf{w}}_{m}^{0}\right) \bm{\epsilon}  - \operatorname{tr}\left( \mathbf{P}\left( {\mathbf{w}}_{m}^{0}\right) \bm{\Omega} \right) \right| }^{2G} |\mathbf{X},\bm{\Psi}}\right) }{{\left( \delta {\xi }_{n}\right) }^{2G}}\notag\\
	\leq & \mathop{\sum }\limits_{{m = 1}}^{M}\frac{{c}_{2}t{r}^{G}\left\{  {{\mathbf{P}}^{\mathrm{T}}\left( {\mathbf{w}}_{m}^{0}\right) \mathbf{P}\left( {\mathbf{w}}_{m}^{0}\right) \bm{\Omega} }\right\}  }{{\left( \delta {\xi }_{n}\right) }^{2G}}\notag\\
	\leq & \mathop{\sum }\limits_{{m = 1}}^{M}\frac{{c}_{2}{\left\{  {R}_{n}\left( {\mathbf{w}}_{m}^{0}\right) \right\}  }^{G}}{{\left( \delta {\xi }_{n}\right) }^{2G}}\notag\\
	= &\frac{{c}_{2}}{{\delta }^{2G}} \cdot  \frac{\mathop{\sum }\limits_{{m = 1}}^{M}{\left\{  {R}_{n}\left( {\mathbf{w}}_{m}^{0}\right) \right\}  }^{G}}{{\xi }_{n}^{2G}}\notag\\
	= &{o}_{p}\left( 1\right) \text{ . }
\end{align}
Applying the steps in (A.55) to the above equation, we can obtain (A.52).

Now, note that
\begin{align}
	&\mathop{\sup }\limits_{{\mathbf{w} \in  {\mathcal{H}}_{n}}}\left| {\frac{{L}_{n}\left( \mathbf{w}\right) }{{R}_{n}\left( \mathbf{w}\right) } - 1}\right|  \notag\\
	= &\mathop{\sup }\limits_{{\mathbf{w} \in  {\mathcal{H}}_{n}}}\left| \frac{2{\bm{\epsilon} }^{\mathrm{T}}{\mathbf{P}}^{\mathrm{T}}\left( \mathbf{w}\right) \left( {\mathbf{P}\left( \mathbf{w}\right)  - \mathbf{{I}}_{n}}\right) \bm{\mu}  + \parallel \mathbf{P}\left( \mathbf{w}\right) \bm{\epsilon} {\parallel }^{2} - \operatorname{tr}\left( {{\mathbf{P}}^{\mathrm{T}}\left( \mathbf{w}\right) \mathbf{P}\left( \mathbf{w}\right) \bm{\Omega} }\right) }{{R}_{n}\left( \mathbf{w}\right) }\right|. \notag
\end{align}
To prove (A.53), it is sufficient to show that
\begin{equation}
	\mathop{\sup }\limits_{{\mathbf{w} \in  {\mathcal{H}}_{n}}}\left| \frac{{\bm{\epsilon} }^{\mathrm{T}}{\mathbf{P}}^{\mathrm{T}}\left( \mathbf{w}\right) \left( {\mathbf{P}\left( \mathbf{w}\right)  - \mathbf{{I}}_{n}}\right) \bm{\mu} }{{R}_{n}\left( \mathbf{w}\right) }\right| \overset{p}{ \rightarrow  }0, 
\end{equation}
and
\begin{equation}
	\mathop{\sup }\limits_{{\mathbf{w} \in  {\mathcal{H}}_{n}}}\left| \frac{\parallel \mathbf{P}\left( \mathbf{w}\right) \bm{\epsilon} {\parallel }^{2} - \operatorname{tr}\left( {{\mathbf{P}}^{\mathrm{T}}\left( \mathbf{w}\right) \mathbf{P}\left( \mathbf{w}\right) \bm{\Omega} }\right) }{{R}_{n}\left( \mathbf{w}\right) }\right| \overset{p}{ \rightarrow  }0.
\end{equation}

By Lemma 2 and Conditions 1 and 2, we have
\begin{align}
	&P\left( {\mathop{\sup }\limits_{{\mathbf{w} \in  {\mathcal{H}}_{n}}}\left| \frac{{\bm{\epsilon} }^{\mathrm{T}}{\mathbf{P}}^{\mathrm{T}}\left( \mathbf{w}\right) \left( {\mathbf{P}\left( \mathbf{w}\right)  - \mathbf{{I}}_{n}}\right) \bm{\mu} }{{R}_{n}\left( \mathbf{w}\right) }\right|  > \delta  |\mathbf{X},\bm{\Psi}}\right)\notag\\
	\leq & P\left( {\mathop{\sup }\limits_{{\mathbf{w} \in  {\mathcal{H}}_{n}}}\mathop{\sum }\limits_{{t = 1}}^{M}\mathop{\sum }\limits_{{m = 1}}^{M}{w}_{t}{w}_{m}\left| {{\bm{\epsilon} }^{\mathrm{T}}{\mathbf{P}}_{\left( t\right) }^{\mathrm{T}}\left( {{\mathbf{P}}_{\left( m\right) } - \mathbf{{I}}_{n}}\right) \bm{\mu} }\right|  > \delta {\xi }_{n} |\mathbf{X},\bm{\Psi}}\right)\notag\\
	\leq & P\left( {\mathop{\max }\limits_{{1 \leq  t \leq  M}}\mathop{\max }\limits_{{1 \leq  m \leq  M}}\left| {{\bm{\epsilon} }^{\mathrm{T}}{\mathbf{P}}_{\left( t\right) }^{\mathrm{T}}\left( {{\mathbf{P}}_{\left( m\right) } - \mathbf{{I}}_{n}}\right) \bm{\mu} }\right|  > \delta {\xi }_{n} |\mathbf{X},\bm{\Psi}}\right)\notag\\
	\leq & \mathop{\sum }\limits_{{t = 1}}^{M}\mathop{\sum }\limits_{{m = 1}}^{M}P\left( {\left| {{\bm{\epsilon} }^{\mathrm{T}}{\mathbf{P}}^{\mathrm{T}}\left( {\mathbf{w}}_{t}^{0}\right) \left( {\mathbf{P}\left( {\mathbf{w}}_{m}^{0}\right)  - \mathbf{{I}}_{n}}\right) \bm{\mu} }\right|  > \delta {\xi }_{n} |\mathbf{X},\bm{\Psi}}\right)\notag\\
	\leq & \mathop{\sum }\limits_{{t = 1}}^{M}\mathop{\sum }\limits_{{m = 1}}^{M}\frac{E\left( {{\left| {\bm{\epsilon} }^{\mathrm{T}}{\mathbf{P}}^{\mathrm{T}}\left( {\mathbf{w}}_{t}^{0}\right) \left( \mathbf{P}\left( {\mathbf{w}}_{m}^{0}\right)  - \mathbf{{I}}_{n}\right) \bm{\mu} \right| }^{2G} |\mathbf{X},\bm{\Psi}}\right) }{{\left( \delta {\xi }_{n}\right) }^{2G}}\notag\\
	\leq & \mathop{\sum }\limits_{{t = 1}}^{M}\mathop{\sum }\limits_{{m = 1}}^{M}\frac{{c}_{3}{\begin{Vmatrix}{\mathbf{P}}^{\mathrm{T}}\left( {\mathbf{w}}_{t}^{0}\right) \left( \mathbf{P}\left( {\mathbf{w}}_{m}^{0}\right)  - \mathbf{{I}}_{n}\right) \bm{\mu} \end{Vmatrix}}^{2G}}{{\left( \delta {\xi }_{n}\right) }^{2G}}\notag\\
	\leq & \mathop{\sum }\limits_{{t = 1}}^{M}\mathop{\sum }\limits_{{m = 1}}^{M}\frac{{c}_{3}{\bar{\lambda }}^{2G}\left( {\mathbf{P}\left( {\mathbf{w}}_{t}^{0}\right) }\right) {\begin{Vmatrix}\left( \mathbf{P}\left( {\mathbf{w}}_{m}^{0}\right)  - \mathbf{{I}}_{n}\right) \bm{\mu} \end{Vmatrix}}^{2G}}{{\left( \delta {\xi }_{n}\right) }^{2G}}\notag\\
	\leq & \mathop{\sum }\limits_{{t = 1}}^{M}\mathop{\sum }\limits_{{m = 1}}^{M}\frac{{c}_{3}{\bar{\lambda }}^{2G}\left( {\mathbf{P}\left( {\mathbf{w}}_{t}^{0}\right) }\right) {\left\{  {R}_{n}\left( {\mathbf{w}}_{m}^{0}\right) \right\}  }^{G}}{{\left( \delta {\xi }_{n}\right) }^{2G}}\notag\\
	= &\frac{{c}_{4}}{{\delta }^{2G}} \cdot  \frac{M{O}_{p}\left( {\widetilde{p}}^{G}\right) \mathop{\sum }\limits_{{m = 1}}^{M}{\left\{  {R}_{n}\left( {\mathbf{w}}_{m}^{0}\right) \right\}  }^{G}}{{\xi }_{n}^{2G}}\notag\\
	= &{o}_{p}\left( 1\right) \text{.} 
\end{align}
Equation (A.57) can be obtained by applying the steps in (A.55) to (A.59).

Furthermore,
\begin{align}
	&P\left( {\mathop{\sup }\limits_{{\mathbf{w} \in  {\mathcal{H}}_{n}}}\left| \frac{\parallel \mathbf{P}\left( \mathbf{w}\right) \bm{\epsilon} {\parallel }^{2} - \operatorname{tr}\left( {{\mathbf{P}}^{\mathrm{T}}\left( \mathbf{w}\right) \mathbf{P}\left( \mathbf{w}\right) \bm{\Omega} }\right) }{{R}_{n}\left( \mathbf{w}\right) }\right|  > \delta  |\mathbf{X},\bm{\Psi}}\right)\notag\\
	\leq & P\left( {\mathop{\sup }\limits_{{\mathbf{w} \in  {\mathcal{H}}_{n}}}\mathop{\sum }\limits_{{t = 1}}^{M}\mathop{\sum }\limits_{{m = 1}}^{M}{w}_{t}{w}_{m}\left| {{\bm{\epsilon} }^{\mathrm{T}}{\mathbf{P}}_{\left( t\right) }^{\mathrm{T}}{\mathbf{P}}_{\left( m\right) }\bm{\epsilon}  - \operatorname{tr}\left( {{\mathbf{P}}_{\left( t\right) }^{\mathrm{T}}{\mathbf{P}}_{\left( m\right) }\bm{\Omega} }\right) }\right|  > \delta {\xi }_{n} |\mathbf{X},\bm{\Psi}}\right)\notag\\
	\leq & P\left( {\mathop{\max }\limits_{{1 \leq  t \leq  M}}\mathop{\max }\limits_{{1 \leq  m \leq  M}}\left| {{\bm{\epsilon} }^{\mathrm{T}}{\mathbf{P}}_{\left( t\right) }^{\mathrm{T}}{\mathbf{P}}_{\left( m\right) }\bm{\epsilon}  - \operatorname{tr}\left( {{\mathbf{P}}_{\left( t\right) }^{\mathrm{T}}{\mathbf{P}}_{\left( m\right) }\bm{\Omega} }\right) }\right|  > \delta {\xi }_{n} |\mathbf{X},\bm{\Psi}}\right)\notag\\
	\leq & \mathop{\sum }\limits_{{t = 1}}^{M}\mathop{\sum }\limits_{{m = 1}}^{M}P\left( {\left| {{\bm{\epsilon} }^{\mathrm{T}}{\mathbf{P}}^{\mathrm{T}}\left( {\mathbf{w}}_{t}^{0}\right) \mathbf{P}\left( {\mathbf{w}}_{m}^{0}\right) \bm{\epsilon}  - \operatorname{tr}\left( {{\mathbf{P}}^{\mathrm{T}}\left( {\mathbf{w}}_{t}^{0}\right) \mathbf{P}\left( {\mathbf{w}}_{m}^{0}\right) \bm{\Omega} }\right) }\right|  > \delta {\xi }_{n} |\mathbf{X},\bm{\Psi}}\right)\notag\\
	\leq & \mathop{\sum }\limits_{{t = 1}}^{M}\mathop{\sum }\limits_{{m = 1}}^{M}\frac{E\left( {{\left| {\bm{\epsilon} }^{\mathrm{T}}{\mathbf{P}}^{\mathrm{T}}\left( {\mathbf{w}}_{t}^{0}\right) \mathbf{P}\left( {\mathbf{w}}_{m}^{0}\right) \bm{\epsilon}  - \operatorname{tr}\left( {\mathbf{P}}^{\mathrm{T}}\left( {\mathbf{w}}_{t}^{0}\right) \mathbf{P}\left( {\mathbf{w}}_{m}^{0}\right) \bm{\Omega} \right) \right| }^{2G} |\mathbf{X},\bm{\Psi}}\right) }{{\left( \delta {\xi }_{n}\right) }^{2G}}\notag\\
	\leq & \mathop{\sum }\limits_{{t = 1}}^{M}\mathop{\sum }\limits_{{m = 1}}^{M}\frac{{c}_{5}t{r}^{G}\left\{  {{\bm{\Omega} }^{1/2}{\mathbf{P}}^{\mathrm{T}}\left( {\mathbf{w}}_{t}^{0}\right) \mathbf{P}\left( {\mathbf{w}}_{m}^{0}\right) \bm{\Omega} {\mathbf{P}}^{\mathrm{T}}\left( {\mathbf{w}}_{m}^{0}\right) \mathbf{P}\left( {\mathbf{w}}_{t}^{0}\right) {\bm{\Omega} }^{1/2}}\right\}  }{{\left( \delta {\xi }_{n}\right) }^{2G}}\notag\\
	\leq & \mathop{\sum }\limits_{{t = 1}}^{M}\mathop{\sum }\limits_{{m = 1}}^{M}\frac{{c}_{5}{\bar{\lambda }}^{2G}\left( {\mathbf{P}\left( {\mathbf{w}}_{m}^{0}\right) }\right) {\bar{\lambda }}^{G}\left( \bm{\Omega} \right) t{r}^{G}\left\{  {{\mathbf{P}}^{\mathrm{T}}\left( {\mathbf{w}}_{t}^{0}\right) \mathbf{P}\left( {\mathbf{w}}_{t}^{0}\right) \bm{\Omega} }\right\}  }{{\left( \delta {\xi }_{n}\right) }^{2G}}\notag\\
	\leq & \mathop{\sum }\limits_{{t = 1}}^{M}\mathop{\sum }\limits_{{m = 1}}^{M}\frac{{c}_{5}{\bar{\lambda }}^{2G}\left( {\mathbf{P}\left( {\mathbf{w}}_{m}^{0}\right) }\right) {\bar{\lambda }}^{G}\left( \bm{\Omega} \right) {R}_{n}^{G}\left( {\mathbf{w}}_{t}^{0}\right) }{{\left( \delta {\xi }_{n}\right) }^{2G}}\notag\\
	= &\frac{{c}_{6}}{{\delta }^{2G}} \cdot  \frac{M{O}_{p}\left( {\widetilde{p}}^{G}\right) \mathop{\sum }\limits_{{t = 1}}^{M}{\left\{  {R}_{n}\left( {\mathbf{w}}_{t}^{0}\right) \right\}  }^{G}}{{\xi }_{n}^{2G}}\notag\\
	= &{o}_{p}\left( 1\right) \text{ . }
\end{align}
Equation (A.58) is obtained by applying the derivation steps in (A.55) to (A.60).

\noindent\textbf{Proof of Theorem 2}

When $\bm{\Omega}$ is replaced by $\widehat{\bm{\Omega}}$, ${C}_{n}(\mathbf{w})$ under the known $\bm{\Omega}$ case is correspondingly changed to
\begin{align}
	{\widehat{C}}_{n}\left( \mathbf{w}\right)  &= {C}_{n}\left( \mathbf{w}\right)  + 2\{ \operatorname{tr}\left( {\mathbf{P}\left( \mathbf{w}\right) \widehat{\bm{\Omega} }}\right)  - \operatorname{tr}\left( {\mathbf{P}\left( \mathbf{w}\right) \bm{\Omega} }\right) \} \notag\\
	&= {C}_{n}\left( \mathbf{w}\right)  + 2(\widehat{\sigma}^2_{M^*}-\sigma^2)\operatorname{tr}\left( {\mathbf{P}\left( \mathbf{w}\right)}\right).\notag
\end{align}
From the result of Theorem 1, it suffices to prove that,
\begin{equation}
	\mathop{\sup }\limits_{{\mathbf{w} \in  {\mathcal{H}}_{n}}}\left| (\widehat{\sigma}^2_{M^*}-\sigma^2)\operatorname{tr}\left( {\mathbf{P}\left( \mathbf{w}\right)}\right)\right| /{R}_{n}\left( \mathbf{w}\right)  = {o}_{p}\left( 1\right) . 
\end{equation}

Under the scenario where all candidate models are incorrect, Condition 7 implies that  $\widehat{\sigma}^2_{M^*}=O_p(1)$. According to Lemma 2 and Condition 8, (A.61) holds.

\noindent\textbf{Proof of Theorem 3}

Define $\mathcal{H}_T = \{\mathbf{w} \in \mathcal{H}_n: w_m = 0, m = M_0+1, \dots, M\}$ as the set of weight vectors that assign zero weight to in-quasi-correct models. Note that, we rewrite our criterion $C_n(\mathbf{w})$ as
\begin{align*}
	C_{n}(\mathbf{w}) =&\|\widehat{\bm{\mu}}(\mathbf{w})-\mathbf{Y}\|^2+2\operatorname{tr}(\mathbf{P}(\mathbf{w})\bm{\Omega}) \\
	=&\|\bm{\mu}^{*}(\mathbf{w})-\bm{\mu}+\widehat{\bm{\mu}}(\mathbf{w})-\bm{\mu}^{*}(\mathbf{w})-\bm{\epsilon}\|^2+2\operatorname{tr}(\mathbf{P}(\mathbf{w})\bm{\Omega}) \\
	=&\|\bm{\mu}^{*}(\mathbf{w})-\bm{\mu}\|^2+\|\widehat{\bm{\mu}}(\mathbf{w})-\bm{\mu}^{*}(\mathbf{w})\|^2+\|\bm{\epsilon}\|^2\\
	&+2(\bm{\mu}^{*}(\mathbf{w})-\bm{\mu})^{\mathrm{T}}(\widehat{\bm{\mu}}(\mathbf{w})-\bm{\mu}^{*}(\mathbf{w}))-2(\bm{\mu}^{*}(\mathbf{w})-\bm{\mu})^{\mathrm{T}}\bm{\epsilon}\\
	&-2(\widehat{\bm{\mu}}(\mathbf{w})-\bm{\mu}^{*}(\mathbf{w}))^{\mathrm{T}}\bm{\epsilon}+2\operatorname{tr}(\mathbf{P}(\mathbf{w})\bm{\Omega}).
\end{align*}
Hence, we can calculate the value of $C_n(\mathbf{w})$ at $\mathbf{w}_{0}=(w_{1}^{0},w_{2}^{0},\cdots,w_{M_{0}}^{0},0,\cdots,0) \in \mathcal{H}_{T}$. It is show that
\begin{align}
	C_{n}(\mathbf{w}_0) &=\|\widehat{\bm{\mu}}(\mathbf{w}_0)-\mathbf{Y}\|^2+2\operatorname{tr}(\mathbf{P}(\mathbf{w}_0)\bm{\Omega}) \notag\\
	&=\|\widehat{\bm{\mu}}(\mathbf{w}_0)-\bm{\mu}^{*}(\mathbf{w}_0)\|^2-2(\widehat{\bm{\mu}}(\mathbf{w}_0)-\bm{\mu}^{*}(\mathbf{w}_0))^{\mathrm{T}}\bm{\epsilon}+ \|\bm{\epsilon}\|^2+2\operatorname{tr}(\mathbf{P}(w_0)\bm{\Omega}) \notag\\
	&=C_{n1}(\mathbf{w}_0)-2C_{n2}(\mathbf{w}_0)+2C_{n3}(\mathbf{w}_0)+\|\bm{\epsilon}\|^2,
\end{align}
where
\begin{align}
	C_{n1}(\mathbf{w}_0) &= \|\widehat{\bm{\mu}}(\mathbf{w}_0)-\bm{\mu}^{*}(\mathbf{w}_0)\|^2,  \\
	C_{n2}(\mathbf{w}_0) &= (\widehat{\bm{\mu}}(\mathbf{w}_0)-\bm{\mu}^{*}(\mathbf{w}_0))^{\mathrm{T}}\bm{\epsilon}, 
\end{align}
and
\begin{equation}
	C_{n3}(\mathbf{w}_0)=\operatorname{tr}(\mathbf{P}(\mathbf{w}_0)\bm{\Omega}).
\end{equation}
Next, we examine each term in (A.62). For (A.63), by Condition 4 and Lemma 3, we observe that
\begin{align}
	C_{n1}(\mathbf{w}_0) &= \|\widehat{\bm{\mu}}(\mathbf{w}_0)-\bm{\mu}^{*}(\mathbf{w}_0)\|^2\notag\\
	&\leq \mathop{\sum }\limits_{{i = 1}}^{n}\mathop{\sum }\limits_{{m = 1}}^{M_0}|\mathbf{X}_{(m)i}^{\mathrm{T}}(\widehat{\bm{\beta}}_{(m)}(\mathbf{s}_i)-\bm{\beta}_{(m)}^{*}(\mathbf{s}_i))|^2\notag\\
	&\leq\mathop{\sum }\limits_{{i = 1}}^{n}\mathop{\sum }\limits_{{m = 1}}^{M_0}\|\mathbf{X}_{(m)i}\|^2\|\widehat{\bm{\beta}}_{(m)}(\mathbf{s}_i)-\bm{\beta}_{(m)}^{*}(\mathbf{s}_i)\|^2\notag\\
	&=O_p(M_0\widetilde{p}^4+\underline{h}^{-2}M_0\widetilde{p}^2+nM_0\widetilde{p}^2\bar{h}^4).
\end{align}
For (A.64), we have 
\begin{equation}
	E\{(\widehat{\bm{\mu}}(\mathbf{w}_0)-\bm{\mu}^{*}(\mathbf{w}_0))^{\mathrm{T}}\bm{\epsilon}\}=0.
\end{equation}
By Condition 11 and (A.66), we obtain
\begin{align}
	&\|\widehat{\bm{\mu}}(\mathbf{w}_0)-\bm{\mu}^{*}(\mathbf{w}_0)\|^2/n\notag\\
	=&O_p(\xi_F/n)O_p(\xi_F^{-1}M_0\widetilde{p}^4+\xi_F^{-1}\underline{h}^{-2}M_0\widetilde{p}^2+\xi_F^{-1}nM_0\widetilde{p}^2\bar{h}^4)\notag\\
	=&o_p(1).
\end{align}
Consequently, by the law of large numbers,
\begin{equation}
	\mathrm{var}\{(\widehat{\bm{\mu}}(\mathbf{w}_0)-\bm{\mu}^{*}(\mathbf{w}_0))^{\mathrm{T}}\bm{\epsilon}\}=O(n),
\end{equation}
which implies
\begin{align}
	|C_{n2}(\mathbf{w}_0)|&=|(\widehat{\bm{\mu}}(\mathbf{w}_0)-\bm{\mu}^{*}(\mathbf{w}_0))^{\mathrm{T}}\bm{\epsilon}|\notag\\
	&=O_p(n^{1/2}).
\end{align}
For (A.65), applying Lemma 2 yields
\begin{align}
	C_{n3}(\mathbf{w}_0)&=\operatorname{tr}(\mathbf{P}(\mathbf{w}_0)\bm{\Omega})\notag\\
	&\leq\mathop{\max }\limits_{{1 \leq  m \leq  M_0}}\operatorname{tr}(\mathbf{P}_{(m)}\bm{\Omega})\notag\\
	&=O_p({\underline{h}}^{-2}\widetilde{p}).
\end{align}
Combining (A.63)-(A.65), we obtain
\begin{equation}
	C_{n}(\mathbf{w}_0)-\|\bm{\epsilon}\|^2=O_p(M_0\widetilde{p}^4+\underline{h}^{-2}M_0\widetilde{p}^2+nM_0\widetilde{p}^2\bar{h}^4+n^{1/2}).
\end{equation}

Next, we calculate the value of $C_{n}(\mathbf{w})$ at $\widehat{\mathbf{w}}$. It is obvious seen that
\begin{align}
	C_{n}(\widehat{\mathbf{w}})=&\|\widehat{\bm{\mu}}(\widehat{\mathbf{w}})-\mathbf{Y} \|^{2}+2\operatorname{tr}(\mathbf{P}(\widehat{\mathbf{w}})\bm{\Omega}) \notag\\
	=&L_n^{*}(\widehat{\mathbf{w}})+\|\widehat{\bm{\mu}}(\widehat{\mathbf{w}})-\bm{\mu}^{*}(\widehat{\mathbf{w}})\|^2+\|\bm{\epsilon}\|^2+2(\bm{\mu}^{*}(\widehat{\mathbf{w}})-\bm{\mu})^{\mathrm{T}}(\widehat{\bm{\mu}}(\widehat{\mathbf{w}})-\bm{\mu}^{*}(\widehat{\mathbf{w}}))\notag\\
	&-2(\bm{\mu}^{*}(\widehat{\mathbf{w}})-\bm{\mu})^{\mathrm{T}}\bm{\epsilon}-2(\widehat{\bm{\mu}}(\widehat{\mathbf{w}})-\bm{\mu}^{*}(\widehat{\mathbf{w}}))^{\mathrm{T}}\bm{\epsilon}+2\operatorname{tr}(\mathbf{P}(\widehat{\mathbf{w}})\bm{\Omega}).
\end{align}
Similar to (A.66), by Condition 4 and Lemma 3,
\begin{align}
	&\|\widehat{\bm{\mu}}(\widehat{\mathbf{w}})-\bm{\mu}^{*}(\widehat{\mathbf{w}})\|^2\notag\\
	\leq &\mathop{\sum }\limits_{{i = 1}}^{n}\mathop{\sum }\limits_{{m = 1}}^{M}|\mathbf{X}_{(m)i}^{\mathrm{T}}(\widehat{\bm{\beta}}_{(m)}(\mathbf{s}_i)-\bm{\beta}_{(m)}^{*}(\mathbf{s}_i))|^2\notag\\
	\leq&\mathop{\sum }\limits_{{i = 1}}^{n}\mathop{\sum }\limits_{{m = 1}}^{M}\|\mathbf{X}_{(m)i}\|^2\|\widehat{\bm{\beta}}_{(m)}(\mathbf{s}_i)-\bm{\beta}_{(m)}^{*}(\mathbf{s}_i)\|^2\notag\\
	=&O_p(M\widetilde{p}^4+\underline{h}^{-2}M\widetilde{p}^2+nM\widetilde{p}^2\bar{h}^4).
\end{align}
Combining Condition 11 and(A.74), and following a derivation similar to (A.66)-(A.70), we obtain
\begin{equation}
	|(\widehat{\bm{\mu}}(\widehat{\mathbf{w}})-\bm{\mu}^{*}(\widehat{\mathbf{w}}))^{\mathrm{T}}\bm{\epsilon}|=O_p(n^{1/2}).
\end{equation}
By Condition 10, we have
\begin{align}
	&\|\bm{\mu}^{*}(\widehat{\mathbf{w}})-\bm{\mu}\|^2\notag\\
	\leq&\sum_{t=1}^{M}\sum_{m=1}^{M}\widehat{w}_t\widehat{w}_m(\bm{\mu}^{*}_{(t)}-\bm{\mu})^{\mathrm{T}}(\bm{\mu}^{*}_{(m)}-\bm{\mu})\notag\\
	\leq&\sum_{t=1}^{M}\sum_{m=1}^{M}\widehat{w}_t\widehat{w}_m\|\bm{\mu}^{*}_{(t)}-\bm{\mu}\|\|\bm{\mu}^{*}_{(m)}-\bm{\mu}\|\notag\\
	=&O_p(n).
\end{align}
Furthermore, from Condition 11 and (A.76), and using an argument analogous to (A.66)-(A.70), we have
\begin{equation}
	|(\bm{\mu}^{*}(\widehat{\mathbf{w}})-\bm{\mu})^{\mathrm{T}}\bm{\epsilon}|=O_p(n^{1/2}).
\end{equation}
By the Cauchy–Schwarz inequality, we obtain
\begin{equation}
	|(\bm{\mu}^{*}(\widehat{\mathbf{w}})-\bm{\mu})^{\mathrm{T}}(\widehat{\bm{\mu}}(\widehat{\mathbf{w}})-\bm{\mu}^{*}(\widehat{\mathbf{w}}))|=O_p(L_n^{*1/2}(\widehat{\mathbf{w}}))O_p(\|\widehat{\bm{\mu}}(\widehat{\mathbf{w}})-\bm{\mu}^{*}(\widehat{\mathbf{w}})\|).
\end{equation}
Similarly to (A.65), it follows from Lemma 2 that
\begin{align}
	\operatorname{tr}(\mathbf{P}(\widehat{\mathbf{w}})\bm{\Omega})
	&\leq\mathop{\max }\limits_{{1 \leq  m \leq  M}}\operatorname{tr}(\mathbf{P}_{(m)}\bm{\Omega})\notag\\
	&=O_p({\underline{h}}^{-2}\widetilde{p}).
\end{align}
On the other hand, we observe that
\begin{align}
	L_n^{*}(\widehat{\mathbf{w}})&=\|\bm{\mu}^{*}(\widehat{\mathbf{w}})-\bm{\mu}\|^2\notag\\
	&=\left\|\left(\sum_{m=1}^{M}\widehat{w}_m(\bm{\mu}^{*}_{(m)}-\bm{\mu})\right)\right\|^2\notag\\
	&=\left\|\left(\sum_{m=1}^{M_0}\widehat{w}_m(\bm{\mu}^{*}_{(m)}-\bm{\mu})+\sum_{m=M_0+1}^{M}\widehat{w}_m(\bm{\mu}^{*}_{(m)}-\bm{\mu})\right)\right\|^2\notag\\
	&=\left\|\left((1-\widehat{\tau}_{M_0})\sum_{m=M_0+1}^{M}\frac{\widehat{w}_m}{1-\widehat{\tau}_{M_0}}(\bm{\mu}^{*}_{(m)}-\bm{\mu})\right)\right\|^2\notag\\
	&=(1-\widehat{\tau}_{M_0})^2\left\|\left(\sum_{m=M_0+1}^{M}\frac{\widehat{w}_m}{1-\widehat{\tau}_{M_0}}(\bm{\mu}^{*}_{(m)}-\bm{\mu})\right)\right\|^2\notag\\
	&=(1-\widehat{\tau}_{M_0})^2\left\|(\bm{\mu}^{*}(\widehat{\mathbf{w}}_F)-\bm{\mu})\right\|^2\notag\\
	&=(1-\widehat{\tau}_{M_0})^2L^{*}_n(\widehat{\mathbf{w}}_F),
\end{align}
where $\widehat{\mathbf{w}}_F = \left( 0, \dots, 0, \frac{\widehat{w}_{M_0+1}}{1 - \widehat{\tau}_{M_0}}, \dots, \frac{\widehat{w}_{M}}{1 - \widehat{\tau}_{M_0}} \right)^{\mathrm{T}}$. Thus by (A.73), (A.75),  and (A.77)-(A.79), we have 
\begin{align}
	C_{n}(\widehat{\mathbf{w}})-\|\bm{\epsilon}\|^2=&L_n^{*}(\widehat{\mathbf{w}})+\|\widehat{\bm{\mu}}(\widehat{\mathbf{w}})-\bm{\mu}^{*}(\widehat{\mathbf{w}})\|^2+O_p(L_n^{*1/2}(\widehat{\mathbf{w}}))O_p(\|\widehat{\bm{\mu}}(\widehat{\mathbf{w}})-\bm{\mu}^{*}(\widehat{\mathbf{w}}))\notag\\
	&+O_p(n^{1/2}+{\underline{h}}^{-2}\widetilde{p}).
\end{align}
From the definitions of $\widehat{\mathbf{w}}$ and $\mathbf{w}_0$, it follows that
$$C(\widehat{\mathbf{w}})-\|\bm{\epsilon}\|^2\leq C(\mathbf{w}_{0})-\|\bm{\epsilon}\|^2.$$
By (A.72), (A.74), (A.80), (A.81) and Condition 11, we have $(1-\widehat{\tau}_{M_0})^2=o_p(1)$, hence the Theorem 3 has been proved.

\noindent\textbf{Proof of Corollary 1}

Notice that
\begin{equation}
	\widehat{\mu}_i(\widehat{\mathbf{w}})-\mu_i=\sum_{m=1}^{M_0}\widehat{w}_m(\widehat{\mu}_{(m)i}-\mu_i)+\sum_{m=M_0+1}^{M}\widehat{w}_m(\widehat{\mu}_{(m)i}-\mu_i).
\end{equation}
By Cauchy-Schwarz inequality, we have that
\begin{equation}
	(\widehat{\mu}_i(\widehat{\mathbf{w}})-\mu_i)^2\leq2\left(|\sum_{m=1}^{M_0}\widehat{w}_m(\widehat{\mu}_{(m)i}-\mu_i)|^2+|\sum_{m=M_0+1}^{M}\widehat{w}_m(\widehat{\mu}_{(m)i}-\mu_i)|^2\right).
\end{equation}
Next, we need to prove that the two terms of right hand side of (A.83) are both $o_p(1)$. For
the quasi-correct models among the candidates, according to Lemma 3, condition 4, and condition 11, we have
\begin{align}
	&|\sum_{m=1}^{M_0}\widehat{w}_m(\widehat{\mu}_{(m)i}-\mu_i)|^2\notag\\
	\leq&\sum_{m=1}^{M_0}|\mathbf{X}_{(m)i}^{\mathrm{T}}(\widehat{\bm{\beta}}_{(m)}(\mathbf{s}_i)-\bm{\beta}_{(m)}(\mathbf{s}_i))|^2\notag\\
	\leq&\sum_{m=1}^{M_0}\|\mathbf{X}_{(m)i}\|^2\|\widehat{\bm{\beta}}_{(m)}(\mathbf{s}_i)-\bm{\beta}_{(m)}(\mathbf{s}_i)\|^2\notag\\
	=&O_p(\xi_F/n)O_p(\xi_F^{-1}M_0\widetilde{p}^4+\xi_F^{-1}\underline{h}^{-2}M_0\widetilde{p}^2+\xi_F^{-1}nM_0\widetilde{p}^2\bar{h}^4)\notag\\
	=&o_p(1).
\end{align}
By Cauchy-Schwarz inequality, we have that
\begin{equation}
	|\sum_{m=M_0+1}^{M}\widehat{w}_m(\widehat{\mu}_{(m)i}-\mu_i)|^2\leq2|\sum_{m=M_0+1}^{M}\widehat{w}_m(\widehat{\mu}_{(m)i}-\mu^*_{(m)i})|^2+2|\sum_{m=M_0+1}^{M}\widehat{w}_m(\mu^*_{(m)i}-\mu_i)|^2.
\end{equation}
Following an argument similar to that in (A.84), and applying Lemma 3, Condition 4, and Condition 11, we have that
\begin{align}
	&|\sum_{m=M_0+1}^{M}\widehat{w}_m(\widehat{\mu}_{(m)i}-\mu^*_{(m)i})|^2\notag\\
	\leq&(1-\widehat{\tau}_{M_0})^2\sum_{m=M_0+1}^{M}|\mathbf{X}_{(m)i}^{\mathrm{T}}(\widehat{\bm{\beta}}_{(m)}(\mathbf{s}_i)-\bm{\beta}_{(m)}^{*}(\mathbf{s}_i))|^2\notag\\
	\leq&(1-\widehat{\tau}_{M_0})^2\sum_{m=M_0+1}^{M}\|\mathbf{X}_{(m)i}\|^2\|\widehat{\bm{\beta}}_{(m)}(\mathbf{s}_i)-\bm{\beta}_{(m)}^{*}(\mathbf{s}_i)\|^2\notag\\
	=&(1-\widehat{\tau}_{M_0})^2O_p(\xi_F/n)O_p(\xi_F^{-1}(M-M_0)\widetilde{p}^4+\xi_F^{-1}\underline{h}^{-2}(M-M_0)\widetilde{p}^2+\xi_F^{-1}n(M-M_0)\widetilde{p}^2\bar{h}^4)\notag\\
	=&o_p(1).
\end{align}
By Theorem 3 and Condition 10, we can obtain that
\begin{align}
	&|\sum_{m=M_0+1}^{M}\widehat{w}_m(\mu^*_{(m)i}-\mu_i)|^2\notag\\
	=&(1-\widehat{\tau}_{M_0})^2\left|\sum_{m=M_0+1}^{M}\frac{\widehat{w}_m}{1-\widehat{\tau}_{M_0}} (\mu^*_{(m)i}-\mu_i)\right|^2\notag\\
	=&o_p(1)O_p(1)\notag\\
	=&o_p(1).
\end{align}
By (A.86) and (A.87), we have that $|\sum_{m=M_0+1}^{M}\widehat{w}_m(\widehat{\mu}_{(m)i}-\mu_i)|^2=o_p(1)$. Along with (A.84), $(\widehat{\mu}_i(\widehat{\mathbf{w}})-\mu_i)^2$ is $o_p(1)$. Corollary 1 has been proved.

\noindent\textbf{Proof of Theorem 4}

According to the definition of $\widehat{\sigma}^2_{M^*}$, together with Condition 10, and Condition 11, we have
\begin{equation}
	\widehat{\sigma}^2_{M^*} = O_p(1).
\end{equation}
By replacing $\bm{\Omega}$ with $\bm{\widehat{\Omega}}$ in the proof of Theorem 3 and following the same steps, Theorem 4 is established.

\noindent\textbf{Proof of Corollary 2}

Based on Theorem 4 and following the same steps as in the proof of Corollary 1, Corollary 2 is established.


\begin{thebibliography}{}



\bibitem[Anselin, 1988]{Anselin1988}
\newblock Anselin, L. (1988). 
\newblock Spatial Econometrics: Methods and Models. \newblock {\em Springer Netherlands}.
	
\bibitem[Buckland et~al., 1997]{BBA97} 
Buckland S T, Burnham K P, Augustin N H. (1997).
\newblock Model selection: an integral part of inference. 
\newblock {\em Biometrics}, 603--618.

\bibitem[Brunsdon et~al., 1998]{BFC98} 
Brunsdon C, Fotheringham S, Charlton M. (1998).
\newblock Geographically weighted regression. 
\newblock {\em Journal of the Royal Statistical Society: Series D (The Statistician)}, 47: 431--443.

\bibitem[Brunsdon et~al., 1999]{BFC1999}
Brunsdon, C., Fotheringham, A. S., Charlton, M. (1999).
\newblock  Some notes on parametric significance tests for geographically weighted regression.
\newblock {\em Journal of regional science}, 39: 497--524.

\bibitem[Comber et~al., 2023]{CBCD23} 
Comber, A., Brunsdon, C., Charlton, M., Dong, G., Harris, R., Lu, B., ... \& Harris, P. (2023).
\newblock A route map for successful applications of geographically weighted regression.
\newblock {\em Geographical Analysis}, 55, 155--178.

\bibitem[Comber et~al., 2018]{CWLZH18} 
Comber, A., Wang, Y., Lü, Y., Zhang, X., Harris, P. (2018).
\newblock Hyper-local geographically weighted regression: extending GWR through local model selection and local bandwidth optimization. 
\newblock {\em Journal of Spatial Information Science}, 17: 63--84.

\bibitem[Fotheringham et~al., 2015]{FCY2015}
Fotheringham, A. S., Crespo, R., Yao, J. (2015).
\newblock  Geographical and temporal weighted regression (GTWR).
\newblock {\em Geographical Analysis}, 47: 431--452.

\bibitem[Fotheringham et~al., 2017]{FYK2017}
Fotheringham, A. S., Yang, W., Kang, W. (2017).
\newblock  Multiscale geographically weighted regression (MGWR).
\newblock {\em Annals of the American Association of Geographers}, 107: 1247--1265.

\bibitem[Guo and Zhang, 2024]{GZ2024}
Guo, C., Zhang, W. (2024).
\newblock  Model-averaging-based semiparametric modeling for conditional quantile prediction.
\newblock {\em Science China Mathematics}, 67: 2843--2872.

\bibitem[Hansen, 2007]{H2007} 
Hansen B E. (2007).
\newblock Least squares model averaging. 
\newblock {\em Econometrica}, 75: 1175--1189.

\bibitem[Hansen and Racine, 2012]{HR2012}
Hansen, B.E., Racine, J.S. (2012).
\newblock  Jackknife model averaging.
\newblock {\em Journal of Econometrics}, 167: 38--46.

\bibitem[Hoeting et~al., 2006]{HDM06}
Hoeting J A, Davis R A, Merton A A, et al. (2006).
\newblock  Model selection for geostatistical models.
\newblock {\em Ecological Applications}, 16: 87--98.

\bibitem[Ji et~al., 2025]{JLL2025}
Ji A., Li J., Li Q. (2025).
\newblock Model averaging for spatial autoregressive panel data models.
\newblock {\em Spatial Statistics}, 70: 100931.

\bibitem[Jiang et~al., 2025]{JLLC2025}
Jiang, B., Lv, J., Li, J., Cheng, M.Y. (2025).
\newblock  Robust model averaging prediction of longitudinal response with ultrahigh-dimensional covariates.
\newblock {\em Journal of the Royal Statistical Society, Series B}, 87: 337--361.

\bibitem[Kim and Wang, 2021]{KW21}
Kim, M., Wang, L. (2021).
\newblock  Generalized spatially varying coefficient models.
\newblock {\em Journal of Computational and Graphical Statistics}, 30: 1--10.

\bibitem[Kim et~al., 2025]{KWW25}
Kim, M., Wang, L., Wang, H.J. (2025).
\newblock  Estimation and inference of quantile spatially
varying coefficient models over complicated domains.
\newblock {\em Journal of the American Statistical Association}, 120: 1853--1867.

\bibitem[Liu and Li, 2024]{LL2024}
Liu, P., Li, J. (2024).
\newblock  Segment regression model average with multiple threshold variables
and multiple structural breaks.
\newblock {\em The Canadian Journal of Statistics}, 52: 131--161.

\bibitem[Li et~al., 2022]{LLWL22}
Li, J., Lv, J., Wan, A., Liao, J. (2022).
\newblock  Adaboost semiparametric model averaging prediction for multiple
categories.
\newblock {\em Journal of the American Statistical Association}, 117: 495--509.

\bibitem[Liao et~al., 2019]{LZG2019}
Liao, J., Zou, G., Gao, Y. (2019).
\newblock  Spatial Mallows model averaging for geostatistical models.
\newblock {\em The Canadian Journal of Statistics }, 47: 336--351.

\bibitem[Liao et~al., 2019]{LZZZ19}
Liao, J., Zong, X., Zhang, X., Zou, G. (2019).
\newblock  Model averaging based on leave-subject-out cross-validation for vector autoregressions.
\newblock {\em Journal of Econometrics}, 209: 35--60.

\bibitem[Miao et~al., 2025]{MFZW2025}
Miao, X., Fang, F., Zhu, X., Wang, H. (2025).
\newblock  Spatial weights matrix selection and model
averaging for multivariate spatial autoregressive
models .
\newblock {\em Econometric Reviews}, in press, https://doi.org/10.1080/07474938.2025.2560637.

\bibitem[Mu et~al., 2018]{MWW18}
Mu, J., Wang, J., Wang, L. (2018).
\newblock  Estimation and inference in spatially varying coefficient models.
\newblock {\em Environmetrics}, 29: e2485.

\bibitem[Mei et~al., 2016]{MXW2016}
Mei, C. L., Xu, M., Wang, N. (2016).
\newblock  A bootstrap test for constant coefficients in geographically weighted regression models.
\newblock {\em International Journal of Geographical Information Science}, 30: 1622--1643.

\bibitem[Seng and Li, 2022]{seng2022}
Seng, L., Li, J. (2022).
\newblock  Structural equation model averaging: methodology and application.
\newblock {\em Journal of Business and Economic Statistics}, 40: 815-828.

\bibitem[Shi et~al., 2024]{SZZ2024}
Shi, P., Zhang, X., Zhong, W. (2024).
\newblock  Estimating conditional average treatment effects with heteroscedasticity by model averaging and matching.
\newblock {\em Economics Letters}, 238, 111679.

\bibitem[Tobler, 1970]{Tobler}
Tobler, W.R. (1970).
\newblock  A computer movie simulating urban growth in the detroit region.
\newblock {\em Economic geography}, 46: 234--240.

\bibitem[Whittle, 1960]{Whittle} 
Whittle P. (1960).
\newblock Bounds for the moments of linear and quadratic forms in independent variables. 
\newblock {\em Theory of Probability \& Its Applications}, 5: 302--305.

\bibitem[Wang et~al., 2014]{WFMWQ2014}
Wang, S., Fang, C., Ma, H., Wang, Y., Qin, J. (2014).
\newblock  Spatial differences and multi-mechanism of carbon footprint based on GWR model in provincial China.
\newblock {\em Journal of Geographical Sciences}, 24: 612--630.

\bibitem[Wang et~al., 2009]{WZH2009}
Wang, Y., Zhang, X., Huang, C. (2009).
\newblock  Spatial variability of soil total nitrogen and soil total phosphorus under different land uses in a small watershed on the Loess Plateau, China.
\newblock {\em Geoderma}, 150, 141--149.

\bibitem[Wan et~al., 2010]{WZZ2010}
Wan, A. T., Zhang, X., Zou, G. (2010).
\newblock  Least squares model averaging by Mallows
criterion.
\newblock {\em Journal of Econometrics}, 156: 277--283.

\bibitem[Yang et~al., 2022]{YDT2022}
Yang, Y., Dogan, O., Taspinar, S. (2022).
\newblock  Model selection and model averaging for matrix exponential spatial models.
\newblock {\em Econometric Reviews}, 41: 827--858.

\bibitem[Yu et~al., 2025]{YWW25}
Yu, S., Wang, G., Wang, L. (2025).
\newblock  Distributed heterogeneity learning for generalized partially linear models with spatially varying coefficients.
\newblock {\em Journal of the American Statistical Association}, 120: 779--793.

\bibitem[Yu et~al., 2022]{YWWG2022}
Yu, S., Wang, Y., Wang, L., Gao, L. (2022).
\newblock  Spatiotemporal autoregressive partially linear varying coefficient models.
\newblock {\em Statistica Sinica}, 32: 2119--2146.

\bibitem[Zhu et~al., 2019]{ZWZZ19}
Zhu, R., Wan, A. T. K., Zhang, X., Zou, G. (2019).
\newblock  A Mallows-type model averaging estimator for the varying-coefficient partially linear model.
\newblock {\em Journal of the American Statistical Association}, 114: 882--892.

\bibitem[Zhang and Yu, 2018]{ZY2018}
Zhang, X., Yu, J. (2018).
\newblock  Spatial weights matrix selection and model averaging for spatial autoregressive models.
\newblock {\em Journal of Econometrics}, 203: 1--18.	

\bibitem[Zhang et~al., 2020]{ZZLC20}
Zhang, X., Zou, G., Liang, H., Carroll, R.J. (2020).
\newblock  Parsimonious model averaging with a diverging number of parameters.
\newblock {\em Journal of the American Statistical Association}, 115: 972--984.

\end{thebibliography}

\end{document}